\documentclass[a4paper,12pt,twoside]{article}

\usepackage{amsfonts}
\usepackage{amsmath}
\usepackage{amsopn}
\usepackage{amsthm}
\usepackage{amssymb}
\usepackage{graphicx}
\usepackage{rotating}
\usepackage[]{natbib}
\usepackage[english]{babel}
\usepackage[utf8]{inputenc}
\usepackage[colorlinks=true,citecolor=blue,linkcolor=blue]{hyperref}
\usepackage{xspace}
\usepackage{float}
\usepackage{setspace}
\usepackage{dsfont}
\usepackage{booktabs}
\usepackage[T1]{fontenc}
\usepackage{anysize}
\usepackage{wrapfig}
\usepackage[labelfont=bf,labelsep=period]{caption}
\floatstyle{plaintop}
\restylefloat{table}
\marginsize{1.87cm}{1.87cm}{1.87cm}{1.87cm} % left, right, top , bottom
\usepackage{fancyhdr}
\usepackage{algorithm}% http://ctan.org/pkg/algorithms
\usepackage[noend]{algpseudocode}% http://ctan.org/pkg/algorithmicx

\RequirePackage{txfonts}
\usepackage{times}

\newtheorem{Proposition}{Proposition}

\newtheorem{Theorem}{Theorem}

\newtheorem{Assumption}{Assumption}

\allowdisplaybreaks

\begin{document}

\def\figureautorefname{Figure}
\def\tableautorefname{Table}
\def\algorithmautorefname{Algorithm}
\def\sectionautorefname{Section}
\def\subsectionautorefname{Section}
\def\subsubsectionautorefname{Section}
\def\Propositionautorefname{Proposition}
\def\Theoremautorefname{Theorem}
\def\Lemmaautorefname{Lemma}
\def\Assumptionautorefname{Assumption}
\renewcommand*\footnoterule{}

\title{Informed reversible jump algorithms}

\author{Philippe Gagnon $^{1}$}

\maketitle

\thispagestyle{empty}

\noindent $^{1}$Department of Mathematics and Statistics, Universit\'{e} de Montr\'{e}al, Canada.

\begin{abstract}
  Incorporating information about the target distribution in proposal mechanisms
generally produces efficient Markov chain Monte Carlo algorithms (or at least, algorithms that are more efficient than uninformed counterparts). % Hamiltonian Monte Carlo represents a successful example of fixed-dimensio\-nal
%algorithms incorporating gradient information.
 For instance, it
has proved successful to incorporate gradient information in fixed-dimensio\-nal
algorithms, as seen with algorithms such as Hamiltonian Monte Carlo. 
In trans-dimen\-sional algorithms,
\cite{green2003trans} recommended to sample the parameter proposals during
model switches from normal distributions with informative means and covariance
matrices. These proposal distributions can be viewed as asymptotic approximations
to the parameter distributions, where the limit is with regard to the sample
size. Models are typically proposed using uninformed uniform distributions.
In this paper, we build on the approach of \cite{zanella2019informed} for
discrete spaces to incorporate information about neighbouring models. We
rely on approximations to posterior model probabilities that are asymptotically
exact. We prove that, in some scenarios, the samplers combining this approach
with that of \cite{green2003trans} behave like ideal ones that use the
exact model probabilities and sample from the correct parameter distributions,
in the large-sample regime. We show that the implementation of the proposed
samplers is straightforward in some cases. The methodology is applied to
a real-data example. The code is available online.\footnote{See ancillary files on \href{https://arxiv.org/abs/1911.02089}{arXiv:1911.02089}.}
\end{abstract}


\section{Introduction}\label{sec_intro}

\subsection{Reversible jump algorithms}\label{sec_RJ}

Reversible jump (RJ, \cite{green1995reversible}) algorithms are Markov
chain Monte Carlo (MCMC) methods that one uses to sample from a target
distribution $\pi (\, \cdot \,, \cdot \mid \mathbf{D}_{n})$ defined on
a union of sets
$\bigcup _{k \in \mathcal{K}} \{k\} \times \mathbb{R} ^{d_{k}}$,
$\mathcal{K}$ being a countable set and $d_{k}$ positive integers. This
distribution corresponds in Bayesian statistics to a joint posterior distribution
of a model indicator $K\in \mathcal{K}$ and the parameters of Model
$K$, denoted by $\mathbf{X}_{K}\in \mathbb{R} ^{d_{K}}$, given
$\mathbf{D}_{n}$, a data sample of size $n$. Such a posterior distribution
allows to jointly infer about $(K, \mathbf{X}_{K})$, or in other words,
to simultaneously achieve model selection/averaging
\citep{hoeting1999bayesian} and parameter estimation. In the following,
we assume for simplicity that the parameters of all models are continuous
random variables. Again for simplicity, we will abuse notation by also
using $\pi (\, \cdot \,, \cdot \mid \mathbf{D}_{n})$ to denote the joint
posterior density with respect to a product of the counting and Lebesgue
measures. We will use $\pi (\, \cdot \, \mid \mathbf{D}_{n})$ and
$\pi (\, \cdot \, \mid k, \mathbf{D}_{n})$ to denote the marginal posterior
probability mass function (PMF) of $K$ and conditional posterior distribution/density
of $\mathbf{X}_{K}$ given that $K = k$, respectively. Note that in this
paper, we assume that $\mathcal{K}$ is defined such that
$\pi (k \mid \mathbf{D}_{n}) > 0$ for all $k$.

The variable $K \in \mathcal{K}$ can take different forms in practice.
In mixture modelling \citep{richardson1997bayesian}, $K$ is the number
of components and $\mathcal{K}$ is a subset of the positive integers and
defines a sequence of nested models, i.e.\ Model 1 is nested in Model 2
which is nested in Model 3, and so on. Often there is no such ``ordering''
between the models. This is for instance the case in variable selection.
In this framework, $K$ is a vector of $0$'s and $1$'s indicating which
covariates are included in a model, and for instance, Model
$k_{0} = (1, 1, 1, 0, 0, 1, 0, \ldots , 0)$ is the model with covariates
1, 2, 3 and 6 (as in Figure~\ref{fig_1} below); $\mathcal{K}$ is thus a collection
of such vectors. A RJ algorithm explores both the model space and the parameter
spaces. We focus in this paper on the case where $\mathcal{K}$
\emph{does not} define a sequence of nested models. When
$\mathcal{K}$ \emph{defines} a sequence of nested models, $K$ can be recoded
as an ordinal discrete variable as above and this case has been recently
studied by \cite{gagnon2019NRJ} who proposed efficient non-reversible trans-dimensional
samplers.

The proposal mechanism in a RJ algorithm can be outlined as follows: given
a current state of the Markov chain $(k, \mathbf{x}_{k})$, the algorithm
first proposes a model, say Model $k'$, by using a PMF
$g(k, \cdot \,)$ and then parameter values for this model by applying a
diffeomorphism $\mathcal{D}_{k\mapsto k'}$ to $\mathbf{x}_{k}$ and auxiliary
variables $\mathbf{u}_{k\mapsto k'} \sim q_{k\mapsto k'}$, yielding the
parameter proposal $\mathbf{y}_{k'}$ and another vector of auxiliary variables
$\mathbf{u}_{k' \mapsto k}$. The proposal is accepted, meaning that the
next state of the Markov chain is $(k', \mathbf{y}_{k'})$, with probability
(assuming that the current state has positive density under the target):
\begin{align}
\label{eqn_acc_prob_RJ}
&\alpha _{\text{RJ}}((k,\mathbf{x}_{k}),(k',\mathbf{y}_{k'})) := 1 \wedge
\frac{\pi (k', \mathbf{y}_{k'} \mid \mathbf{D}_{n}) \, g(k',k) \, q_{k'\mapsto k}(\mathbf{u}_{k'\mapsto k})}{\pi (k, \mathbf{x}_{k} \mid \mathbf{D}_{n}) \, g(k,k') \, q_{k\mapsto k'}(\mathbf{u}_{k\mapsto k'}) \, |J_{\mathcal{D}_{k\mapsto k'}}(\mathbf{x}_{k}, \mathbf{u}_{k\mapsto k'})|^{-1}},
\end{align}
where $x \wedge y := \min (x, y)$ and
$|J_{\mathcal{D}_{k\mapsto k'}}(\mathbf{x}_{k}, \mathbf{u}_{k\mapsto k'})|$
is the absolute value of the determinant of the Jacobian matrix of the
function $\mathcal{D}_{k\mapsto k'}$. If the proposal is rejected, the
chain remains at the same state $(k,\mathbf{x}_{k})$. Recall that a diffeomorphism
is a differentiable map having a differentiable inverse and note that we
abused notation by using $q_{k\mapsto k'}$ to denote both the distribution
and the probability density function (PDF) of
$\mathbf{u}_{k\mapsto k'}$. The notation $k \mapsto k'$ in subscript is
used to highlight a dependence on the model transition that is proposed,
which is from Model $k$ to Model $k'$. In the particular case where
$k' = k$, we say that a \textit{parameter update} is proposed; otherwise,
we say that a \textit{model switch} is proposed. Model $k'$ is reachable
from Model $k$ if it belongs to the support of $g(k, \cdot \,)$ which is
considered in this paper to be the neighbourhood of Model $k$, denoted
by $\mathbf{N}(k)$.

In trans-dimensional samplers, the neighbourhoods are usually formed of
the ``closest'' models. The notion of ``proximity'' is often natural, making
these closest models straightforward to identify. For instance, in Section~\ref{sec_application} we apply the proposed methodology in a variable-selection
example and the closest models are those with an additional and one less
covariates, to which the current model is added. In this paper, we consider
that Model $k$ belongs to the support of $g(k, \cdot \,)$ for all
$k$; this will be seen to be useful when the posterior model PMF concentrates.

Algorithm~\ref{algo_RJ} presents a general RJ.

%a1 #&#
\begin{algorithm}[ht]
\caption{RJ}
\label{algo_RJ}
\begin{enumerate}
 \itemsep 0mm

  \item Sample $k' \sim g(k, \cdot \,)$, $\mathbf{u}_{k\mapsto k'} \sim q_{k\mapsto k'}$ and $u \sim \mathcal{U}[0, 1]$.

  \item Apply $\mathcal{D}_{k\mapsto k'}(\mathbf{x}_{k}, \mathbf{u}_{k\mapsto k'}) = (\mathbf{y}_{k'}, \mathbf{u}_{k' \mapsto k})$.

  \item If $u \leq \alpha _{\text{RJ}}((k,\mathbf{x}_{k}),(k',\mathbf{y}_{k'}))$, set the next state of the chain to $(k', \mathbf{y}_{k'})$. Otherwise, set it to $(k,\mathbf{x}_{k})$.

  \item Go to Step 1.
 \end{enumerate}
\end{algorithm}

Looping over the steps described in Algorithm~\ref{algo_RJ} produces Markov chains
that are reversible with respect to the target distribution
$\pi (\, \cdot \,, \cdot \mid \mathbf{D}_{n})$. If the chains are in addition
irreducible and aperiodic, then they are ergodic
\citep{tierney1994markov}, which guarantees, among others, that the law
of large numbers holds, implying that ergodic averages are approximations
to expectations under $\pi (\, \cdot \,, \cdot \mid \mathbf{D}_{n})$.

\subsection{Problems with RJ}\label{sec_problem}

Algorithm~\ref{algo_RJ} takes as inputs: an initial state
$(k,\mathbf{x}_{k})$, a total number of iterations, a PMF
$g(k, \cdot \,)$ for each Model $k$, and functions $q_{k\mapsto k'}$ and
$\mathcal{D}_{k\mapsto k'}$ for each pair $(k, k')$ with Model $k'$ being
reachable from Model $k$. Implementing RJ is well known for being a difficult
task considering the large number of functions that need to be specified
and the often lack of intuition about how one should achieve their specification,
the latter being especially true for functions involved in model switches.
Significant amount of work has been carried out to address the specification
of the functions $q_{k\mapsto k'}$ and $\mathcal{D}_{k\mapsto k'}$ involved
in model switches in a principled and informed way; see, e.g.,
\cite{green2003trans} and \cite{brooks2003efficient}. The approaches of
these authors are arguably the most popular. Their objective is the following:
given $\mathbf{x}_{k} \sim \pi (\, \cdot \mid k,\mathbf{D}_{n})$, we want
to identify $q_{k\mapsto k'}, q_{k'\mapsto k}$ and
$\mathcal{D}_{k\mapsto k'}$ such that applying the transformation
$\mathcal{D}_{k\mapsto k'}$ to
$(\mathbf{x}_{k}, \mathbf{u}_{k\mapsto k'}) \sim \pi (\, \cdot \mid k,
\mathbf{D}_{n})\otimes q_{k\mapsto k'}$ leads to
$(\mathbf{y}_{k'},\mathbf{u}_{k'\mapsto k}) \sim \pi (\, \cdot \mid k',
\mathbf{D}_{n})\otimes q_{k'\mapsto k}$ (at least approximatively). They
essentially look for a way to sample from the conditional distributions
$\pi (\, \cdot \mid k',\mathbf{D}_{n})$, in this constrained framework.
This in turn aims at increasing the acceptance probability
$\alpha _{\text{RJ}}$ defined in \eqref{eqn_acc_prob_RJ} towards
%
%e2 #&#
\begin{align}
\label{acc_prob_marginal}
\alpha _{\text{MH}}(k, k'):=1\wedge
\frac{\pi (k'\mid \mathbf{D}_{n}) \, g(k',k)}{\pi (k \mid \mathbf{D}_{n}) \, g(k,k')},
\end{align}
which corresponds to the acceptance probability in a Metropolis--Hastings
(MH, \cite{metropolis1953equation} and \cite{hastings1970monte}) algorithm
targeting the PMF $\pi ( \, \cdot \mid \mathbf{D}_{n})$. The approach of
\cite{green2003trans}, for instance, proceeds as if the conditional distributions
$\pi (\, \cdot \mid k',\mathbf{D}_{n})$ were normal.

Notwithstanding the merit of this objective, it is to be noticed that,
even when it is achieved, ``half'' of the work for maximizing
$\alpha _{\text{RJ}}$ is done because, as shown in Figure~\ref{fig_1}, poor
models may be proposed more often than they should be if $g$ is not well
designed, leading to smaller acceptance rates. Note that when considering
candidate proposal distributions $g(k, \cdot \,)$ having all the same support
$\mathbf{N}(k)$, choosing the one that maximizes the acceptance probability
$\alpha _{\text{RJ}}$ represents a first step towards proving that it is
the best (or at least one of the best) given that they all allow to reach
the same models; this choice is indeed expected to improve mixing.

%f1 #&#
\begin{figure}[ht]
  $\begin{array}{ccc}
  \hspace{-3mm}\includegraphics[width = 0.345\textwidth, trim = {0 0 20mm 0}, clip]{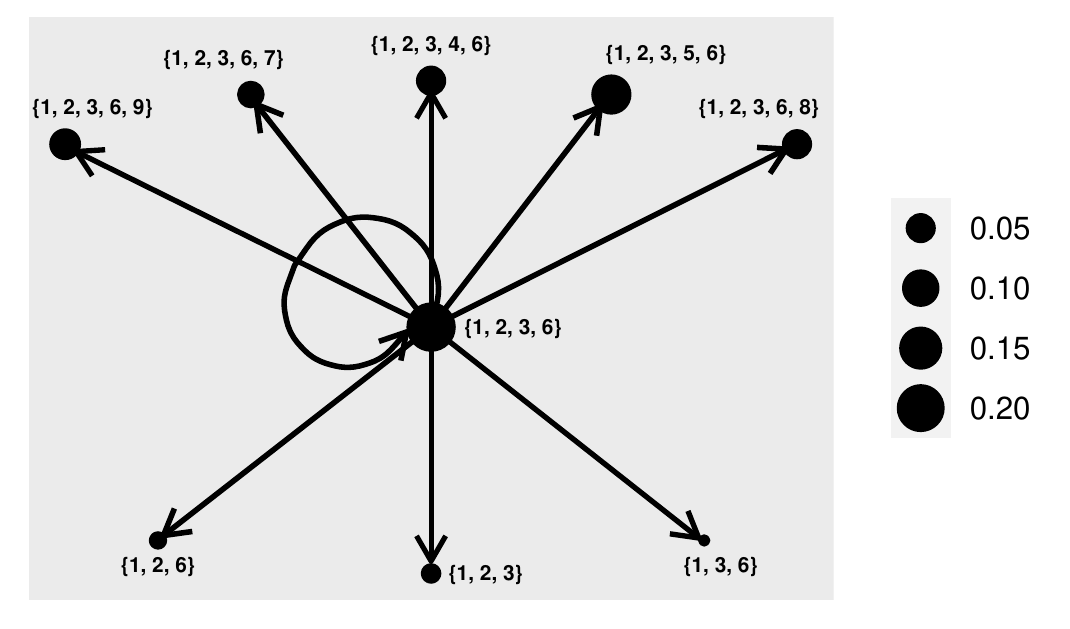} & \hspace{-6mm}\includegraphics[width = 0.345\textwidth, trim = {0 0 20mm 0}, clip]{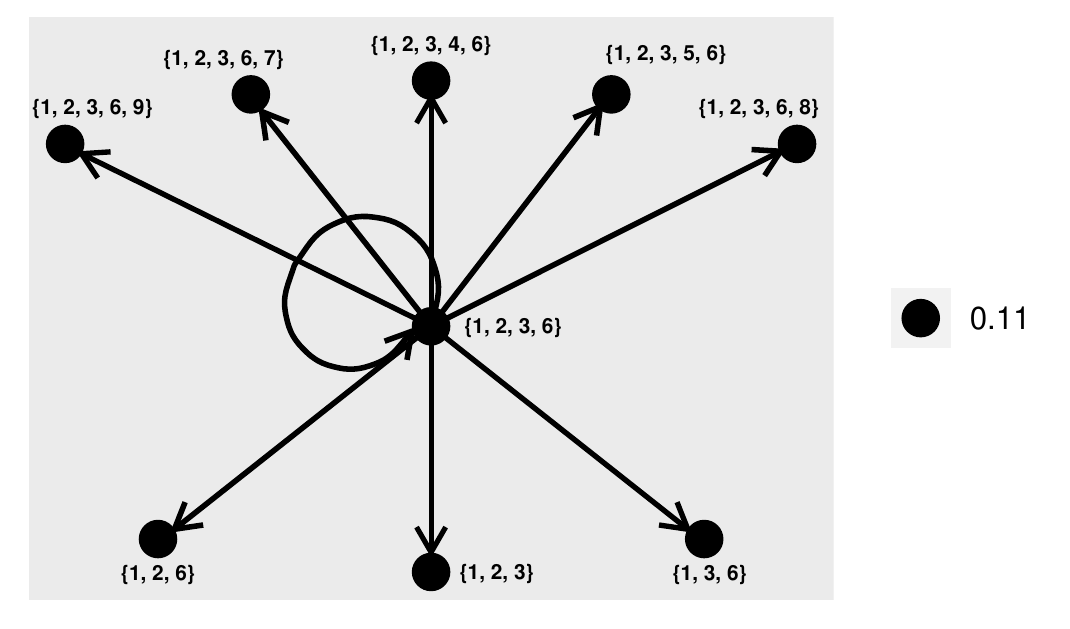} & \hspace{-6mm}\includegraphics[width = 0.345\textwidth, trim = {0 0 20mm 0}, clip]{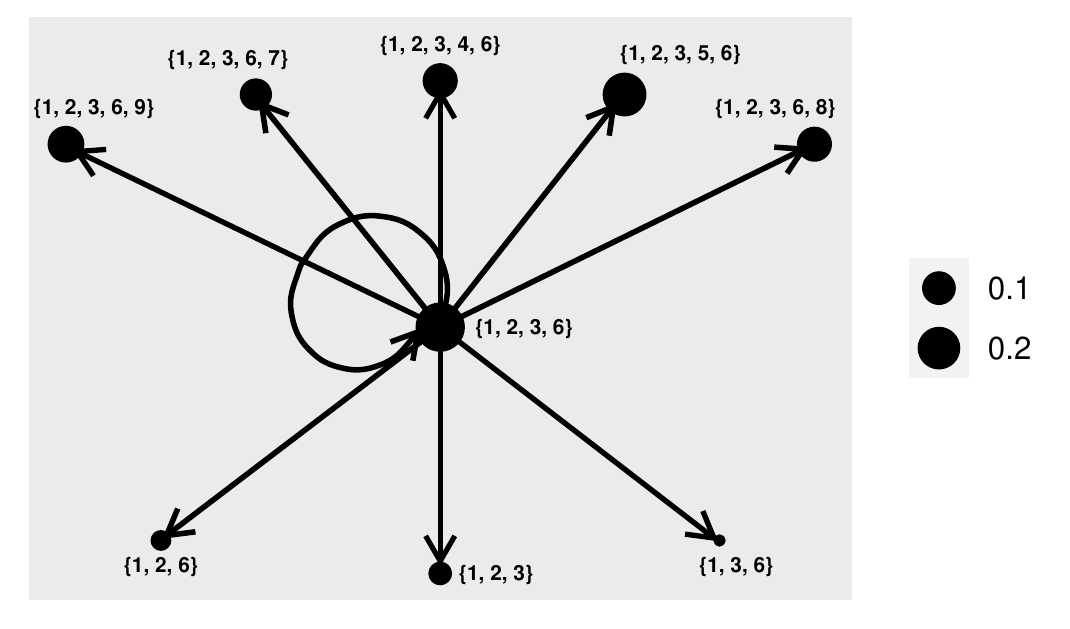} \cr
  \hspace{-4mm}\textbf{\scriptsize (a) Posterior probabilities} & \hspace{-7mm} \textbf{\scriptsize (b) Uninformed sampler} & \hspace{-7mm} \textbf{\scriptsize (c) Informed sampler}
  \end{array}$
\caption{(a) posterior probabilities (represented by the size
of the dots) of models forming the neighbourhood of the model with covariates
$\{1, 2, 3, 6\}$, in a linear-regression analysis of the prostate-cancer
data of \cite{stamey1989prostate}; the models with an additional covariate
can be found in the upper part, whereas those with one less variable are
in the lower part; (b) probabilities of proposing the models forming the
neighbourhood (represented by the size of the dots) when the current model
is that with covariates $\{1, 2, 3, 6\}$ and $g$ is the uniform over the
neighbourhood; (c) probabilities of proposing the models forming the neighbourhood
(represented by the size of the dots) when the current model is that with
covariates $\{1, 2, 3, 6\}$ and $g$ is instead an informed version (as
defined in \protect Section~\ref{sec_objectives})}
\label{fig_1}
\end{figure}

\subsection{Objectives}\label{sec_objectives}

The specification of $g$ has been overlooked; this PMF is indeed typically
set to a uniform distribution as in Figure~\ref{fig_1}~(b). The first objective
of this paper is to propose methodology to incorporate information about
neighbouring models in its design to make fully-informed RJ available.
To achieve this we rely on a generic technique recently introduced by
\cite{zanella2019informed} that is used to construct informed samplers
for discrete state-spaces. Let us assume for a moment that we have access
to the unnormalized posterior model probabilities and that we use a MH
algorithm to sample from $\pi (\, \cdot \mid \mathbf{D}_{n})$; the technique
consists in constructing what the author calls
\textit{locally-balanced} proposal distributions:
%
%e3 #&#
\begin{align}
\label{eqn_g}
g(k, k') \propto h\left (
\frac{\pi (k' \mid \mathbf{D}_{n})}{\pi (k \mid \mathbf{D}_{n})}
\right ) \, \mathds{1}(k' \in \mathbf{N}(k)),
\end{align}
where $h$ is a continuous positive function such that $h(0) = 0$ and
$h(x) \big / h(1/x) = x$ for all positive $x$ (the square root satisfies
this condition for instance), and $\mathds{1}$ is the indicator function.
Such a function $h$ leads to an acceptance probability in the MH sampler
given by $\alpha _{\text{MH}}(k, k') = 1 \wedge c(k)/c(k')$, where
$c(k)$ is the normalizing constant of $g(k, \cdot \,)$. A motivation for
using this technique is that, in the limit, when the ambient space becomes
larger and larger (but the neighbourhoods have a fixed size), there is
no need for an accept-reject step anymore; the proposal distributions leave
the distribution $\pi (\, \cdot \mid \mathbf{D}_{n})$ invariant.
\cite{zanella2019informed} indeed proves that
$c(k)/c(k') \longrightarrow 1$ for all possible pairs $(k, k')$ under some
conditions. These properties suggest that locally-balanced samplers are
efficient, at least in high dimensions. The author in fact empirically
shows that they perform better than alternative solutions to sample from
PMFs in some practical examples, with a highly marked difference in the
examples with high-dimensional spaces.

The first obstacle to achieving our first objective is that we typically
do not have direct access to the model probabilities, because they involve
integrals over the parameter spaces. Drawing inspiration from the approach of \cite{green2003trans}
that can be viewed as using asymptotic approximations
to $\pi (\, \cdot \mid k',\mathbf{D}_{n})$ to design RJ, where the limit
is with regard to the sample size, we propose to use approximations to
the unnormalized version of the model probabilities
$\pi (k \mid \mathbf{D}_{n})$ whose accuracy increases with $n$. We prove
that, in some scenarios, RJ using both approximations behave asymptotically
as RJ that set
$g(k, k') \propto h\left (
\frac{\pi (k' \mid \mathbf{D}_{n})}{\pi (k \mid \mathbf{D}_{n})}
\right ) \, \mathds{1}(k' \in \mathbf{N}(k))$ and that sample parameters
from the correct conditional distributions
$\pi (\, \cdot \mid k',\mathbf{D}_{n})$. These ideal RJ have acceptance
probabilities equal to $\alpha _{\text{MH}}$. All this suggests that the
resulting RJ with asymptotically locally-balanced proposal distributions
are efficient, at least when $n$ is large enough. We show that it is the
case even in a moderate-size real-data example. In this example, the PMF
$\pi ( \, \cdot \mid \mathbf{D}_{n})$ concentrates on several models. We
also provide evidences that the proposed RJ are efficient when $n$ is large
and the PMF $\pi ( \, \cdot \mid \mathbf{D}_{n})$ is highly concentrated
so that the mass is non-negligible only for a handful of models.

The approximations on which the proposed methodology is based are more
accurate when all parameters of all models take values on the real line.
We thus recommend to apply a transformation to parameters for which this
is not the case. Such a practice is often beneficial for MCMC methods.
It is in fact required when using Hamiltonian Monte Carlo (HMC, see, e.g.,
\cite{neal2011mcmc}), which is employed to perform parameter updates within
RJ in the real-data application of Section~\ref{sec_application}.

The second objective of this paper is to make clear how each function required
for implementation should be specified, at least in some cases, allowing
a fully-automated procedure. These cases are those where each of the log
conditional densities $\log \pi (\, \cdot \mid k, \mathbf{D}_{n})$ has
a well-defined mode and second derivatives that exist and that are continuous.
The procedure can be executed even if the model space is large or infinite,
as long as the posterior model PMF concentrates on a reasonable number
of models, which is expected in practice.

\subsection{Organization of the paper}

In Section~\ref{sec_RJ_and_input}, we address the specification of the inputs
required for the implementation of the proposed samplers; in particular,
we discuss the design of the PMFs $g(k, \cdot \,)$, so that both objectives
of the paper are achieved. We next present in Section~\ref{sec_asymptotics} a theoretical result about the asymptotic behaviour
of the proposed RJ, and an analysis of the limiting RJ. In Section~\ref{sec_application}, the methodology is evaluated in a real-data
variable-selection example. The paper finishes in Section~\ref{sec_discussion} with retrospective comments and possible directions
for future research.

The approximations on which the design of the functions
$g(k, \cdot \,)$, $q_{k\mapsto k'}$ and $\mathcal{D}_{k\mapsto k'}$ rely
can be inaccurate when the sample size is not sufficiently large. There
exist methods that allow to build upon the approximations used to design
$q_{k\mapsto k'}$ and $\mathcal{D}_{k\mapsto k'}$ to improve the parameter-proposal
mechanism. The methods of \cite{karagiannis2013annealed} and
\cite{andrieu2018utility} are examples of such methods. An overview of
these methods is provided in Section~\ref{sec_RJ_and_input}, while the details
are presented in Appendix~\ref{sec_improve_approx}. In Section~\ref{sec_RJ_and_input}, it is explained that these methods are also
useful for improving the model-proposal mechanism. In fact, as the precision
parameters of these methods increase without bounds, the resulting samplers
converge towards ideal ones that have access to the unnormalized model
probabilities to design $g$ and that sample parameters from the correct
conditional distributions $\pi (\, \cdot \mid k',\mathbf{D}_{n})$,
\emph{for fixed $n$} (see Appendix~\ref{sec_improve_approx} for the details).
The methods of \cite{karagiannis2013annealed} and
\cite{andrieu2018utility} are sophisticated and come at a computational
cost (especially that of \cite{karagiannis2013annealed}), which may not
be worth it when $n$ is large enough for the problem at hand. This is the
case in our variable-selection example, highlighting that the simple RJ
proposed in this paper that use simple approximations are efficient in
certain practical situations. The methods of
\cite{karagiannis2013annealed} and \cite{andrieu2018utility} are nevertheless
useful in the variable-selection example as they allow to show which of
the approximations used in the design of $g(k, \cdot \,)$,
$q_{k\mapsto k'}$ and $\mathcal{D}_{k\mapsto k'}$ explain the gap between
the proposed RJ and the ideal ones. Their computational cost is offset
by a significant enough improvement in, for instance, the change-point
problem presented in \cite{green1995reversible} (as shown in
\cite{karagiannis2013annealed} and \cite{andrieu2018utility}).

Some details of the variable-selection example are presented in Appendix~\ref{sec_supp_variable_selection}. All proofs of theoretical results
are provided in Appendix~\ref{sec_proofs}.

\section{Input design and proposed RJ}\label{sec_RJ_and_input}

We start in Section~\ref{sec_model_proposal_design} with the proposed design
of $g$. We next turn in Section~\ref{sec_parameter_proposal_design} to the
proposed design of the functions $q_{k\mapsto k'}$ and
$\mathcal{D}_{k\mapsto k'}$, which will be seen to represent a simplified
version of that in \cite{green2003trans}. We present in Section~\ref{sec_resulting_RJ} the resulting RJ. In Section~\ref{sec_overview_methods_improve}, we present an overview of the methods
of \cite{karagiannis2013annealed} and \cite{andrieu2018utility} to improve
upon the approximations used in the design of the proposal mechanism when
these approximations are not accurate enough.

\subsection{Specification of the model proposal distributions}\label{sec_model_proposal_design}

The design of $g$ starts with the definition of neighbourhoods around all
models which specify the support of $g(k,\cdot \,)$ for all $k$. As mentioned
in Section~\ref{sec_RJ}, it is typically possible to define these neighbourhoods
in a natural way. It will thus be considered that the neighbourhoods are
given and that we can build on these to address the specification of
$g$.

In practice, $g(k,\cdot \,)$ is commonly set to the uniform distribution
over $\mathbf{N}(k)$: $g(k, k') = 1/|\mathbf{N}(k)|$ for all
$k'\in \mathbf{N}(k)$, where $|\mathbf{N}(k)|$ is the cardinality of
$\mathbf{N}(k)$. In the simple situation where
$\mathbf{u}_{k \mapsto k'} = \mathbf{y}_{k'} \sim \pi (\, \cdot \mid k',
\mathbf{D}_{n})$ and $|\mathbf{N}(k)| = |\mathbf{N}(k')|$ for all possible
pairs $(k, k')$ (the latter is true in variable selection with neighbourhoods
defined as in Section~\ref{sec_RJ}), the RJ acceptance probability reduces
to
\begin{equation*}
\alpha _{\text{RJ}}((k,\mathbf{x}_{k}),(k',\mathbf{y}_{k'})) = 1
\wedge
\frac{\pi (k'\mid \mathbf{D}_{n})}{\pi (k \mid \mathbf{D}_{n})}.
\end{equation*}
It is thus easily seen that a chain may get stuck for several iterations
when there are a lot of poor models with
$\pi (k'\mid \mathbf{D}_{n}) \ll \pi (k \mid \mathbf{D}_{n})$ in the neighbourhood;
this is especially true in high dimensions because the number of poor models
may be very high. Our goal is to include local neighbourhood information
in the PMFs $g(k, \cdot \,)$ to skew the latter towards high probability
models.

We propose here to follow the strategy of \cite{zanella2019informed}, implying
that ideally we would set
$g(k, k') \propto h\left (\pi (k' \mid \mathbf{D}_{n}) \big / \pi (k
\mid \mathbf{D}_{n})\right )$ for all $k'\in \mathbf{N}(k)$ (recall \eqref{eqn_g}), where for all $k$,
\begin{equation*}
\pi (k \mid \mathbf{D}_{n}) = \int \pi (k, \mathbf{x}_{k} \mid
\mathbf{D}_{n}) \, \mathrm{d}\mathbf{x}_{k}.
\end{equation*}
These integrals are typically intractable. We propose to approximate them
and for this we assume that each log conditional density
$\log \pi (\, \cdot \mid k, \mathbf{D}_{n})$ has a well-defined mode and
that its second derivatives exist and are continuous. In fact, the integrals
that we propose to approximate are
\begin{equation*}
\int \pi _{\text{un.}}(\mathbf{x}_{k} \mid k, \mathbf{D}_{n}) \,
\mathrm{d}\mathbf{x}_{k},
\end{equation*}
where
$\pi _{\text{un.}}(\, \cdot \mid k, \mathbf{D}_{n}) := \mathcal{L}(\,
\cdot \mid k, \mathbf{D}_{n}) \, \pi (\, \cdot \mid k)$ is the unnormalized
posterior density under Model $k$ given by the product of the likelihood
function and prior parameter density under Model $k$, and we use that
\begin{equation*}
\pi (k \mid \mathbf{D}_{n}) \propto \pi (k) \int \pi _{\text{un.}}(
\mathbf{x}_{k} \mid k, \mathbf{D}_{n}) \, \mathrm{d}\mathbf{x}_{k},
\end{equation*}
$\pi (k)$ being the prior probability assigned to Model $k$. The approximation
that we propose to apply is the Laplace approximation, meaning that we
write
$\pi _{\text{un.}}(\, \cdot \mid k, \mathbf{D}_{n}) = \exp (\log \pi _{
\text{un.}}(\, \cdot \mid k, \mathbf{D}_{n}))$ and
$\log \pi _{\text{un.}}(\mathbf{x}_{k} \mid k, \mathbf{D}_{n})$ is developed
using a Taylor series expansion (see, e.g.,
\cite{davison1986approximate}). This yields:
%
%e4 #&#
\begin{align}
\label{eqn_pi_hat}
\hat{\pi }(k \mid \mathbf{D}_{n}) \propto \pi (k) \, (2\pi )^{d_{k}/2}
\, \pi _{\text{un.}}(\hat{\mathbf{x}}_{k} \mid k, \mathbf{D}_{n}) \, |
\hat{\mathcal{I}}_{k}|^{-1/2} , \quad \text{for all $k$,}
\end{align}
where $\hat{\mathbf{x}}_{k}$ is the maximizer of
$\pi _{\text{un.}}(\, \cdot \mid k, \mathbf{D}_{n})$ (and
$\pi (\, \cdot \mid k, \mathbf{D}_{n})$), $\hat{\mathcal{I}}_{k}$ is minus
the matrix of second derivatives of
$\log \pi _{\text{un.}}(\, \cdot \mid k, \mathbf{D}_{n})$ (and
$\log \pi (\, \cdot \mid k, \mathbf{D}_{n})$). We use
$\hat{\mathcal{I}}_{k}$ to denote the matrix of minus the second derivatives
because when the prior density of the parameters is uniform, this matrix
corresponds to the observed information matrix. It is evaluated at
$\hat{\mathbf{x}}_{k}$, but we make this implicit to simplify because it
will always be the case in this paper. Note that other approximations may
be employed. In our framework, we are interested by approximations that
are easy to compute (at least in some cases) and asymptotically exact as
$n \longrightarrow \infty $. This is the case for
$\hat{\pi }(k \mid \mathbf{D}_{n})$ in \eqref{eqn_pi_hat}; it is more precisely
a consistent estimator of $\pi (k \mid \mathbf{D}_{n})$ under the assumptions
mentioned above.

We thus propose to define the model proposal distributions as follows:
\begin{equation*}
g(k, k') \propto h\left (\hat{\pi }(k' \mid \mathbf{D}_{n}) \big /
\hat{\pi }(k \mid \mathbf{D}_{n})\right ),
\end{equation*}
for all $k'\in \mathbf{N}(k)$. The function $h$ needs to be specified.
A natural choice (that does not satisfy the conditions mentioned in Section~\ref{sec_objectives}) is the identity function:
$g(k, k') \propto \hat{\pi }(k' \mid \mathbf{D}_{n}) \big / \hat{\pi }(k
\mid \mathbf{D}_{n}) \propto \hat{\pi }(k' \mid \mathbf{D}_{n})$. The resulting
proposal distributions are called \textit{globally-balanced} in
\cite{zanella2019informed}. To understand why, consider the situation where
the size of $\mathcal{K}$ is small and it is feasible to switch from any
model to any other one, and thus we can set
$\mathbf{N}(k) = \mathcal{K}$ for all $k$. In this case, $g(k, k')$ is
exactly equal to $\hat{\pi }(k' \mid \mathbf{D}_{n})$, and for large enough
$n$,
$\hat{\pi }(k' \mid \mathbf{D}_{n}) \approx \pi (k' \mid \mathbf{D}_{n})$,
corresponding to independent Monte Carlo sampling when
$\mathbf{u}_{k \mapsto k'} =\mathbf{y}_{k'} \sim \pi (\, \cdot \mid k',
\mathbf{D}_{n})$. Using such \textit{global} proposal distributions
$g(k, \cdot \,)$ with $\mathbf{N}(k) = \mathcal{K}$ thus asymptotically
leaves the target distribution invariant without an accept/reject step,
hence the name \textit{globally-balanced}. In contrast, when the ambient
space becomes larger and larger and the neighbourhoods have a fixed size,
i.e.\ when the locally-balanced proposal distributions asymptotically leave
the target distribution invariant, setting $h$ to the identity function
instead asymptotically leaves a distribution with
$\mathbf{X}_{K} \mid K \sim \pi (\, \cdot \mid K,\mathbf{D}_{n})$ and
$K \sim \pi (\, \cdot \mid \mathbf{D}_{n})^{2}$ invariant. We recommend
to set $h$ to the identity only when it is to feasible to set
$\mathbf{N}(k) = \mathcal{K}$ for all $k$. Otherwise, it seems better to
use locally-balanced proposal distributions; this is the case for instance
in \cite{zanella2019informed} and our moderate-size (both in $n$ and dimension)
variable-selection example.

Two choices of locally-balanced functions $h$ are considered in
\cite{zanella2019informed}: $h(x) = \sqrt{x}$ and $h(x) = x/(1+x)$. The
choice $h(x) = x/(1+x)$ yields what the author calls
\textit{Barker proposal distributions} because of the connection with
\cite{barker1965monte}'s acceptance-probability choice:
\begin{equation*}
\frac{\widehat{\pi }(k' \mid \mathbf{D}_{n}) \, \big / \, \widehat{\pi }(k \mid \mathbf{D}_{n})}{1 + \widehat{\pi }(k' \mid \mathbf{D}_{n}) \, \big / \, \widehat{\pi }(k \mid \mathbf{D}_{n})}=
\frac{\widehat{\pi }(k' \mid \mathbf{D}_{n})}{\widehat{\pi }(k' \mid \mathbf{D}_{n}) + \widehat{\pi }(k \mid \mathbf{D}_{n})}.
\end{equation*}
The empirical results of \cite{zanella2019informed} suggest that this latter
choice is superior because $h$ is bounded, which stabilizes the normalizing
constants $c(k)$ and $c(k')$ and thus the acceptance probabilities.
\cite{livingstone2019robustness} reached the same conclusion using locally-balanced
proposal distributions for continuous random variables. In our numerical
analyses, both choices lead to similar performances. In practice, a user
will choose one; we thus recommend setting $h(x) = x/(1+x)$.

\subsection{Specification of the parameter proposal distributions}\label{sec_parameter_proposal_design}

In this section, we discuss the specification of the functions
$q_{k\mapsto k'}$ and $\mathcal{D}_{k\mapsto k'}$ that are used when
$k' \neq k$, i.e.\ during model-switch attempts. As mentioned, in Section~\ref{sec_objectives}, HMC is used during parameter-update attempts.
The tuning of HMC free parameters is discussed in Section~\ref{sec_resulting_RJ}.

\cite{green2003trans} proposed to set $q_{k\mapsto k'}$ and
$\mathcal{D}_{k\mapsto k'}$ as if the parameters of both Model $k$ and
Model $k'$ were normally distributed. More precisely, when
$d_{k'} > d_{k}$ (the other cases are similar), the author proposed to
set
%
%e5 #&#
\begin{align}
\label{eqn_D_green}
\mathbf{y}_{k'} = \boldsymbol{\mu }_{k'} + \mathbf{L}_{k'}\left (
\begin{array}{cc}
\mathbf{L}_{k}^{-1}(\mathbf{x}_{k} - \boldsymbol{\mu }_{k})
\cr
\mathbf{u}_{k\mapsto k'}
\end{array}\right ),
\end{align}
where $\boldsymbol{\mu }_{k'}$ and $\boldsymbol{\mu }_{k}$ are vectors and
$\mathbf{L}_{k'}$ and $\mathbf{L}_{k}$ are matrices. If
$\pi (\, \cdot \mid k, \mathbf{D}_{n})$ and
$\pi (\, \cdot \mid k', \mathbf{D}_{n})$ are normal distributions with
means $\boldsymbol{\mu }_{k}$ and $\boldsymbol{\mu }_{k'}$ and covariance matrices
$\mathbf{L}_{k} \mathbf{L}_{k}^{T}$ and
$\mathbf{L}_{k'} \mathbf{L}_{k'}^{T}$, then setting
$q_{k\mapsto k'}$ to a $(d_{k'} - d_{k})$-dimensional standard normal yields
$\mathbf{y}_{k'} \sim \pi (\, \cdot \mid k', \mathbf{D}_{n})$. The function
yielding $\mathbf{y}_{k'}$ in \eqref{eqn_D_green} defines
$\mathcal{D}_{k\mapsto k'}$. Note that in this case
$\mathbf{u}_{k'\mapsto k}$ is non-existent. Note also that in the following,
we will write
$\varphi ( \, \cdot \, ; \boldsymbol{\mu }, \boldsymbol{\Sigma })$ to denote
both the distribution and PDF of a normal with mean
$\boldsymbol{\mu }$ and covariance matrix $\boldsymbol{\Sigma }$.

The rationale behind the approach of \cite{green2003trans} is to propose
$\mathbf{y}_{k'}$ around a suitable location with a suitable covariance
structure using a normal distribution. This approach is asymptotically
valid in some cases, by virtue of Bernstein-von Mises theorems (see, e.g.,
Theorem 10.1 in \cite{van2000asymptotic} and
\cite{kleijn2012bernstein}), meaning that for regular models,
$\mathbf{X}_{k} \mid k, \mathbf{D}_{n}$ is asymptotically normally distributed
as $n \longrightarrow \infty $, for all $k$, when $\mathcal{K}$ and
$d_{k}$ are fixed and do not depend on $n$. In fact,
$\pi (\, \cdot \mid k, \mathbf{D}_{n})$ concentrates around
$\hat{\mathbf{x}}_{k}$, with
$\hat{\mathbf{x}}_{k} \longrightarrow \mathbf{x}_{k}^{*}$, at a rate of
$1/\sqrt{n}$ with a shape that resembles a normal for large enough
$n$ (see Section~\ref{sec_model_proposal_design} for the definition of
$\hat{\mathbf{x}}_{k}$). When Model $k$ is well specified,
$\mathbf{x}_{k}^{*}$ is the true parameter value and the covariance matrix
is $\mathcal{I}_{k}^{-1}$, the latter denoting the inverse information
matrix based on the whole data set, evaluated at
$\mathbf{x}_{k}^{*}$. We make the dependence on $\mathbf{x}_{k}^{*}$ implicit
to simplify; the inverse information matrix will always be evaluated at
$\mathbf{x}_{k}^{*}$ in the following. The information matrix
$\mathcal{I}_{k}$ is equivalent to $n$ times the information of one data
point in the independent and identically distributed (IID) setting. In
the misspecified case, $\mathbf{x}_{k}^{*}$ is the
\textit{best possible value} under Model $k$ and the covariance matrix is
slightly different than $\mathcal{I}_{k}^{-1}$ (see
\cite{kleijn2012bernstein} for the details).

The approach to specify $q_{k\mapsto k'}$ and
$\mathcal{D}_{k\mapsto k'}$ proposed here is a simpler version of that
of \cite{green2003trans}, but one that is based on the same idea: set
\begin{equation*}
\mathbf{y}_{k'} = \mathbf{u}_{k\mapsto k'},
\end{equation*}
with
$q_{k\mapsto k'} = \varphi (\, \cdot \,; \hat{\mathbf{x}}_{k'},
\hat{\mathcal{I}}_{k'}^{-1})$, for any reachable Model $k'$, regardless
of $k$ (see Section~\ref{sec_model_proposal_design} for the definition of
$\hat{\mathcal{I}}_{k'}$). Consequently,
$\mathbf{u}_{k'\mapsto k} = \mathbf{x}_{k}$ and
$\mathcal{D}_{k\mapsto k'}(\mathbf{x}_{k}$,
$ \mathbf{u}_{k\mapsto k'}) = (\mathbf{u}_{k\mapsto k'}, \mathbf{x}_{k})
= (\mathbf{y}_{k'}, \mathbf{u}_{k' \mapsto k})$, implying that
$q_{k'\mapsto k} = \varphi (\, \cdot \,; \hat{\mathbf{x}}_{k},
\hat{\mathcal{I}}_{k}^{-1})$ and
$|J_{\mathcal{D}_{k\mapsto k'}}(\mathbf{x}_{k}, \mathbf{u}_{k\mapsto k'})|
= 1$. Thus, if Model $k'$ is a well-specified and regular model,
$\mathbf{y}_{k'}$ is asymptotically generated from the correct conditional
distribution, i.e.\ $q_{k\mapsto k'}$ and
$\pi (\, \cdot \mid k', \mathbf{D}_{n})$ are asymptotically equal, because
$\pi (\, \cdot \mid k', \mathbf{D}_{n})$ is asymptotically equal to
$\varphi (\, \cdot \,; \hat{\mathbf{x}}_{k'}, \mathcal{I}_{k'}^{-1})$ (Bernstein-von
Mises theorem) and
$\hat{\mathcal{I}}_{k'} / n - \mathcal{I}_{k'} / n \longrightarrow 0$ in
probability, where we say that a matrix converges in probability whenever
all entries converge in probability.

It may be the case that $q_{k\mapsto k'}$ and
$\pi (\, \cdot \mid k', \mathbf{D}_{n})$ are \emph{not} asymptotically equal
because Model $k'$ may be misspecified, implying that the limiting covariance
of $\pi (\, \cdot \mid k', \mathbf{D}_{n})$ is not
$\mathcal{I}_{k'}^{-1}$. The difference between
$\mathcal{I}_{k}^{-1}$ and the correct limiting covariance depends on how
close Model $k'$ is to the \textit{true model}, in a Kullback-Leibler divergence
sense. In general, the true model is unknown; therefore we cannot obtain
an analytical expression of the correct limiting covariance. The recommended
methodology thus represents in some cases an approximated version of the
asymptotically exact one, but it is a methodology that can be easily applied
and, fortunately, the resulting algorithms are expected to propose rarely
models that are far from the true model because of the design of $g$, as
long as a subset of $\mathcal{K}$ are suitable approximations to the true
model. In fact when the latter is true, the resulting algorithms are expected
to behave similarly to the asymptotically exact ones. This discussion is
made more precise in Section~\ref{sec_asymptotics}. In our numerical example,
the resulting algorithms are sufficiently close to the asymptotically exact
ones and $n$ is sufficiently large so as to yield efficient samplers. Note
that the covariance matrices can be estimated otherwise to obtain a generic
asymptotically exact methodology; for instance, they can be estimated using
pilot MCMC runs. This has for disadvantage of being computationally expensive
and challenging.

We finish this section by noting that the approach of
\cite{green2003trans} is also an approximated version of the asymptotically
exact one when $\boldsymbol{\mu }_{k} = \hat{\mathbf{x}}_{k}$ and
$\mathbf{L}_{k}$ is such that
$\mathbf{L}_{k} \mathbf{L}_{k}^{T} = \hat{\mathcal{I}}_{k}^{-1}$ for all
$k$. We also note that a RJ user employing the approach of
\cite{green2003trans} or that described above, already computes what is
required to design $g(k, \cdot \,)$ as in Section~\ref{sec_model_proposal_design}. This user can therefore implement
the proposed methodology without much additional effort.

\subsection{Resulting RJ}\label{sec_resulting_RJ}

We now present in Algorithm~\ref{proposed_RJ} a special case of Algorithm~\ref{algo_RJ} with inputs specified as in Sections \ref{sec_model_proposal_design} and \ref{sec_parameter_proposal_design}.

%a2 #&#
\begin{algorithm}[ht]
\caption{Proposed RJ}
\label{proposed_RJ}
\begin{enumerate}
 \itemsep 0mm

  \item Sample $k' \sim g(k, \cdot \,)$ with $g(k, k') \propto h\left (\hat{\pi}(k' \mid \mathbf{D}_{n}) \big/ \hat{\pi}(k \mid \mathbf{D}_{n})\right )$ for all $k'\in \mathbf{N}(k)$.

  \item[2.(a)] If $k' = k$, attempt a parameter update using a MCMC kernel of invariant distribution $\pi (\, \cdot \mid k, \mathbf{D}_{n})$ while keeping the value of the model indicator k fixed.

  \item[2.(b)] If $k' \neq k$, sample $\mathbf{y}_{k'} \sim \varphi (\, \cdot \,; \hat{\mathbf{x}}_{k'}, \hat{\mathcal{I}}_{k'}^{-1})$ and $u \sim \mathcal{U}[0, 1]$. If
      \[
       u \leq \alpha _{\text{RJ}}((k,\mathbf{x}_{k}),(k',\mathbf{y}_{k'})) = 1 \wedge \frac{\pi (k', \mathbf{y}_{k'} \mid \mathbf{D}_{n}) \, g(k', k) \, \varphi (\mathbf{x}_{k}; \hat{\mathbf{x}}_{k}, \hat{\mathcal{I}}_{k}^{-1})}{\pi (k, \mathbf{x}_{k} \mid \mathbf{D}_{n}) \, g(k,k') \, \varphi (\mathbf{y}_{k'}; \hat{\mathbf{y}}_{k'}, \hat{\mathcal{I}}_{k'}^{-1})},
      \]
      set the next state of the chain to $(k', \mathbf{y}_{k'})$. Otherwise, set it to $(k,\mathbf{x}_{k})$.

  \item[3.] Go to Step 1.
 \end{enumerate}
\end{algorithm}

Recall that we recommend to set $h$ to the identity function when
$\mathbf{N}(k) = \mathcal{K}$ for all $k$ and to a locally-balanced function,
namely $h(x) = x/(1+x)$, otherwise.

We recommend to apply HMC to update the parameters in Step 2.(a). HMC generate
paths which evolve according to Hamiltonian dynamics, thus exploiting the
local structure of the target through gradient information; the endpoints
of these paths are used as proposals within a MH scheme. This proposal
mechanism allows the chain to take large steps while not ending up in areas
with negligible densities. HMC thus has the ability to decorrelate, which
is an interesting feature that is especially important in RJ given that
the samplers may not spend a lot of consecutive iterations updating the
parameters.

There exist several variants of HMC, the most popular being the
\textit{No-U-Turn sampler} \citep{hoffman2014no}. The latter can be run
automatically for parameter estimation using the probabilistic programming
langage \textit{Stan} (see, e.g., \cite{carpenter2017stan}). To our knowledge,
it is not possible to directly use the Stan implementation within RJ to
update the parameters, and given that the No-U-Turn sampler is difficult
to implement, we have chosen to employ the ``vanilla'' version of HMC (see,
e.g., \cite{neal2011mcmc}). The implementation of this version for a given
model requires the specification of three free parameters: a step size,
a trajectory length and a mass matrix. To tune these parameters for a given
model, we recommend to first run Stan with the option
\textsf{static HMC} (meaning that the vanilla HMC is employed), and to
extract the step size and marginal empirical standard deviations of all
parameters. This step size results from a tuning procedure and can thus
directly be used. The mass matrix is set to a diagonal matrix with diagonal
elements being these standard deviations. For the trajectory length, a
grid search is applied and the best (found) value is used. The merit of
each trajectory length is evaluated through the minimum of marginal effective
sample sizes. Note that in this trans-dimensional framework, the momentum
variables need in theory to be considered with care. Theoretically, we
may consider that a momentum refreshment is performed every odd iteration,
and that the algorithm proceeds as in Algorithm~\ref{proposed_RJ} every even
iteration. Also, we need (in theory) to add or withdraw momentum variables
when switching models. In practice, we do not have to proceed in this way.
Given that momentum variables are only required when updating the parameters
and that the expectations that one wants to approximate typically do not
depend on these variables, we may generate them only when it is known that
a parameter update is proposed (i.e.\ $k' = k$).

Some authors (for instance, \cite{green2003trans}) mentioned that using
informed RJ samplers like Algorithm~\ref{proposed_RJ} may be problematic when
it is required to gather information for each model before running the algorithms,
because this is infeasible for large (or infinite) model spaces. The information
that is required for an iteration of Algorithm~\ref{proposed_RJ} is in fact
$\hat{\mathbf{x}}_{k'}$ and $\hat{\mathcal{I}}_{k'}$ for all
$k' \in \mathbf{N}(k)$, and possibly free parameters for an HMC step in
Step 2.(a). The required information can thus be gathered on the fly as
the chains reach new models; it is unnecessary to gather all of it beforehand.
The maximizers $\hat{\mathbf{x}}_{k'}$, matrices
$\hat{\mathcal{I}}_{k'}$ and free parameters should be stored and reused
next time they are needed. We only need to make sure that these are independent
of the chain paths to have a valid procedure. For maximizers, any general-purpose
optimization tool can be employed, but a starting point, if needed, should
consequently not depend on the chain path to make sure there is no interference;
the same recommendation follows for the computation of
$\hat{\mathcal{I}}_{k}$ and tuning of HMC free parameters.

This strategy makes the implementation of informed RJ samplers possible,
even if the model space is large or infinite, provided that the posterior
model PMF concentrates on a reasonable number of models (in the sense that
the number of different models visited during algorithm runs is on average
reasonable). When the PMF concentrates on few models, this implementation
strategy is expected to be highly effective as the information required
for model switches and parameter updates will in practice be gathered for
these few models and their neighbours only (if the algorithm is well initialized).

\subsection{Methods to use when the approximations are not accurate}
\label{sec_overview_methods_improve}

For finite $n$, the shapes of the posterior parameter densities under Model
$k$ and Model $k'$ involved during a model-switch attempt may be quite
different from bell curves. When this is the case, using normal approximations
to the distributions, i.e.\ $q_{k \mapsto k'} = \varphi (\, \cdot \,;
\hat{\mathbf{x}}_{k'}, \hat{\mathcal{I}}_{k'}^{-1})$ and
$q_{k' \mapsto k} = \varphi (\, \cdot \,; \hat{\mathbf{x}}_{k},
\hat{\mathcal{I}}_{k}^{-1})$, may lead to high rejection rates. The method
of \cite{karagiannis2013annealed} allows to build upon these approximations
by generating a sequence
$(\mathbf{x}_{k}^{(0)}, \mathbf{y}_{k'}^{(0)}), \ldots , (\mathbf{x}_{k}^{(T
- 1)}$, $ \mathbf{y}_{k'}^{(T-1)})$ using inhomogeneous Markov kernels
$K_{k \mapsto k'}^{(t)}$, with $\mathbf{y}_{k'}^{(T-1)}$ being the proposal
for the parameters of Model $k'$. More precisely, the starting point
$(\mathbf{x}_{k}^{(0)}, \mathbf{y}_{k'}^{(0)})$ is such that
$\mathbf{x}_{k}^{(0)} = \mathbf{x}_{k}$ and
$\mathbf{y}_{k'}^{(0)} \sim \varphi (\, \cdot \,; \hat{\mathbf{x}}_{k'},
\hat{\mathcal{I}}_{k'}^{-1})$, and the following
$(\mathbf{x}_{k}^{(t)}, \mathbf{y}_{k'}^{(t)})$ are sampled from
$K_{k \mapsto k'}^{(t)}((\mathbf{x}_{k}^{(t-1)}, \mathbf{y}_{k'}^{(t-1)}),
\cdot \,)$, $K_{k \mapsto k'}^{(t)}$ being, for each $t$, reversible with
respect to an annealing intermediate distribution given by
%
%e6 #&#
\begin{align}
\label{eqn_def_rho}
\rho _{k\mapsto k'}^{(t)}(\mathbf{x}_{k}^{(t)},\mathbf{y}_{k'}^{(t)})&
\propto \left [\pi (k,\mathbf{x}_{k}^{(t)}\mid \mathbf{D}_{n})\,
\varphi (\mathbf{y}_{k'}^{(t)}; \hat{\mathbf{x}}_{k'},
\hat{\mathcal{I}}_{k'}^{-1})\right ]^{1-\gamma _{t}}
\cr
&\qquad \times \left [\pi (k',\mathbf{y}_{k'}^{(t)}\mid \mathbf{D}_{n})
\, \varphi (\mathbf{x}_{k}^{(t)}; \hat{\mathbf{x}}_{k},
\hat{\mathcal{I}}_{k}^{-1})\right ]^{\gamma _{t}}.
\end{align}
The free parameter $T$ is a positive integer,
$\gamma _{0} = 0, \gamma _{T} = 1$ and $\gamma _{t} = t / T$ for each
$t \in \{1, \ldots , T - 1\}$.

We notice that when switching from Model $k$ to Model $k'$, we start with
distributions $\rho _{k\mapsto k'}^{(t)}$ close to
$\pi (k,\,\cdot \,\mid \mathbf{D}_{n})\otimes \varphi (\, \cdot \,;
\hat{\mathbf{x}}_{k'}, \hat{\mathcal{I}}_{k'}^{-1})$ to finish, after a
transition phase, with distributions close to
$\varphi (\, \cdot \,; \hat{\mathbf{x}}_{k}, \hat{\mathcal{I}}_{k}^{-1})
\otimes \pi (k',\,\cdot \,\mid \mathbf{D}_{n})$. Under some regularity
conditions, choosing $T$ large enough allows a smooth transition and makes
that the last steps of the sequence, with $t \geq t^{*}$, are similar to
steps of a time-homogeneous Markov chain with
$\varphi (\, \cdot \,; \hat{\mathbf{x}}_{k}, \hat{\mathcal{I}}_{k}^{-1})
\otimes \pi (k',\,\cdot \,\mid \mathbf{D}_{n})$ as a stationary distribution
(because $t / T \geq t^{*}/T \approx 1$), ensuring that
$(\mathbf{x}_{k}^{(T - 1)}, \mathbf{y}_{k'}^{(T - 1)})$ is approximately
distributed as
$\varphi (\, \cdot \,; \hat{\mathbf{x}}_{k}, \hat{\mathcal{I}}_{k}^{-1})
\otimes \pi (k',\,\cdot \,\mid \mathbf{D}_{n})$. A proof can be found in
\cite{gagnon2019NRJ}. Also, the acceptance probability in the resulting
RJ, given by
%
%e7 #&#
\begin{align}
\label{eqn_acc_andrieu_2013}
\alpha _{\text{RJ2}}((k,\mathbf{x}_{k}^{(0)}),(k',\mathbf{y}_{k'}^{(T-1)}))&:=1
\wedge \frac{g(k',k)}{g(k,k')} \,r_{\text{RJ2}}((k,\mathbf{x}_{k}^{(0)}),(k',
\mathbf{y}_{k'}^{(T-1)}))
\end{align}
with
%
%e8 #&#
\begin{align}
\label{eqn_ratio_A2013}
r_{\text{RJ2}}((k,\mathbf{x}_{k}^{(0)}),(k',\mathbf{y}_{k'}^{(T-1)})):=
\prod _{t=0}^{T-1}
\frac{\rho _{k\mapsto k'}^{(t+1)}(\mathbf{x}_{k}^{(t)},\mathbf{y}_{k'}^{(t)})}{\rho _{k\mapsto k'}^{(t)}(\mathbf{x}_{k}^{(t)},\mathbf{y}_{k'}^{(t)})},
\end{align}
is such that
$r_{\text{RJ2}}((k,\mathbf{x}_{k}^{(0)}),(k',\mathbf{y}_{k'}^{(T-1)}))$
is a consistent estimator of $\pi (k'\mid \mathbf{D}_{n}) / $
$\pi (k\mid \mathbf{D}_{n})$ as $T\longrightarrow \infty $. This shows
that the resulting RJ behave asymptotically like the ones in which the
parameter proposals are sampled from the correct conditional distributions
$\pi (\,\cdot \,\mid k', \mathbf{D}_{n})$, \emph{for fixed $n$}.

Given that
$r_{\text{RJ2}}((k,\mathbf{x}_{k}^{(0)}),(k',\mathbf{y}_{k'}^{(T-1)}))$
can be seen as an estimator of
$\pi (k'\mid \mathbf{D}_{n})/\pi (k\mid \mathbf{D}_{n})$,
\cite{andrieu2018utility} proposed to exploit this through a scheme allowing
to generate in parallel $N$ estimates to average them to reduce the variance.
With $T$ fixed, increasing $N$ thus gets the resulting RJ closer to the
ones in which the parameter proposals are sampled from the correct conditional
distributions $\pi (\,\cdot \,\mid k', \mathbf{D}_{n})$. In Appendix~\ref{sec_improve_approx}, it is shown that the methods of
\cite{karagiannis2013annealed} and \cite{andrieu2018utility} with their
estimators of ratios of posterior probabilities can be used to enhance
the approximations
$\widehat{\pi }(k'\mid \mathbf{D}_{n})/\widehat{\pi }(k\mid \mathbf{D}_{n})$
in $g(k, \cdot \,)$.

\section{Asymptotic result and analysis}\label{sec_asymptotics}

In Section~\ref{sec_parameter_proposal_design}, we described situations where
using an informed but approximate design for $q_{k\mapsto k'}$ is expected
to yield RJ that behave like asymptotically exact ones, the latter being
asymptotically equivalent to ideal RJ that sample parameters
$\mathbf{y}_{k'}$ during model switches from the correct conditional distributions
$\pi (\, \cdot \mid k', \mathbf{D}_{n})$ and that use the unnormalized
version of the model probabilities $\pi (k \mid \mathbf{D}_{n})$ to construct
$g$. In Section~\ref{sec_asymptotic_result}, we make this discussion precise
and present a weak convergence result. We next provide in Section~\ref{sec_analysis_limiting_RJ} an analysis of the limiting RJ.

\subsection{Asymptotic result}\label{sec_asymptotic_result}

The weak convergence result presented in this section is about weak convergence
of Markov chains, but the convergence does not happen for almost all data
sets $\mathbf{D}_{n}$. It rather occurs with a probability that becomes
closer to 1 as $n \longrightarrow \infty $, where the randomness comes
from $\mathbf{D}_{n}$. A first statement of this kind in MCMC recently
appeared in \cite{schmon2018large}. It allows to exploit Berstein-von Mises
theorems which are about the convergence of posterior parameter distributions
towards normal distributions in total variation (TV), with a probability
that converges to 1 (thus not for almost all data sets
$\mathbf{D}_{n}$). In this section, the randomness for all convergence
in probability comes from $\mathbf{D}_{n}$; the source of randomness is
thus omitted to simplify the statements.

The first condition for the weak convergence result to hold is about an
explicit independence between $n$ and the state-space.

%A1 #&#
\begin{Assumption}%
\label{ass1}
The model space $\mathcal{K}$ and the parameter dimensions, i.e.\ $d_{k}$
for all $k \in \mathcal{K}$, do not change with $n$.
\end{Assumption}

The second condition is that $\hat{\pi }(k \mid \mathbf{D}_{n})$ is an asymptotically
exact estimator of $\pi (k \mid \mathbf{D}_{n})$, the latter admitting
a limit $\bar{\pi }(k)$, for all $k$. In some cases, the posterior model
mass concentrates. For instance when a model, say Model $k^{*}$, is well
specified and regular, $\bar{\pi }(k^{*}) = 1$ (see, e.g.,
\cite{johnson2012bayesian} in linear regression).

%A2 #&#
\begin{Assumption}%
\label{ass2}
For each $k \in \mathcal{K}$, the pair of random variables
$(\widehat{\pi }(k\mid \mathbf{D}_{n}), \pi (k\mid \mathbf{D}_{n}))$ is
such that
$|\widehat{\pi }(k\mid \mathbf{D}_{n}) - \pi (k\mid \mathbf{D}_{n})|$ and
$|\pi (k\mid \mathbf{D}_{n}) - \bar{\pi }(k)|$ converge in probability towards
0 as $n \longrightarrow \infty $.
\end{Assumption}

The last assumption is about the asymptotic equivalence in TV between
$q_{k \mapsto k'}$ that is used to generate
$\mathbf{u}_{k \mapsto k'} = \mathbf{y}_{k'}$ and
$\pi (\, \cdot \mid k', \mathbf{D}_{n})$, for all $k'$. This happens when
a Bernstein-von Mises theorem holds for each Model $k'$ and
$q_{k \mapsto k'} = \varphi (\, \cdot \,; \hat{\mathbf{x}}_{k'},
\hat{\boldsymbol{\Sigma }}_{k'} / n)$ with $\hat{\mathbf{x}}_{k'}$ and
$\hat{\boldsymbol{\Sigma }}_{k'}$ being consistent estimator of
$\mathbf{x}_{k'}^{*}$ and $\boldsymbol{\Sigma }_{k'}$, respectively, where
$\boldsymbol{\Sigma }_{k'}$ is the limiting covariance matrix (if we divide
by $n$) of $\pi (\, \cdot \mid k', \mathbf{D}_{n})$. When, for instance,
Model $k'$ is well defined and regular and data points are IID,
$\hat{\boldsymbol{\Sigma }}_{k'} = (\hat{\mathcal{I}}_{k'} / n)^{-1}$ and
$\boldsymbol{\Sigma }_{k'}$ is the inverse information matrix of one data
point.

%A3 #&#
\begin{Assumption}%
\label{ass3}
For each $k' \in \mathcal{K}$, there exist $\hat{\mathbf{x}}_{k'}$,
$\mathbf{x}_{k'}^{*}$ and $\boldsymbol{\Sigma }_{k'}$ such that
$\hat{\mathbf{x}}_{k'} \longrightarrow \mathbf{x}_{k'}^{*}$ and
$\text{TV}(\pi (\, \cdot \mid k', \mathbf{D}_{n}), \varphi (\, \cdot
\,; \hat{\mathbf{x}}_{k'}, \boldsymbol{\Sigma }_{k'} / n))
\longrightarrow 0$, both in probability (Bernstein-von Mises theorem).
Also,
$q_{k \mapsto k'} = \varphi (\, \cdot \,; \hat{\mathbf{x}}_{k'},
\hat{\boldsymbol{\Sigma }}_{k'} / n)$ with
$\hat{\boldsymbol{\Sigma }}_{k'} \longrightarrow \boldsymbol{\Sigma }_{k'}$
in probability, implying that
%
%e9 #&#
\begin{align*}
\label{eqn_Bernstein}
\text{TV}(\varphi (\, \cdot \,; \hat{\mathbf{x}}_{k'},
\boldsymbol{\Sigma }_{k'} / n), q_{k \mapsto k'}) \longrightarrow 0
\quad \text{and} \quad \text{TV}(\pi (\, \cdot \mid k', \mathbf{D}_{n}),
q_{k \mapsto k'}) \longrightarrow 0,
\end{align*}
in probability, for all $k'$ (thus regardless of $k$).
\end{Assumption}

As mentioned in Section~\ref{sec_parameter_proposal_design}, the proposed methodology
represents in some cases an approximation to the asymptotically exact one,
because for some probable transitions from Model $k$ to Model $k'$,
$\text{TV}(\varphi (\, \cdot \,; \hat{\mathbf{x}}_{k'},
\boldsymbol{\Sigma }_{k'} / n), q_{k \mapsto k'})$ does not converge to 0
if Model $k'$ is misspecified. In fact, the proposed methodology is asymptotically
exact only if all models are well specified, or at least, if some are misspecified,
then $\bar{\pi }$ must assign a probability of 0 to them. Indeed, in the
latter case, considering that the chain is currently visiting Model
$k$ with $\bar{\pi }(k) > 0$,
\begin{equation*}
g(k, k') \propto h\left (\hat{\pi }(k' \mid \mathbf{D}_{n}) \big /
\hat{\pi }(k \mid \mathbf{D}_{n})\right ) \longrightarrow 0,
\end{equation*}
for all misspecified models with $k' \neq k$, so that the proposal distributions
$q_{k \mapsto k'}$ with
$\text{TV}(\varphi (\, \cdot \,; \hat{\mathbf{x}}_{k'},
\boldsymbol{\Sigma }_{k'} / n), q_{k \mapsto k'}) \not \longrightarrow 0$
will asymptotically never be used.

In practice, it is expected that all models are misspecified, but that
some well approximate the true model. We now briefly explain why the proposed
methodology is a good approximation to the asymptotically exact one in
this case and next present the weak convergence result. Denote by
$\mathcal{K}^{*} \subset \mathcal{K}$ the subset formed of the good models,
and consider that for all $k \in \mathcal{K}^{*}$ and for large enough
$n$,
$\text{TV}(\varphi (\, \cdot \,; \hat{\mathbf{x}}_{k'},
\boldsymbol{\Sigma }_{k'} / n), q_{k \mapsto k'})$ is small (because
$\hat{\boldsymbol{\Sigma }}_{k'} = (\hat{\mathcal{I}}_{k'} / n)^{-1}$ and
$\boldsymbol{\Sigma }_{k'} $ are similar), implying that
$\text{TV}(\pi (\, \cdot \mid k', \mathbf{D}_{n}), q_{k \mapsto k'})$ is
small. Consider also to simplify that all models are either good or poor
approximations to the true model, and that for all the poor models, i.e.\ with
$k' \notin \mathcal{K}^{*}$, $g(k, k') \longrightarrow 0$ for
$k \in \mathcal{K}^{*}$ (because
$\hat{\pi }(k' \mid \mathbf{D}_{n}) / \hat{\pi }(k \mid \mathbf{D}_{n})
\longrightarrow 0$). Therefore, for large enough $n$, if the chain is currently
visiting Model $k \in \mathcal{K}^{*}$, then with probability close to
1, a model in $\mathcal{K}^{*}$ is proposed, say Model $k'$, with
\begin{equation*}
g(k, k') \propto h\left (\hat{\pi }(k' \mid \mathbf{D}_{n}) \big /
\hat{\pi }(k \mid \mathbf{D}_{n})\right ) \approx h\left (\pi (k'
\mid \mathbf{D}_{n}) \big / \pi (k \mid \mathbf{D}_{n})\right ),
\end{equation*}
and next, parameters
$\mathbf{y}_{k'} = \mathbf{u}_{k \mapsto k'} \sim q_{k \mapsto k'}$ are
proposed, with
$\text{TV}(\pi (\, \cdot \mid k', \mathbf{D}_{n}), q_{k \mapsto k'})$ small.
In this situation, the proposed RJ are well approximating the asymptotically
exact ones. We keep this situation in mind for the rest of Section~\ref{sec_asymptotics}.

We now introduce notation that are required to present the weak convergence
result. Use $\{(K, \mathbf{X}_{K})_{n}(m): m \in \mathbb{N} \}$ to denote
a Markov chain simulated by a RJ that targets
$\pi (\, \cdot \,, \cdot \mid \mathbf{D}_{n})$ and sets
$g(k, k') \propto h\left (\hat{\pi }(k' \mid \mathbf{D}_{n}) \big /
\hat{\pi }(k \mid \mathbf{D}_{n})\right )$,
$q_{k \mapsto k'} = \varphi (\, \cdot \,; \hat{\mathbf{x}}_{k'},
\hat{\boldsymbol{\Sigma }}_{k'} / n)$ and
$\mathcal{D}_{k\mapsto k'}(\mathbf{x}_{k}, \mathbf{u}_{k\mapsto k'}) =
(\mathbf{u}_{k\mapsto k'}, \mathbf{x}_{k}) = (\mathbf{y}_{k'},
\mathbf{u}_{k' \mapsto k})$, for all $k$ and $k'\in \mathbf{N}(k)$. Use
$\{(K, \mathbf{X}_{K})_{\text{ideal}}(m): m \in \mathbb{N} \}$ to denote a
Markov chain simulated by an ideal RJ that targets
$\pi (\, \cdot \,, \cdot \mid \mathbf{D}_{n})$ as well, but instead sets
$g(k, k') \propto h\left (\pi (k' \mid \mathbf{D}_{n}) \big / \pi (k
\mid \mathbf{D}_{n})\right )$ and
$q_{k \mapsto k'} = \pi (\, \cdot \mid k', \mathbf{D}_{n})$ with
$\mathcal{D}_{k\mapsto k'}(\mathbf{x}_{k}, \mathbf{u}_{k\mapsto k'}) =
(\mathbf{u}_{k\mapsto k'}, \mathbf{x}_{k}) = (\mathbf{y}_{k'},
\mathbf{u}_{k' \mapsto k})$, for all $k$ and $k'\in \mathbf{N}(k)$.

Because the target distribution concentrates, we need to apply a transformation
to the Markov chains to obtain a non-trivial limit. We apply a transformation
reflecting that even for large $n$, the samplers continue to explore the
parameter spaces, but at different scales given that the parameters are
continuous variables. In contrast, if the posterior PMF
$\pi (\, \cdot \mid \mathbf{D}_{n})$ concentrates on say, one model, then
fewer and fewer models are visited during an algorithm run as more and
more of them have negligible mass as $n \longrightarrow \infty $. The transformation
that is applied aims at reflecting this reality and thus obtaining a limiting
situation that represents an approximation to what one encounters in practice.
We more precisely standardize the parameter variable
$\mathbf{X}_{K}$, but leave the model indicator $K$ as is, i.e.\ we define
$\{(K, \mathbf{Z}_{K})_{n}(m): m \in \mathbb{N} \}$ and
$\{(K, \mathbf{Z}_{K})_{\text{ideal}}(m): m \in \mathbb{N} \}$ such that
$(K, \mathbf{Z}_{K})_{n}(m) = (K, \sqrt{n}(\mathbf{X}_{K} -
\hat{\mathbf{x}}_{K}))_{n}(m)$ and
$(K, \mathbf{Z}_{K})_{\text{ideal}}(m) = (K, \sqrt{n}(\mathbf{X}_{K} -
\hat{\mathbf{x}}_{K}))_{\text{ideal}}(m)$, respectively, for all~$m$.\looseness=1

To be an admissible limit of
$\{(K, \mathbf{Z}_{K})_{n}(m): m \in \mathbb{N} \}$,
$\{(K, \mathbf{Z}_{K})_{\text{ideal}}(m): m \in \mathbb{N} \}$ would need to
be independent of $n$, but it is readily seen to be not the case. For instance,
the invariant conditional density of $\mathbf{Z}_{K}$ given $K$ is
$\pi (\, \cdot \, / \sqrt{n} + \hat{\mathbf{x}}_{k} \mid k,
\mathbf{D}_{n}) / \sqrt{n}$, which is asymptotically equivalent to a normal
with mean $\mathbf{0}$ and covariance $\boldsymbol{\Sigma }_{k}$, but not
equal to it for all $n$. We thus introduce a process for which the transformation
is independent of $n$ and prove that
$\{(K, \mathbf{Z}_{K})_{n}(m): m \in \mathbb{N} \}$ and
$\{(K, \mathbf{Z}_{K})_{\text{ideal}}(m): m \in \mathbb{N} \}$ weakly converge
towards it, showing that they all share a similar behaviour for large enough
$n$.

Denote by
$\{(K, \mathbf{X}_{K})_{\text{limit}}(m): m \in \mathbb{N} \}$ this process.
It is an RJ which targets a distribution such that
$\mathbf{X}_{K} \mid K \sim \varphi (\, \cdot \,; \hat{\mathbf{x}}_{K},
\boldsymbol{\Sigma }_{K} / n)$ and $K \sim \bar{\pi }$. It sets
$g(k, k') \propto h\left (\bar{\pi }(k') \big / \bar{\pi }(k)\right )$,
$q_{k \mapsto k'} = \varphi (\, \cdot \,; \hat{\mathbf{x}}_{k'},
\boldsymbol{\Sigma }_{k'} / n)$, and
$\mathcal{D}_{k\mapsto k'}(\mathbf{x}_{k}, \mathbf{u}_{k\mapsto k'}) =
(\mathbf{u}_{k\mapsto k'}, \mathbf{x}_{k}) = (\mathbf{y}_{k'},
\mathbf{u}_{k' \mapsto k})$, for all $k$ with $\bar{\pi }(k) > 0$ and
$k'\in \mathbf{N}(k)$. Because the chains will be assumed to start in stationarity,
there is no need to define the functions for $k$ with
$\bar{\pi }(k) = 0$. Denote the standardized version of
$\{(K, \mathbf{X}_{K})_{\text{limit}}(m): m \in \mathbb{N} \}$ by
$\{(K, \mathbf{Z}_{K})_{\text{limit}}(m): m \in \mathbb{N} \}$, which is such
that
$(K, \mathbf{Z}_{K})_{\text{limit}}(m) = (K, \sqrt{n}(\mathbf{X}_{K} -
\hat{\mathbf{x}}_{K}))_{\text{limit}}(m)$, for all $m$. We denote the weak
convergence of, for instance,
$\{(K, \mathbf{Z}_{K})_{n}(m): m \in \mathbb{N} \}$ to
$\{(K, \mathbf{Z}_{K})_{\text{limit}}(m): m \in \mathbb{N} \}$ by
$\{(K, \mathbf{Z}_{K})_{n}(m): m \in \mathbb{N} \} \Longrightarrow \{(K,
\mathbf{Z}_{K})_{\text{limit}}(m): m \in \mathbb{N} \}$.

%T1 #&#
\begin{Theorem}[Weak convergence]%
\label{thm_convergence}
Under Assumptions \ref{ass1} to \ref{ass3}, we have that
$\{(K, \mathbf{Z}_{K})_{n}(m): m \in \mathbb{N} \} \Longrightarrow \{(K,
\mathbf{Z}_{K})_{\text{limit}}(m): m \in \mathbb{N} \}$ and
$\{(K, \mathbf{Z}_{K})_{\text{ideal}}(m): m \in \mathbb{N} \}
\Longrightarrow \{(K, \mathbf{Z}_{K})_{\text{limit}}(m): m \in
\mathbb{N} \}$, both in probability as $n \longrightarrow \infty $, provided
that all chains start in stationarity.
\end{Theorem}

Formally, in Theorem~\ref{thm_convergence}, it is considered that, for each
sampler, the proposal mechanism for parameter updates and model switches
is the same. For instance, in the implementable RJ (that simulate
$\{(K, \mathbf{Z}_{K})_{n}(m): m \in \mathbb{N} \}$), a normal is used to
sample the parameter proposals in parameter-update attempts. This is often
not the case in practice; for instance in our numerical example, we use
HMC steps. To accommodate for such types of parameter-update mechanisms,
an additional, rather technical, assumption is required about the convergence
of the associated Markov kernels in some sense as
$n \longrightarrow \infty $. We proceeded in that way to simplify and for
brevity.

\subsection{Analysis of the limiting RJ}\label{sec_analysis_limiting_RJ}

The stochastic process
$\{(K, \mathbf{Z}_{K})_{\text{limit}}(m): m \in \mathbb{N} \}$ can be used
as a proxy to $\{(K, \mathbf{Z}_{K})_{n}(m): m \in \mathbb{N} \}$ in the large-sample
regime. An analysis of it thus helps understand how
$\{(K, \mathbf{Z}_{K})_{n}(m): m \in \mathbb{N} \}$, and thus
$\{(K, \mathbf{X}_{K})_{n}(m): m \in \mathbb{N} \}$, behave in this regime.
The stationary distribution of
$\{(K, \mathbf{Z}_{K})_{\text{limit}}(m): m \in \mathbb{N} \}$ is such that
$\mathbf{Z}_{K} \mid K \sim \varphi (\, \cdot \,; \mathbf{0},
\boldsymbol{\Sigma }_{K})$ and $K \sim \bar{\pi }$.

A theoretical result that can be established in the situation where the
posterior model PMF concentrates on one model is the dominance of the limiting
RJ over any RJ targeting the same distribution and using the same functions,
except $g(k, \cdot \,)$. This is true because the limiting RJ does not
waste iterations trying to propose to switch to models other than the best.
This shows the relevance of using informed RJ which automatically incorporate
such features of the target.

Let us denote by $P_{\text{limit},1}$ and $P_{\text{limit},2}$ the Markov
kernels of the limiting RJ setting
$g(k, k') \propto h\left (\bar{\pi }(k') \big / \bar{\pi }(k)\right )$ and
that of the other RJ with the only difference being in $g$, respectively.
The result is a Peskun-Tierney ordering
\citep{peskun1973optimum,tierney1998note}, meaning an order on the Markov
kernels, implying an order on the asymptotic variances. We use
$\text{var}(f, P_{\text{limit},i})$ to denote the asymptotic variance of
the function $f$ applied to the Markov chain of transition kernel
$P_{\text{limit},i}$ at equilibrium, $i = 1, 2$.

%P1 #&#
\begin{Proposition}%
\label{prop_peskun}
Assume that there exists $k^{*}$ with $\bar{\pi }(k^{*}) = 1$. Then, for
any $\mathbf{z}_{k^{*}} \in \mathbb{R} ^{d_{k^{*}}}$ and measurable set
$A$,
\begin{equation*}
P_{\text{limit},1}((k^{*}, \mathbf{z}_{k^{*}}), A \setminus \{(k^{*},
\mathbf{z}_{k^{*}})\}) \geq P_{\text{limit},2}((k^{*}, \mathbf{z}_{k^{*}}),
A \setminus \{(k^{*}, \mathbf{z}_{k^{*}})\}),
\end{equation*}
implying that for any square-integrable function $f$,
\begin{equation*}
\text{var}(f, P_{\text{limit},1}) \leq \text{var}(f, P_{\text{limit},2}).
\end{equation*}
\end{Proposition}

Note that this result holds for globally-balanced proposal distributions
as well. What is required is that $h(x) = 0$ when $x = 0$ and
$h(x) > 0$ when $x> 0$, which holds when $h$ is the identity function.

In the limiting RJ, the parameters are sampled from the correct conditional
distributions and therefore the transition probabilities for $K$ are the
same as a MH sampler targeting $\bar{\pi }$ using
$g(k, k') \propto h\left (\bar{\pi }(k') \big / \bar{\pi }(k)\right )$ with
acceptance probabilities given by
%
%e10 #&#
\begin{align}
\label{eqn_acc_limiting}
1 \wedge \frac{\bar{\pi }(k') \, g(k', k)}{\bar{\pi }(k) \, g(k, k')} = 1
\wedge \frac{c(k)}{c(k')}.
\end{align}
This implies that any analysis of MH samplers with $h$ set to a locally-balanced
function (which is what we recommend in typical cases) applies to the limiting
RJ. For a thorough analysis of such samplers, we refer the reader to
\cite{zanella2019informed}. Recall that $c(k)$ and $c(k')$ are the normalizing
constants of $g(k, \cdot \,)$ and $g(k', \cdot \,)$, respectively. As mentioned
in Section~\ref{sec_objectives}, \cite{zanella2019informed} proved that
$c(k)/c(k')$ are close to 1 in some situations, suggesting that locally-balanced
samplers are efficient.

Note that in practice even when $\pi (\, \cdot \mid \mathbf{D}_{n})$ concentrates,
we have that $\pi (k \mid \mathbf{D}_{n}) >0$ for any finite $n$ and any
$k$, implying that before finding the best model with a significantly larger
probability, the chain may move around for a while. For these transitions
with acceptance probabilities close to $\alpha _{\text{MH}}(k, k')$ in the
large-sample regime, the result of \cite{zanella2019informed}, stating
that $\alpha _{\text{MH}}(k, k') \longrightarrow 1$ as the ambient space
becomes larger and larger, holds under some conditions; this suggests again
an efficiency in exploring the state-space and thus finding the best model.

When $h$ is the globally-balanced function, i.e.\ the identity function,
and $\mathbf{N}(k) = \mathcal{K}$ for all $k$, then the limiting RJ proposes
models and parameters independently of the current state of the chain and
these proposals are always accepted. It thus corresponds to independent
Monte Carlo sampling, which is often seen as the gold standard in MCMC.

\section{Application: Variable selection in wholly-robust linear regression}\label{sec_application}

In this section, we apply the proposed RJ methodology to a variable-selection
problem in a wholly-robust linear regression. An overview of wholly-robust
linear regression is provided in Section~\ref{sec_context_example} and then
the computational results are presented and analysed in Section~\ref{sec_results_example}.

\subsection{Wholly-robust linear regression}\label{sec_context_example}

A new technique emerged to gain robustness against outliers in parametric
modelling: replace the traditional distribution assumption (which is a
normal assumption in the problems previously studied) by a super-heavy-tailed
distribution assumption
\citep{desgagne2015robustness,DesGag2019,gagnon2018regression, gagnon2017PCR}.\vadjust{\goodbreak}
The rationale is that this latter assumption is more adapted to the eventual
presence of outliers by giving higher probabilities to extreme values.
A proof of effectiveness of the approach resides in the following: the
posterior distribution converges towards that based on the non-outliers
only (i.e.\ excluding the outliers) as the outliers move further and further
away from the bulk of the data. This theoretical result corresponds to
a property in Bayesian statistics called \textit{whole robustness}
\citep{desgagne2015robustness}, and implies a
\textit{conflict resolution} \citep{o2012bayesian}. As explained in the
papers cited above, the models with super-heavy-tailed distributions have
built-in robustness that resolve conflicts due to contradictory sources
of information in a sensitive way. It takes full consideration of non-outliers
and excludes observations that are undoubtedly outlying; in between these
two extremes, it balances and bounds the impact of possible outliers, reflecting
that there is a uncertainty about whether or not these observations really
are outliers.

These features of wholly-robust models are appealing, but they come at
price, mainly a computational one. The super-heavy-tailed distribution
used in the papers cited above cannot indeed be represented as a scale
mixture of normal distributions, contrarily to the Student distribution
which is commonly assumed in Bayesian robust models. A scale mixture representation
of the Student allows a straightforward implementation of the Gibbs sampler
in, for instance, linear regression, for parameter estimation and variable
selection \citep{verdinelli1991bayesian}. A motivation for assuming a distribution
with heavier tails than the Student is that the latter only allows to reach
\textit{partial robustness}. In linear regression, partial robustness translates
into robust regression-coefficient point estimates, but inflated posterior
variances, which contaminates uncertainty assessments and thus model selection
(see \cite{hayashi2020bayesian} for a proof of partial robustness in linear
regression).

In \cite{gagnon2018regression}, the convergence of the posterior distribution
is proved under the most general linear-regression framework, encompassing
analysis of variance and covariance (ANOVA and ANCOVA), and variable selection.
In this section, we apply the RJ methodology presented in the previous
sections to sample from a joint posterior distribution of wholly-robust
linear regression models and their parameters. RJ is thus required for
this task. The data analysed are the same prostate-cancer data mentioned
in Figure~\ref{fig_1}.

The models are thus linear regressions, but the errors follow a super-heavy-tailed
distribution. See Appendix~\ref{sec_supp_variable_selection} for all the details
about the models and inputs required for RJ implementation. The super-heavy-tailed
distribution used is called \textit{log-Pareto-tailed normal} (LPTN). This
distribution was introduced by \cite{desgagne2015robustness}. Its density
exactly matches that of the normal on the interval $[-\tau , \tau ]$, where
$\mathbb{P}(-\tau \leq \varphi ( \, \cdot \,; 0, 1) \leq \tau ) =
\rho $. Outside of this area, the tails of this continuous density behave
as a log-Pareto: $(1 / |x|)(\log |x|)^{-\lambda - 1}$, hence its name.
The only free parameter of this distribution is $\rho $: the parameter
$\lambda $ is a function of $\rho $ and $\tau $, the latter being itself
a function of $\rho $. The linear regression with a LPTN is thus expected
to behave similarly to the traditional normal one in the absence of outliers.
Not only this is the case in the absence of outliers, but the limiting
robust-regression posterior distribution (as the distance between the outliers
and the bulk of the data approaches infinity) is also similar to the normal
one, the latter being instead based on the non-outliers only.

This resemblance is useful in our informed RJ framework because, first,
it suggests that the proposed methodology is an approximation to the asymptotically
exact one if some linear regressions are good approximations to the true
model. Indeed, in the normal linear-regression framework, Assumptions \ref{ass1} and \ref{ass2} of Section~\ref{sec_asymptotic_result} hold, and
a Berstein-von Mises theorem holds for each model. It is thus expected
to be the case in the robust LPTN framework as well, yet it is much more
difficult to prove.

Second, the resemblance between the wholly-robust and normal regressions
provides us with an approximation to the observed information matrix. We
indeed set $\hat{\mathcal{I}}_{k}$ in the robust LPTN framework to the
observed information matrix of the normal regression, but evaluated at
the maximizer under the wholly-robust model. The matrix
$\hat{\mathcal{I}}_{k}$ is thus basically the maximum a posteriori estimate
of the scale parameter in the robust model multiplied by the observed covariate
correlation matrix. The approximation is thus expected to be accurate when
there are no severe outliers among the covariate observations. One may
instead use a robust version of the correlation matrix (see
\cite{gagnon2017PCR} for such a robust alternative), but this is not investigated
here for brevity. Whether or not a robust version of the covariate matrix
is used is expected to have no impact on our comparison of the algorithms
in the next section because they all use
$q_{k \mapsto k'} = \varphi (\, \cdot \,; \hat{\mathbf{x}}_{k'},
\hat{\mathcal{I}}_{k'}^{-1})$ and the matrix
$\hat{\mathcal{I}}_{k'}$ only appears in $q_{k \mapsto k'}$. Indeed,
$|\hat{\mathcal{I}}_{k'}|^{-1/2}$ in
$\hat{\pi }(k' \mid \mathbf{D}_{n})$ \eqref{eqn_pi_hat} cancels with another
term in the same expression.

Finally, the resemblance between the wholly-robust and normal regressions
is useful to validate our RJ computer code. Indeed, the robust approach
yields an outlier-detection method and a ``clean'' data set without outliers
can thus be identified. Based on this data set, posterior coefficient estimates
and model probabilities of the normal models can be computed using their
closed-form expressions. These estimates and probabilities are expected
to be similar to those obtained under the robust models.

\subsection{Results and analysis}\label{sec_results_example}

The performance of the different algorithms are summarized in Table~\ref{table_performances} and Figure~\ref{fig_2}.
The results for the uninformed
RJ are based on 1,000 runs of 100,000 iterations, with burn-ins of 10,000.
The uninformed RJ is in fact a ``non-fully'' informed RJ corresponding
to Algorithm~\ref{proposed_RJ} with the only difference that
$g(k, \cdot \,) = \mathcal{U}\{\mathbf{N}(k)\}$ instead of locally- or
globally-balanced proposal distributions such that
$g(k, k') \propto h\left (\hat{\pi }(k' \mid \mathbf{D}_{n}) \big /
\hat{\pi }(k \mid \mathbf{D}_{n})\right )$ for all
$k'\in \mathbf{N}(k)$. It may thus be seen as corresponding to the approach
of \cite{green2003trans}. For a fair comparison with this RJ, the number
of iterations for which the fully informed RJ, i.e. Algorithm~\ref{proposed_RJ},
is run is reduced to account for the additional complexity of using informed
distributions $g(k, \cdot \,)$. The number of iterations has to be reduced
to 85,000 to reach the same runtime as the uninformed RJ; one run takes
about one minute in R on a i9 computer.\footnote{We do not claim optimality
of the computer code; we coded all algorithms in the same fashion for a
fair comparison.} To measure performance, we display the model-switching
acceptance rates, model-visit rates and relative increases in TV with respect
to the best sampler, which is Algorithm~\ref{proposed_RJ} with
$h(x)=\sqrt{x}$. The TVs are in between approximated model distributions
from RJ outputs and the posterior model PMF\footnote{We used accurate MCMC approximations
to the posterior model probabilities. We verified that the TV goes to 0
for all algorithms as the number of iterations increases.}.

The model-switching acceptance rates are the acceptance rates, but computed
considering only the iterations in which model switches are proposed. The
visit rates are the average number of model switches in one run, reported
per iteration. For both these measures, we count the number of accepted
model switches, and this number is divided by either the number of proposed
model switches or total number of iterations. These rates are thus similar
but they convey different information: model-switching acceptance rates
reflect the quality of the proposal distributions used during model switches,
while visit rates measure the propensity to propose a switch to another
model (and accept it). They together allow to understand why, for instance,
Algorithm~\ref{proposed_RJ} with $h(x)=\sqrt{x}$ is (slightly) better than
Algorithm~\ref{proposed_RJ} with $h(x)=x / (1 + x)$ in terms of TV. Indeed, even
if it has a (slightly) worse model-switch proposal scheme (model-switching
acceptance rates: 66\% vs 67\%), it proposes more often to switch models
(visit rates: 55\% vs 53\%), which has an impact on the TV.

%t1 #&#
\begin{table}
\caption{Performance of the uninformed and informed RJ in terms of model-switching
acceptance rate, model-visit rate and relative increase in TV with respect
to \protect Algorithm~\ref{proposed_RJ} with $h(x)=\sqrt{x}$}
\label{table_performances}
\centering
\begin{tabular}{l rrr}
\toprule
\textbf{Algorithms} & \textbf{Acc. rate} & \textbf{Visit rate} & \textbf{Rel. increase in TV} \cr
\midrule
Algorithm~\ref{proposed_RJ} w. $h(x)=\sqrt{x}$ & 66\% & 55\% & --- \cr
\midrule
Algorithm~\ref{proposed_RJ} w. $h(x)=x / (1 + x)$ & 67\% & 53\% & 2\% \cr
Uninformed RJ & 30\% & 27\% & 28\% \cr
Algorithm~\ref{proposed_RJ} w. $h(x) = x$ & 57\% & 46\% & 40\% \cr
\bottomrule
\end{tabular}
\end{table}

This robust-variable-selection application shows that even in a moderate-size
example, in this case with 8 covariates and 256 models, the locally-balanced
proposal distributions lead to RJ that can outperform RJ proposing models
uniformly at random, \emph{at fixed computational budget}. In contrast,
the globally-balanced proposal distribution does not yield an improvement
that is significant enough to compensate for the decrease in number of
iterations. We observed that even when the RJ with the function
$h(x) = x$ is run for the same number of iterations as the uniformed one,
the former only slightly outperforms the latter.

For the informed RJ to be effective, the approximations on which they rely
have to be accurate. Recall that their accuracy depends on the sample size
(at least in some scenarios regarding the approximations of
$\pi (\, \cdot \mid k', \mathbf{D}_{n})$ by $q_{k \mapsto k'}$, as explained
in Section~\ref{sec_asymptotic_result}). In this example, $n = 97$. To understand
if this is actually large for such a robust-linear-regression problem with
8 covariates, we can compare the model-switching acceptance rates of Algorithm~\ref{proposed_RJ}
with those of the ideal RJ they approximate. Recall
that ideal RJ set
$g(k, k') \propto h\left (\pi (k' \mid \mathbf{D}_{n}) \big / \pi (k
\mid \mathbf{D}_{n})\right )$ and
$q_{k \mapsto k'} = \pi (\, \cdot \mid k', \mathbf{D}_{n})$. It is possible
to evaluate their model-switching acceptance rates because these correspond
to the acceptance rates of MH samplers targeting the PMF
$\pi (\, \cdot \mid \mathbf{D}_{n})$ using
$g(k, k') \propto h\left (\pi (k' \mid \mathbf{D}_{n}) \big / \pi (k
\mid \mathbf{D}_{n})\right )$, computed considering only iterations in
which $k' \neq k$ is proposed. When $h(x)=x / (1 + x)$, the acceptance
rate of the MH sampler is 88\%, which is sufficiently close to 1 to suggest
that the asymptotic regime presented in \cite{zanella2019informed} where
the size of the discrete space increases and the acceptance probabilities
converge to 1 is nearly reached. To understand which approximations explain
the gap between Algorithm~\ref{proposed_RJ} with $h(x)=x / (1 + x)$ and a model-switching
acceptance rate of 67\% and its ideal counterpart with a rate of 88\%,
we can look at the model-switching acceptance rate of a RJ using methods
to sample parameter proposals from a distribution closer to
$\pi (\, \cdot \mid k', \mathbf{D}_{n})$ and next at that of a RJ additionally
using a method making $g(k, k')$ closer to being proportional to
$h\left (\pi (k' \mid \mathbf{D}_{n}) \big / \pi (k \mid \mathbf{D}_{n})
\right )$ (recall the overview of the methods presented in Section~\ref{sec_overview_methods_improve}). It is 84\% for the first RJ (with
$T = N = 10$), and it increases to 87\% for the second RJ. Using such methods
simply make RJ too computationally expensive to be competitive in this
example. Recall that their precision parameters can be increased to make
the second RJ arbitrarily close to the ideal one with a model-switching
acceptance rate of 88\%.

These results suggest that $n$ is of moderate size, in the sense that the
approximations used in Algorithm~\ref{proposed_RJ} are good enough to yield an
improvement, but it is expected that the improvement would be even more
marked if we had access to more data points. Also, the results suggest
that the approximations explaining the most part of the gap between Algorithm~\ref{proposed_RJ} with $h(x)=x / (1 + x)$ and a model-switching acceptance
rate of 67\% and its ideal counterpart with a rate of 88\% are the approximations
of $\pi (\, \cdot \mid k', \mathbf{D}_{n})$ by $q_{k \mapsto k'}$. We investigated
to see if there are significant differences between the covariance matrices
used in $q_{k \mapsto k'}$ and those of
$\pi (\, \cdot \mid k', \mathbf{D}_{n})$, and it is not the case. This
means that the approximations of
$\pi (\, \cdot \mid k', \mathbf{D}_{n})$ by $q_{k \mapsto k'}$ are not
as accurate as we would like because the densities
$\pi (\, \cdot \mid k', \mathbf{D}_{n})$ are significantly different from
bell curves. Our analysis also allows to rule out an issue with using the
same observed covariance matrix as the normal regression in the robust
LPTN framework. In fact, a robust-linear-regression analysis reveals the
presence of outliers, but only for models with negligible posterior mass
(both in the robust and non-robust frameworks). We finally note that, for
models with non-negligible posterior mass, there is no gross violation
of the assumptions underlying linear regression.

%f2 #&#
\begin{figure}[ht]
\centering
\includegraphics[width = 0.6\textwidth]{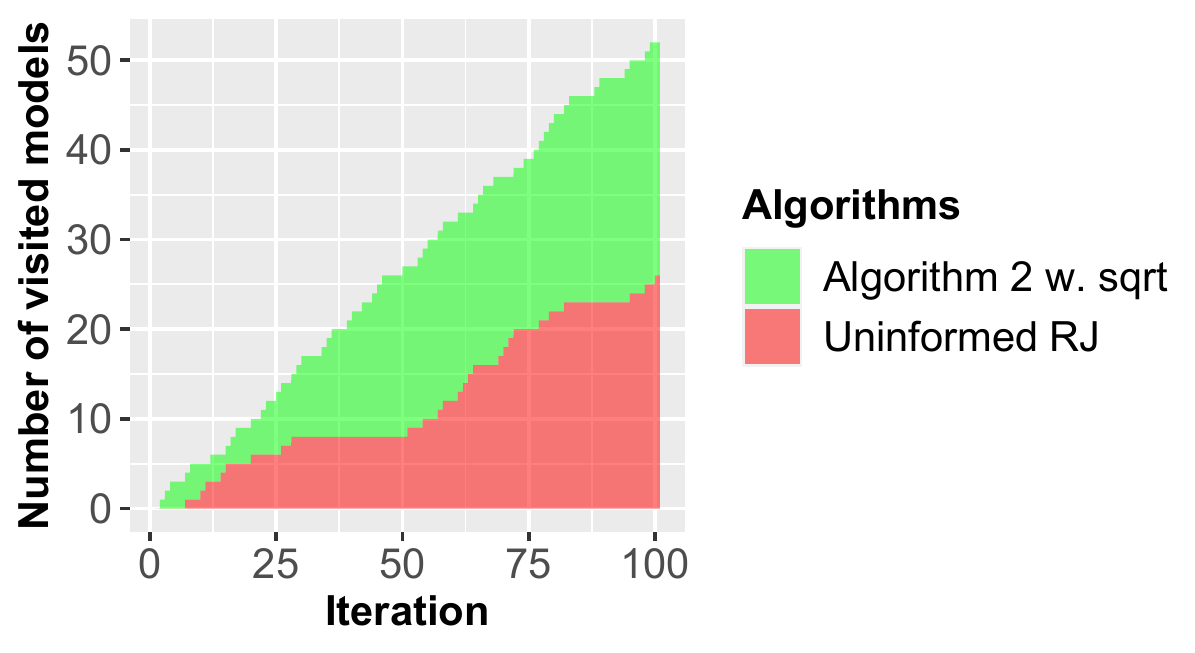}
\caption{Number of visited models as a function of the iteration
number in a typical run of the uninformed RJ and \protect
Algorithm~\ref{proposed_RJ} with $h(x)=\sqrt{x}$}
\label{fig_2}\vspace*{-3.5pt}
\end{figure}
\normalsize

\section{Discussion}\label{sec_discussion}

In this paper, we proposed simple fully-informed RJ and described the situations
in which they are expected to be efficient: when the sample size is sufficiently
large, Assumptions \ref{ass1}--\ref{ass2} of Section~\ref{sec_asymptotic_result} are verified, some models among
those considered well approximate\vadjust{\goodbreak} the true data generating process and Bernstein-von Mises theorems hold for these models. Informed
proposal distributions are crucial when the model probabilities and parameter
densities vary significantly within neighbourhoods of states of the Markov
chains. They are expected to vary significantly in the large-sample regime
in which the target concentrates. But we noticed with our numerical example
that they may vary greatly even when this large-sample regime is not reached.
The proposed RJ show major improvements in this example as the chains spend
less iterations at the same state, compared with the RJ that naively proposes
models uniformly at random and thus often tries to reach low probability
models, leading to higher rejection rates.\looseness=-1

Yet, the proposed samplers are reversible which allows them to return to
recently visited models often. The next step in the line of research of
trans-dimensional samplers for non-nested model selection is to propose
sampling schemes which do not suffer from this diffusive behaviour, but
instead induce persistent movement in the model-indicator process. A first
step in this direction has recently been made by
\cite{gagnon2020lifted}, but their approach cannot be used in all contexts
of non-nested model selection. It can nevertheless be applied in any variable-selection
problem. It is however expected to be efficient when the posterior model
mass is diffused over a large number of models with different number of
covariates, which is not the case in our numerical example.\vspace*{-3.5pt}

\bibliographystyle{rss}
\bibliography{reference}

\appendix

\section{Improving the approximations}\vspace*{-3.5pt}
\label{sec_improve_approx}

In practice, the sample size may not be large enough for the approximations
in the proposal mechanisms to be accurate. In Sections \ref{sec_andrieu_2013} and \ref{sec_andrieu_2018}, we present the details
of the methods reviewed in Section~\ref{sec_overview_methods_improve} that
allow to improve upon the normal-distribution approximations to
$\pi (\, \cdot \mid k, \mathbf{D}_{n})$. These methods turn out to be useful
for improving upon the approximations forming the model proposal distributions
$g(k, \cdot \,)$ as well; this is explained in detail in Section~\ref{sec_improve_models}.\vspace*{-3.5pt}

\subsection{RJ with the method of \cite{karagiannis2013annealed}}\vspace*{-3.5pt}
\label{sec_andrieu_2013}

In the definition of the annealing intermediate distributions
$\rho _{k \mapsto k'}^{(t)}$ in \eqref{eqn_def_rho}, it is considered that
$\mathcal{D}_{k\mapsto k'}(\mathbf{x}_{k}, \mathbf{u}_{k\mapsto k'}) =
(\mathbf{u}_{k\mapsto k'}, \mathbf{x}_{k}) = (\mathbf{y}_{k'},
\mathbf{u}_{k' \mapsto k})$, implying that\vadjust{\goodbreak}
$|J_{\mathcal{D}_{k\mapsto k'}}(\mathbf{x}_{k}, \mathbf{u}_{k\mapsto k'})|
= 1$, and
$q_{k'\mapsto k} = \varphi (\, \cdot \,; \hat{\mathbf{x}}_{k},
\hat{\mathcal{I}}_{k}^{-1})$. For the general definition, see \cite{karagiannis2013annealed}.

We now present in Algorithm~\ref{algo_RJ_andrieu_2013} an informed RJ incorporating
the method of \cite{karagiannis2013annealed}. Recall that the acceptance
probability in this RJ is defined in \eqref{eqn_acc_andrieu_2013}. Note
that Algorithm~\ref{algo_RJ_andrieu_2013} corresponds to Algorithm~\ref{proposed_RJ} when $T=1$; no path is sampled in Step 2.(b).

\begin{algorithm}[ht]
\caption{RJ incorporating the method of \cite{karagiannis2013annealed}}%
\label{algo_RJ_andrieu_2013}
\begin{enumerate}
 \itemsep 0mm

  \item[1.] Sample $k' \sim g(k, \cdot \,)$ with $g(k, k') \propto h\left (\hat{\pi}(k' \mid \mathbf{D}_{n}) \big/ \hat{\pi}(k \mid \mathbf{D}_{n})\right )$ for all $k'\in \mathbf{N}(k)$.

  \item[2.(a)] If $k' = k$, attempt a parameter update using a MCMC kernel of invariant distribution $\pi (\, \cdot \mid k, \mathbf{D}_{n})$ while keeping the value of the model indicator k fixed.

  \item[2.(b)] If $k' \neq k$, sample $\mathbf{y}_{k'}^{(0)} \sim \varphi (\, \cdot \,; \hat{\mathbf{x}}_{k'}, \hat{\mathcal{I}}_{k'}^{-1})$ and $u \sim \mathcal{U}[0, 1]$, and set $\mathbf{x}_{k}^{(0)}:=\mathbf{x}_{k}$. Sample a path $(\mathbf{x}_{k}^{(1)}, \mathbf{y}_{k'}^{(1)}),\ldots ,(\mathbf{x}_{k}^{(T-1)}, \mathbf{y}_{k'}^{(T-1)})$, where $(\mathbf{x}_{k}^{(t)}, \mathbf{y}_{k'}^{(t)})\sim K_{k \mapsto k'}^{(t)}((\mathbf{x}_{k}^{(t-1)}, \mathbf{y}_{k'}^{(t-1)}), \cdot \,)$. If $u  \leq \alpha _{\text{RJ2}}((k,\mathbf{x}_{k}^{(0)}),(k',\mathbf{y}_{k'}^{(T-1)}))$ set the next state of the chain to $(k',\mathbf{y}_{k'}^{(T-1)})$. Otherwise, set it to $(k,  \mathbf{x}_{k})$.

  \item[3.] Go to Step 1.
 \end{enumerate}
\end{algorithm}

Using the same proof technique as in \cite{gagnon2019NRJ}, it can be proved
that under regularity conditions a Markov chain simulated by Algorithm~\ref{algo_RJ_andrieu_2013} converges weakly to that of a RJ which is
able to sample from $\pi (\,\cdot \,\mid k, \mathbf{D}_{n})$ for all
$k$ with acceptance probabilities $\alpha _{\text{MH}}$ but using
$g(k, k') \propto h\left (\hat{\pi }(k' \mid \mathbf{D}_{n}) \big /
\hat{\pi }(k \mid \mathbf{D}_{n})\right )$, as
$T\longrightarrow \infty $, for fixed $n$. Note that the weak convergence
in this case is not in probability because the target is considered non-random
(contrarily to the framework under which Theorem~\ref{thm_convergence} is stated).

As $T$ increases, it is expected that the acceptance probabilities increase
steadily towards $\alpha _{\text{MH}}$ until convergence is reached. A RJ
user thus has to find a balance between a RJ close to be ideal and the
computational cost associated to this. As explained in
\cite{karagiannis2013annealed}, the computational cost scales linearly
with $T$. In some situations, the improvement as a function of $T$ is very
marked for $T \leq T_{0}$, effectively offsetting the computational cost.
This is not the case in our numerical example.

\cite{karagiannis2013annealed} prove that under two conditions Algorithm~\ref{algo_RJ_andrieu_2013} is valid, in the sense that the target distribution
is an invariant distribution. These conditions are the following.
\begin{description}
\item[Symmetry condition:] For $t=1,\ldots ,T-1$ the pairs of transition
kernels $K_{k \mapsto k'}^{(t)}(\,\cdot \,, \cdot \,)$ and
$K_{k' \mapsto k}^{(T-t)}(\,\cdot \,, \cdot \,)$ satisfy
\begin{align}
\label{eqn_symmetry}
K_{k \mapsto k'}^{(t)}((\mathbf{x}_{k}, \mathbf{y}_{k'}), \cdot \,) = K_{k'
\mapsto k}^{(T-t)}((\mathbf{x}_{k}, \mathbf{y}_{k'}), \cdot \,)
\quad \text{for any } (\mathbf{x}_{k}, \mathbf{y}_{k'}).
\end{align}
\item[Reversibility condition:] For $t=1,\ldots ,T-1$, and for any
$(\mathbf{x}_{k},\mathbf{y}_{k'})$ and
$(\mathbf{x}_{k}', \mathbf{y}_{k'}')$,
\begin{align}
\label{eqn_reversibility}
&\rho _{k\mapsto k'}^{(t)}(\mathbf{x}_{k},\mathbf{y}_{k'}) \, K_{k
\mapsto k'}^{(t)}((\mathbf{x}_{k}, \mathbf{y}_{k'}), (\mathbf{x}_{k}',
\mathbf{y}_{k'}'))
 =\rho _{k\mapsto k'}^{(t)}(\mathbf{x}_{k}',\mathbf{y}_{k'}')
\, K_{k \mapsto k'}^{(t)}((\mathbf{x}_{k}', \mathbf{y}_{k'}'), (
\mathbf{x}_{k}, \mathbf{y}_{k'})).
\end{align}
\end{description}

As mentioned in \cite{karagiannis2013annealed}, \eqref{eqn_symmetry} is
verified if for all $t$,
$K_{k \mapsto k'}^{(t)}(\,\cdot \,, \cdot \,)$ and
$K_{k' \mapsto k}^{(T-t)}(\,\cdot \,, \cdot \,)$ are MH kernels sharing
the same proposal distributions. \eqref{eqn_reversibility} is satisfied
if $K_{k \mapsto k'}^{(t)}(\,\cdot \,, \cdot \,)$ is a MH kernel targeting
$\rho _{k\mapsto k'}^{(t)}$. We recommend to use MALA (Metropolis-adjusted
Langevin algorithm, \cite{roberts1996exponential}) proposal distributions
whenever this is possible; see \cite{karagiannis2013annealed} for other
examples. If MALA proposal distributions are used, a step size is required
for each pair $(k, k')$ defining a transition from Model $k$ to Model
$k'$. For any of these model switches, we recommend to set the step size
to $\ell / (d_{k} + d_{k'})^{1/6}$, according to the findings in
\cite{roberts1998optimal}, and to tune $\ell $. To achieve the latter,
do a line search with a fixed value of $T$ and choose $\ell $ that maximizes
the model-switching acceptance rate, so that the resulting RJ is the closest
to the ideal one.

\subsection{RJ with additionally the method of \cite{andrieu2018utility}}
\label{sec_andrieu_2018}

Denote the $N$ estimates of
$\pi (k'\mid \mathbf{D}_{n})/\pi (k\mid \mathbf{D}_{n})$ produced by the
scheme of \cite{karagiannis2013annealed} by
$r_{\text{RJ2}}((k,\mathbf{x}_{k}^{(0)}),(k',\mathbf{y}_{k'}^{(T-1,1)})),
\ldots , r_{\text{RJ2}}((k $,
$\mathbf{x}_{k}^{(0)}),(k',\mathbf{y}_{k'}^{(T-1,N)}))$. Denote the average
(with simplified notation) by
\begin{equation*}
\bar{r}(k, k'):=\frac{1}{N} \sum _{j=1}^{N} r_{\text{RJ2}}((k,
\mathbf{x}_{k}^{(0)}),(k',\mathbf{y}_{k'}^{(T-1,j)})).
\end{equation*}

We now present in Algorithm~\ref{algo_RJ_andrieu_2018} the RJ additionally incorporating
the method of \cite{andrieu2018utility}.

\begin{algorithm}[ht]
 \caption{RJ additionally incorporating the method of \cite{andrieu2018utility}}
\label{algo_RJ_andrieu_2018}
{\setlength\leftmargini{30pt}
\begin{enumerate}
 \itemsep 0mm

  \item[1.] Sample $k' \sim g(k, \cdot \,)$ with $g(k, k') \propto h\left (\hat{\pi}(k' \mid \mathbf{D}_{n}) \big/ \hat{\pi}(k \mid \mathbf{D}_{n})\right )$ for all $k'\in \mathbf{N}(k)$.

  \item[2.(a)] If $k' = k$, attempt a parameter update using a MCMC kernel of invariant distribution $\pi (\, \cdot \mid k, \mathbf{D}_{n})$ while keeping the value of the model indicator k fixed.

  \item[2.(b)] If $k'\neq k$, generate $u_{a}, u_{c}\sim \mathcal{U}(0, 1)$. If $u_{c}\leq 1/2$ go to Step 2.(b-i), otherwise go to Step 2.(b-ii).

  \item[2.(b-i)] Sample $N$ proposals $\mathbf{y}_{k'}^{(T-1, 1)},\ldots ,\mathbf{y}_{k'}^{(T-1, N)}$ as in Step 2.(b) of \autoref{algo_RJ_andrieu_2013}. Sample $j^{*}$ from a PMF such that $\mathbb {P}(J^{*}=j) \propto r_{\text{RJ2}}((k,\mathbf{x}_{k}), (k', \mathbf{y}_{k'}^{(T-1,j)}))$. If
\begin{align*}
   u_{a}  &\leq \frac{g(k',k)}{g(k,k')} \, \bar{r}(k, k'),
  \end{align*}
set the next state of the chain to $(k',\mathbf{y}_{k'}^{(T-1,j^{*})})$. Otherwise, set it to $(k,  \mathbf{x}_{k})$.

  \item[2.(b-ii)] Sample one forward path as in Step 2.(b) of \autoref{algo_RJ_andrieu_2013}. Denote the endpoint by $\mathbf{y}_{k'}^{(T-1, 1)}$. From $\mathbf{y}_{k'}^{(T-1, 1)}$,  sample $N - 1$ reverse paths again as in Step 2.(b) of \autoref{algo_RJ_andrieu_2013}, yielding $N - 1$ proposals for the parameters of Model $k$. If
\begin{align*}
   u_{a}  &\leq \frac{g(k',k)}{g(k,k')} \, \bar{r}(k', k)^{-1},
  \end{align*}
set the next state of the chain to $(k', \mathbf{y}_{k'}^{(T-1,1)})$. Otherwise, set it to $(k,  \mathbf{x}_{k})$.

  \item[3.] Go to Step 1.
 \end{enumerate}}
\end{algorithm}

No additional assumptions to those presented in Section~\ref{sec_andrieu_2013} are required to guarantee that Algorithm~\ref{algo_RJ_andrieu_2018} is valid. \cite{andrieu2018utility} prove
that increasing $N$ decreases monotonically the asymptotic variances of
the Monte Carlo estimates produced by RJ incorporating their approach.

It is expected that increasing $N$ (as increasing $T$ in the last section)
leads to a steady increase in the acceptance probabilities until convergence
is reached. As with Algorithm~\ref{algo_RJ_andrieu_2013}, there exists a balance
between a RJ close to be ideal and the computational cost associated to
this. An advantage of the approach presented in this section is that several
operations can be executed in parallel, so that the computational burden
is alleviated. The computational cost (over that of using the approach
of \cite{karagiannis2013annealed}) nevertheless depends on $N$ because
computational overheads have to be taken into account. In our numerical
example, the improvement as a function of $N$ is not significant enough
to offset the computational cost.

\subsection{Improving the model proposal distribution}
\label{sec_improve_models}

We have seen in Section~\ref{sec_andrieu_2018} that $\bar{r}(k, k')$ and the
ratios $r_{\text{RJ2}}$ forming it are estimators of
$\pi (k'\mid \mathbf{D}_{n})/\pi (k\mid \mathbf{D}_{n})$. They can thus
be used to enhance the approximations to the posterior probability ratios
in $g(k, \cdot \,)$.

If we want to enhance the PMF $g(k, \cdot \,)$, we need to improve
$\widehat{\pi }(l\mid \mathbf{D}_{n}) / \widehat{\pi }(k\mid \mathbf{D}_{n})$
for all $l\in \mathbf{N}(k)$ as these are all involved in the construction
of the PMF. Also, once $k'$ has been sampled, we need to do the same for
$g(k', \cdot \,)$ given that this PMF comes into play in the computation
of the acceptance probabilities (see, e.g., Algorithm~\ref{algo_RJ_andrieu_2018}). We thus need parameter proposals
$\mathbf{y}_{l}^{(T-1,1)}, \ldots , \mathbf{y}_{l}^{(T-1,N)}$ for all Models
$l\in \mathbf{N}(k)$, and also for all models belonging to
$\mathbf{N}(k')$. The ratios $r_{\text{RJ2}}$ are next computed.

There are several ways to combine these ratios with
$\widehat{\pi }(l\mid \mathbf{D}_{n}) / \widehat{\pi }(k\mid \mathbf{D}_{n})$
and
$\widehat{\pi }(s\mid \mathbf{D}_{n}) / \widehat{\pi }(k'\mid
\mathbf{D}_{n})$ to improve the estimation of
$\pi (l\mid \mathbf{D}_{n}) / \pi (k\mid \mathbf{D}_{n})$ and
$\pi (s\mid \mathbf{D}_{n}) / \pi (k'\mid \mathbf{D}_{n})$. We define the
improved version of the PMF $g$ as follows to reflect this flexibility:
\begin{align}
\label{eqn_g_imp}
g_{\text{imp.}}(k, k', \mathbf{x}_{k \mapsto \bullet }^{(0:T-1,
\bullet )}, \mathbf{y}_{\bullet }^{(0:T-1, \bullet )}) := h\left (
\frac{\tilde{\pi }(k'\mid \mathbf{D}_{n})}{\tilde{\pi }(k\mid \mathbf{D}_{n})}
\right ) \bigg / c_{\text{imp.}}(k),
\end{align}
where
\begin{align*}
\frac{\tilde{\pi }(k'\mid \mathbf{D}_{n})}{\tilde{\pi }(k\mid \mathbf{D}_{n})}
&:= \varrho \left (
\frac{\widehat{\pi }(k'\mid \mathbf{D}_{n})}{\widehat{\pi }(k\mid \mathbf{D}_{n})},
r_{\text{RJ2}}((k,\mathbf{x}_{k}^{(0)}),(k',\mathbf{y}_{k'}^{(T-1,1)})),
\ldots , r_{\text{RJ2}}((k,\mathbf{x}_{k}^{(0)}),(k',\mathbf{y}_{k'}^{(T-1,N)}))
\right ),
\end{align*}
$\mathbf{y}_{\bullet }^{(0:T-1, \bullet )}$ is the vector containing
$\mathbf{y}_{l}^{(0, j)}, \ldots , \mathbf{y}_{l}^{(T-1, j)}$ for all
$j\in \{1,\ldots ,N\}$ and $l\in \mathbf{N}(k)$, and
$c_{k}^{\text{imp.}}$ is the normalizing constant. $\varrho $ is a function
aiming at putting together the information; its choice is discussed below.
$\mathbf{x}_{k \mapsto \bullet }^{(0:T-1, \bullet )}$ is the vector containing
$\mathbf{x}_{k \mapsto l}^{(0, j)}, \ldots , \mathbf{x}_{k \mapsto l}^{(T-1,
j)}$ for all $j\in \{1,\ldots ,N\}$ and $l\in \mathbf{N}(k)$. Here we added
the subscript $k \mapsto l$ and superscript $(t, j)$ to highlight that
the sequence generated using Step 2.(b) of Algorithm~\ref{algo_RJ_andrieu_2013} depends on $j$ (different sequences are
generated for different $j$), and also on $l$ through
$\rho _{k \mapsto l}^{(t)}$ \eqref{eqn_def_rho}. Note that
$\tilde{\pi }(l\mid \mathbf{D}_{n}) / \tilde{\pi }(k\mid \mathbf{D}_{n})$
is an estimator of
$\pi (l\mid \mathbf{D}_{n}) / \pi (k\mid \mathbf{D}_{n})$ which is in fact
a function of
$\mathbf{x}_{k \mapsto l}^{(0, j)}, \ldots , \mathbf{x}_{k \mapsto l}^{(T-1,
j)}$ and
$\mathbf{y}_{l}^{(0, j)}, \ldots , \mathbf{y}_{l}^{(T-1, j)}$ for all
$j\in \{1,\ldots ,N\}$ additionally to $k$ and $l$; we used this notation
to simplify and make the connection with
$\widehat{\pi }(l\mid \mathbf{D}_{n}) / \widehat{\pi }(k\mid \mathbf{D}_{n})$.

Algorithm~\ref{algo_RJ_imp_model} includes the idea of improving
$\widehat{\pi }(l\mid \mathbf{D}_{n}) / \widehat{\pi }(k\mid \mathbf{D}_{n})$
using ratios $r_{\text{RJ2}}$ in a valid way (as indicated by Proposition~\ref{prop_inv_algo_model_imp} below). This algorithm is quite complicated,
but once a code for Algorithm~\ref{algo_RJ_andrieu_2018} has been written, we
can use the connections between Algorithm~\ref{algo_RJ_imp_model} and Algorithm~\ref{algo_RJ_andrieu_2018} to facilitate the coding of the former.
Another drawback of Algorithm~\ref{algo_RJ_imp_model} is that it requires to
perform the computations for $g_{\text{imp.}}(k', \cdot \,)$ even when
$k'=k$. This is because
$g_{\text{imp.}}(k, k', \mathbf{x}_{k \mapsto \bullet }^{(0:T-1,
\bullet )}, \mathbf{y}_{\bullet }^{(0:T-1, \bullet )})$ is different from
$g_{\text{imp.}}(k', k, \mathbf{y}_{k' \mapsto \bullet }^{(0:T-1,j^{*})}
$,
$\mathbf{z}_{\bullet }^{(0:T-1, j^{*})}, \mathbf{z}_{\bullet }^{(0:T-1,
\bullet )})$ (see Step 2.(i)), even when $k' = k$. At least, most of computations
for the two main steps (Steps 2.(i) and 2.(ii)) can be performed in parallel.
Again, computational overheads have to be taken into account. In our numerical
example, the improvement is not significant enough to offset the computational
cost.

\begin{algorithm}[H]
\caption{RJ additionally improving the model proposal distribution}
\label{algo_RJ_imp_model}
\begin{enumerate}
 \itemsep 0mm

  \item[1.]  Sample $u_{a}, u_{c}\sim \mathcal{U}(0, 1)$. If $u_{c}\leq 1/2$ go to Step 2.(i), otherwise go to Step 2.(ii).

  \item[2.(i)] Do:

  \begin{enumerate}
    \itemsep 0mm

    \item[(a)] For all $l\in \mathbf{N}(k)$, sample proposals as in Step 2.(b-i) of \autoref{algo_RJ_andrieu_2018}: from $\mathbf{x}_{k}$, sample $N$ proposals $\mathbf{y}_{l}^{(T-1,1)}, \ldots , \mathbf{y}_{l}^{(T-1,N)}$ as in Step 2.(b) of \autoref{algo_RJ_andrieu_2013}.

    \item[(b)] Compute $\tilde{\pi}(l\mid \mathbf{D}_{n})\big/\tilde{\pi}(k\mid \mathbf{D}_{n})$ for all $l\in \mathbf{N}(k)$ and $g_{\text{imp.}}(k, \cdot \,)$ as in \eqref{eqn_g_imp}.

    \item[(c)] Sample $k'\sim g_{\text{imp.}}(k, \cdot \,)$ and $j^{*}$ from a PMF such that $\mathbb {P}(J^{*}=j) \propto  r_{\text{RJ2}}((k,\mathbf{x}_{k}^{(0)}),(k',\mathbf{y}_{k'}^{(T-1,j)}))$, and compute $\bar{r}(k, k')$.

    \item[(d)] For all $s\in \mathbf{N}(k')\setminus \{k\}$, sample proposals as in Step 2.(b-ii) of \autoref{algo_RJ_andrieu_2018}: from $\mathbf{y}_{k'}^{(T-1,j^{*})}$, sample one path as in Step 2.(b) of \autoref{algo_RJ_andrieu_2013}. Denote the endpoint by $\mathbf{z}_{s}^{(0, j^{*})}$. From $\mathbf{z}_{s}^{(0, j^{*})}$,  sample $N - 1$ reverse paths again as in Step 2.(b) of \autoref{algo_RJ_andrieu_2013}, yielding $N - 1$ proposals for the parameters of Model $k'$ that we denote by $\mathbf{z}_{s \mapsto k'}^{(T-1, j)}$, $j \in \{1, \ldots , N\} \setminus \{j^{*}\}$.

    \item[(e)]  Compute $\tilde{\pi}(s\mid \mathbf{D}_{n})\big/\tilde{\pi}(k'\mid \mathbf{D}_{n})$ for all $s\in \mathbf{N}(k')\setminus \{k\}$. Compute $g_{\text{imp.}}(k', \cdot \,)$ using these and $(\tilde{\pi}(k'\mid \mathbf{D}_{n})\big/\tilde{\pi}(k\mid \mathbf{D}_{n}))^{-1}$ computed in (b).

    \item[(f)]  If
\begin{align*}
   u_{a}  &\leq \frac{g_{\text{imp.}}(k', k, \mathbf{y}_{k' \mapsto \bullet}^{(0:T-1,j^{*})}, \mathbf{z}_{\bullet}^{(0:T-1, j^{*})}, \mathbf{z}_{\bullet}^{(0:T-1, \bullet )})}{g_{\text{imp.}}(k, k', \mathbf{x}_{k \mapsto \bullet}^{(0:T-1, \bullet )},  \mathbf{y}_{\bullet}^{(0:T-1, \bullet )})} \, \bar{r}(k, k'),
  \end{align*}
set the next state of the chain to $(k',\mathbf{y}_{k'}^{(T-1,j^{*})})$. Otherwise, set it to $(k,  \mathbf{x}_{k})$.

  \end{enumerate}

  \item[2.(ii)]  Do:

  \begin{enumerate}
    \itemsep 0mm

    \item[(a)] For all $l\in \mathbf{N}(k)$, sample proposals as in Step 2.(b-ii) of \autoref{algo_RJ_andrieu_2018}: from $\mathbf{x}_{k}$, sample one path as in Step 2.(b) of \autoref{algo_RJ_andrieu_2013}. Denote the endpoint by $\mathbf{y}_{l}^{(T-1, 1)}$. From $\mathbf{y}_{l}^{(T-1, 1)}$,  sample $N - 1$ reverse paths again as in Step 2.(b) of \autoref{algo_RJ_andrieu_2013}, yielding $N - 1$ proposals for the parameters of Model $k$ that we denote by $\mathbf{y}_{l \mapsto k}^{(0, 2)}, \ldots , \mathbf{y}_{l \mapsto k}^{(0, N)}$.

     \item[(b)] Compute $\tilde{\pi}(l\mid \mathbf{D}_{n})\big/\tilde{\pi}(k\mid \mathbf{D}_{n})$ for all $l\in \mathbf{N}(k)$ and $g_{\text{imp.}}(k, \cdot \,)$ as in \eqref{eqn_g_imp}.

    \item[(c)] Sample $k'\sim g_{\text{imp.}}(k, \cdot \,)$ and compute $\bar{r}(k', k)^{-1}$.

     \item[(d)] For all $s\in \mathbf{N}(k')\setminus \{k\}$, sample proposals as in Step 2.(b-i) of \autoref{algo_RJ_andrieu_2018}: from $\mathbf{y}_{k'}^{(T-1, 1)}$, sample $N$ proposals $\mathbf{z}_{s}^{(0,1)}, \ldots , \mathbf{z}_{s}^{(0,N)}$ as in Step 2.(b) of \autoref{algo_RJ_andrieu_2013}.

     \item[(e)]  Compute $\tilde{\pi}(s\mid \mathbf{D}_{n})\big/\tilde{\pi}(k'\mid \mathbf{D}_{n})$ for all $s\in \mathbf{N}(k')\setminus \{k\}$. Compute $g_{\text{imp.}}(k', \cdot \,)$ using these and $(\tilde{\pi}(k'\mid \mathbf{D}_{n})\big/\tilde{\pi}(k\mid \mathbf{D}_{n}))^{-1}$ computed in (b).

     \item[(f)] If
\begin{align*}
   u_{a}  &\leq \frac{g_{\text{imp.}}(k', k, \mathbf{y}_{k' \mapsto \bullet}^{(0:T-1, 1)}, \mathbf{z}_{\bullet}^{(0:T-1, \bullet )})}{g_{\text{imp.}}(k, k', \mathbf{x}_{k \mapsto \bullet}^{(0:T-1, 1)},   \mathbf{y}_{\bullet}^{(0:T-1, 1)}, \mathbf{y}_{\bullet \mapsto k}^{(0:T-1, \bullet )})} \, \bar{r}(k', k)^{-1},
  \end{align*}
set the next state of the chain to $(k',\mathbf{y}_{k'}^{(T-1, 1)})$. Otherwise, set it to $(k,  \mathbf{x}_{k})$.

  \end{enumerate}

  \item[3.] Go to Step 1.
 \end{enumerate}
\end{algorithm}

\begin{Proposition}%
\label{prop_inv_algo_model_imp}
Under the two assumptions presented in Section~\ref{sec_andrieu_2013}, \eqref{eqn_symmetry}--\eqref{eqn_reversibility}, Algorithm~\ref{algo_RJ_imp_model} is valid.
\end{Proposition}

It is natural to set
$\tilde{\pi }(l\mid \mathbf{D}_{n}) / \tilde{\pi }(k\mid \mathbf{D}_{n})$
to 1 when $l = k$. This implies that we in fact do not need to sample proposals
for Model $k$ in the first parts of Steps 2.(i) and 2.(ii). Also, in the
second parts of Steps 2.(i) and 2.(ii), it is not required to sample proposals
for $s = k'$ for the same reason. When $k' = k$, we actually perform a
``vanilla'' parameter-update step (an HMC step in our numerical example).
In this case, the MH ratio as to be multiplied by
\begin{equation*}
\frac{g_{\text{imp.}}(k', k, \mathbf{y}_{k' \mapsto \bullet }^{(0:T-1,j^{*})}, \mathbf{z}_{\bullet }^{(0:T-1, j^{*})}, \mathbf{z}_{\bullet }^{(0:T-1, \bullet )})}{g_{\text{imp.}}(k, k', \mathbf{x}_{k \mapsto \bullet }^{(0:T-1, \bullet )}, \mathbf{y}_{\bullet }^{(0:T-1, \bullet )})}
\end{equation*}
or
\begin{equation*}
\frac{g_{\text{imp.}}(k', k, \mathbf{y}_{k' \mapsto \bullet }^{(0:T-1, 1)}, \mathbf{z}_{\bullet }^{(0:T-1, \bullet )})}{g_{\text{imp.}}(k, k', \mathbf{x}_{k \mapsto \bullet }^{(0:T-1, 1)}, \mathbf{y}_{\bullet }^{(0:T-1, 1)}, \mathbf{y}_{\bullet \mapsto k}^{(0:T-1, \bullet )})},
\end{equation*}
depending if Step 2.(i) or Step 2.(ii) is used.

The function $\varrho $ in \eqref{eqn_g_imp} specifies the way the information
is combined. It may be set for instance to the simple average:
\begin{align}
\label{eqn_est_ratio}
\frac{\tilde{\pi }(l\mid \mathbf{D}_{n})}{\tilde{\pi }(k\mid \mathbf{D}_{n})}
= \frac{1}{N + 1} \left (
\frac{\widehat{\pi }(l\mid \mathbf{D}_{n})}{\widehat{\pi }(k\mid \mathbf{D}_{n})}
+ \sum _{j=1}^{N} r_{\text{RJ2}}((k,\mathbf{x}_{k}^{(0)}),(l,
\mathbf{y}_{l}^{(T-1,j)}))\right ).
\end{align}
One may alternatively take the average of
$\widehat{\pi }(l\mid \mathbf{D}_{n}) / \widehat{\pi }(k\mid \mathbf{D}_{n})$
and $\bar{r}(k, l)$:
\begin{align*}
\frac{\tilde{\pi }(l\mid \mathbf{D}_{n})}{\tilde{\pi }(k\mid \mathbf{D}_{n})}
&= \frac{1}{2} \left (
\frac{\widehat{\pi }(l\mid \mathbf{D}_{n})}{\widehat{\pi }(k\mid \mathbf{D}_{n})}
+ \bar{r}(k, l)\right )
\cr
&= \frac{1}{2} \left (
\frac{\widehat{\pi }(l\mid \mathbf{D}_{n})}{\widehat{\pi }(k\mid \mathbf{D}_{n})}
+ \frac{1}{N} \sum _{j=1}^{N} r_{\text{RJ2}}((k,\mathbf{x}_{k}^{(0)}),(l,
\mathbf{y}_{l}^{(T-1,j)}))\right ).
\end{align*}
These reflect a choice of putting more or less weight on
$\widehat{\pi }(l\mid \mathbf{D}_{n}) / \widehat{\pi }(k\mid \mathbf{D}_{n})$.
We know that if $T$ and $N$ are large enough then $\bar{r}(k, l)$ is close
to $\pi (l\mid \mathbf{D}_{n}) / \pi (k\mid \mathbf{D}_{n})$, which may
not be the case for
$\widehat{\pi }(l\mid \mathbf{D}_{n}) / \widehat{\pi }(k\mid \mathbf{D}_{n})$
when $n$ is not sufficiently large. The latter ratio may thus act as outlying/conflicting
information against which these averages above are not robust. A robust
approach consists in setting $\varrho $ to be the median of
$\widehat{\pi }(l\mid \mathbf{D}_{n}) / \widehat{\pi }(k\mid \mathbf{D}_{n}),
r_{\text{RJ2}}((k,\mathbf{x}_{k}^{(0)}),(l,\mathbf{y}_{l}^{(T-1,1)})) $,
$\ldots , r_{\text{RJ2}}((k,\mathbf{x}_{k}^{(0)}),(l,\mathbf{y}_{l}^{(T-1,N)}))$.
We recommend this approach and use it in our numerical example.

Furthermore, as $T,N\longrightarrow \infty $,
$\tilde{\pi }(l\mid \mathbf{D}_{n})/\tilde{\pi }(k\mid \mathbf{D}_{n})$ is
a consistent estimator of
$\pi (l\mid \mathbf{D}_{n}) / \pi (k\mid \mathbf{D}_{n})$ for fixed
$n$ when the median or \eqref{eqn_est_ratio} is used (recall the properties
of $r_{\text{RJ2}}$ and $\bar{r}$ mentioned in Sections \ref{sec_andrieu_2013} and \ref{sec_andrieu_2018}). Therefore, if the function
$h$ is such that $h(x) = x \, h(1/x)$ for $x>0$, then the acceptance probabilities
in Algorithm~\ref{algo_RJ_imp_model} converge towards
$1 \wedge c(k) / c(k')$, where $c(k)$ and $c(k')$ are the normalizing constants
with
\begin{equation*}
c(k) = \sum _{l \in \mathbf{N}(k)} h\left (
\frac{\pi (l\mid \mathbf{D}_{n})}{\pi (k\mid \mathbf{D}_{n})}\right ).
\end{equation*}
In fact, the same technique as in the proof of Theorem 1 in
\cite{gagnon2019NRJ} allows to prove that a Markov chain simulated by Algorithm~\ref{algo_RJ_imp_model} converges weakly for fixed $n$ to that of an
ideal RJ which has access to the unnormalized version of the posterior probabilities
$\pi (k\mid \mathbf{D}_{n})$ and is able to sample from the conditional
distributions $\pi (\, \cdot \mid k,\mathbf{D}_{n})$ (and for which the
acceptance probabilities are $1 \wedge c(k) / c(k')$), with its good mixing
properties as discussed in Section~\ref{sec_objectives}.

\section{Details of the variable-selection example}
\label{sec_supp_variable_selection}

We review linear regression and introduce notation in Section~\ref{sec_linear_reg}. We present in Section~\ref{sec_normal_reg} all the
details to compute estimates for the normal linear-regression model. Some
are useful for the computation of algorithm inputs, such as the observed
information matrix, and to verify the validity of our RJ code (as mentioned
in Section~\ref{sec_context_example}). In Section~\ref{sec_robust_reg}, we turn
to the details of the robust linear regressions and the computation of
algorithm inputs.

\subsection{Linear regression}
\label{sec_linear_reg}

We first introduce notation. We define
$\gamma _{1}, \ldots , \gamma _{n} \in \mathbb{R} $ to be $n$ data points
from the dependent variable, and
$\boldsymbol{\gamma }_{n} := (\gamma _{1}, \ldots , \gamma _{n})^{T}$. We
denote the full design matrix containing $n$ observations from all covariates
by $\mathbf{C} \in \mathbb{R} ^{n \times p}$, where $p$ is a positive integer.
For simplicity, we refer to the first column of $\mathbf{C}$ as the
\textit{first} covariate even if, as usual, it is a column of $1$'s. The
design matrix associated with Model $k$ whose columns form a subset of
$\mathbf{C}$ is denoted by $\mathbf{C}_{k}$, with lines denoted by
$\mathbf{c}_{i,k}^{T}$. We use $d_{k}$ to denote the number of covariates
in Model $k$; we therefore slightly abuse notation given that the number
of parameters for Model $k$ is $d_{k}+1$ (one regression coefficient per
covariate plus the scale parameter of the error term).

The linear regression is
\begin{equation*}
\gamma _{i} = \mathbf{c}_{i, K}^{T} \, \boldsymbol{\beta }_{K} +
\epsilon _{i, K}, \quad i = 1, \ldots , n, \quad K \in \mathcal{K},
\end{equation*}
where $K$ is the model indicator, $\boldsymbol{\beta }_{K}$ is the random
vector containing the regression coefficients of Model $K$ and
$\epsilon _{1,K}, \ldots , \epsilon _{n, K} \in \mathbb{R} $ are the random
errors of Model $K$. We assume that
$\epsilon _{1,K}, \ldots , \epsilon _{n, K}$ and
$\boldsymbol{\beta }_{K}$ are $n+1$ conditionally independent random variables
given $(K, \sigma _{K})$, with $\sigma _{K}>0$ being the scale parameter
of the errors of Model $K$. The conditional density of
$\epsilon _{i, K}$ is given by
\begin{equation*}
\epsilon _{i, K} \mid K, \sigma _{K}, \boldsymbol{\beta }_{K}
\stackrel{d}{=} \epsilon _{i, K} \mid K, \sigma _{K}
\stackrel{d}{\sim } (1 / \sigma _{K}) f(\epsilon _{i, K} / \sigma _{K}),
\quad i = 1, \ldots , n.
\end{equation*}

\subsection{Normal linear regression}
\label{sec_normal_reg}

We present in this section a result giving the precise form of the joint
posterior density for the normal linear regression.

\begin{Proposition}%
\label{prop_posterior_normal}
If $f = \varphi (\, \cdot \,; 0, 1)$ and
$\pi (\boldsymbol{\beta }_{k}, \sigma _{k} \mid k) \propto 1 / \sigma _{k}$,
then
\begin{gather}
\label{eqn_post_k_reg}
\pi (k \mid \boldsymbol{\gamma }_{n}) \propto \pi (k) \,
\frac{\Gamma ((n - d_{k}) / 2) \, \pi ^{d_{k} / 2}}{\|\boldsymbol{\gamma }_{n} - \widehat{\boldsymbol{\gamma }}_{k}\|_{2}^{n - d_{k}} \, |\mathbf{C}_{k}^{T} \mathbf{C}_{k}|^{1/2}},
\\
\pi (\sigma _{k} \mid k, \boldsymbol{\gamma }_{n}) =
\frac{2^{1-\frac{n-d_{k}}{2}} \|\boldsymbol{\gamma }_{n} - \widehat{\boldsymbol{\gamma }}_{k}\|_{2}^{n - d_{k}}}{\Gamma ((n-d_{k})/2) \, \sigma _{k}^{n-d_{k}+1}}
\, \exp \left \{  -\frac{1}{2\sigma _{k}^{2}}\, \|\boldsymbol{\gamma }_{n}
- \widehat{\boldsymbol{\gamma }}_{k}\|_{2}^{2}\right \}  ,
\nonumber
\end{gather}
and
\begin{equation*}
\boldsymbol{\beta }_{K} \mid K, \sigma _{K}, \boldsymbol{\gamma }_{n}
\sim \varphi (\, \cdot \,; \widehat{\boldsymbol{\beta }}_{k}, \sigma _{K}^{2}
(\mathbf{C}_{K}^{T} \mathbf{C}_{K})^{-1}),
\end{equation*}
where
$\widehat{\boldsymbol{\gamma }}_{k} := \mathbf{C}_{k} (\mathbf{C}_{k}^{T}
\mathbf{C}_{k})^{-1} \mathbf{C}_{k}^{T} \boldsymbol{\gamma }_{n}$,
$\widehat{\boldsymbol{\beta }}_{k} := (\mathbf{C}_{K}^{T} \mathbf{C}_{K})^{-1}
\mathbf{C}_{K}^{T} \boldsymbol{\gamma }_{n}$, and $\| \cdot \|_{2}$ is the
Euclidian norm. Note that the normalization constant of
$\pi (k \mid \boldsymbol{\gamma }_{n})$ is the sum over $k$ of the expression
on the right-hand side (RHS) in \eqref{eqn_post_k_reg}.
\end{Proposition}

Note that $\sigma _{K}^{2} \mid K, \boldsymbol{\gamma }_{n}$ has an inverse-gamma
distribution with shape and scale parameters given by
$(n - d_{k}) / 2$ and
$\|\boldsymbol{\gamma }_{n} - \widehat{\boldsymbol{\gamma }}_{k}\|_{2}^{2} /
2$, respectively. We work on the log-scale for the scale parameters so
that all the parameters take values on the real line, and presumably, have
densities that are closer to the bell curve. We thus define
$\eta _{k} := \log \sigma _{k}$. The associated conditional distribution
is given by
\begin{equation*}
\pi (\eta _{k} \mid k, \boldsymbol{\gamma }_{n}) :=
\frac{2^{1-\frac{n-d_{k}}{2}} \|\boldsymbol{\gamma }_{n} - \widehat{\boldsymbol{\gamma }}_{k}\|_{2}^{n - d_{k}}}{\Gamma ((n-d_{k})/2) \, \mathrm{e}^{(n - d_{k}) \eta _{k}}}
\, \exp \left \{  -\frac{1}{2 \mathrm{e}^{2 \eta _{k}}}\, \|
\boldsymbol{\gamma }_{n} - \widehat{\boldsymbol{\gamma }}_{k}\|_{2}^{2}
\right \}  .
\end{equation*}
The other terms in the joint posterior do not change except that we now
consider that
$\boldsymbol{\beta }_{K} \mid K, \eta _{K}, \boldsymbol{\gamma }_{n}
\sim \varphi (\, \cdot \,; \widehat{\boldsymbol{\beta }}_{k}, \exp (2
\eta _{K}) (\mathbf{C}_{K}^{T} \mathbf{C}_{K})^{-1})$.

The maximizers of the conditional posterior densities given that
$K =k$ are:
\begin{equation*}
\hat{\mathbf{x}}_{k} = (\widehat{\boldsymbol{\beta }}_{k},
\widehat{\eta }_{k}),
\end{equation*}
where
\begin{equation*}
\widehat{\eta }_{k} := \log \sqrt{\frac{1}{n} \, \|\boldsymbol{\gamma }_{n}
- \widehat{\boldsymbol{\gamma }}_{k}\|_{2}^{2}}.
\end{equation*}
The observed information matrix evaluated at the maximizers is thus given
by
\begin{equation*}
\hat{\mathcal{I}}_{k} = \left (
\begin{array}{cc}
\mathbf{C}_{k}^{T} \mathbf{C}_{k} / \mathrm{e}^{2 \hat{\eta }_{k}} &
\mathbf{0}
\cr
\mathbf{0} & 2n
\end{array}
\right ).
\end{equation*}
Then,
\begin{align}
\label{eqn_inverse_fisher}
\hat{\mathcal{I}}_{k}^{-1} = \left (
\begin{array}{cc}
\mathrm{e}^{2 \hat{\eta }_{k}}(\mathbf{C}_{k}^{T} \mathbf{C}_{k})^{-1}
& \mathbf{0}
\cr
\mathbf{0} & (2n)^{-1}
\end{array} \right ).
\end{align}

Relying on improper priors such as
$\pi (\boldsymbol{\beta }_{k}, \sigma _{k} \mid k) \propto 1 / \sigma _{k}$
may lead to inconsistencies in model selection (see, e.g.,
\cite{casella2009consistency}). When this problem happens, the phenomenon
is referred to as the \textit{Jeffreys-Lindley paradox} (\cite{lindley1957paradox}
and \cite{jeffreys1967prob}) in the literature. It can be shown that the
Jeffreys-Lindley paradox does not arise under the framework of normal linear
regression described above when the prior distribution on $K$ is carefully
chosen, and more precisely, set to
\begin{equation*}
\pi (k) \propto |\mathbf{C}_{k}^{T} \mathbf{C}_{k}|^{1/2} / n^{d_{k} /
2}.
\end{equation*}
See \cite{gagnon2017PCR} for an analogous proof in a framework of principal
component regression.

When the covariates are orthonormal,
$|\mathbf{C}_{k}^{T} \mathbf{C}_{k}|^{1 / 2} = |n \mathbf{I}_{d_{k}}|^{1
/ 2} = n^{d_{k} / 2}$ (if all covariate observations are standardized using
a standard deviation in which the divisor is $n$). In this case,
$\pi (k) \propto |\mathbf{C}_{k}^{T} \mathbf{C}_{k}|^{1/2} / n^{d_{k} /
2} = 1$. The proposed prior on $K$ can thus be seen as a relative adjustment
of the volume spanned by the columns of
$\mathbf{C}_{K}^{T} \mathbf{C}_{K}$.

\subsection{Robust linear regression}
\label{sec_robust_reg}

In this section, we consider that the density $f$ is a LPTN with parameter
$\rho \in (2\Phi (1) - 1, 1) \approx (0.6827, 1)$, i.e.
\begin{equation}
\label{eqn_log_pareto_pcr}
f(x)=\left \{
\begin{array}{lcc}
\varphi (x; 0, 1) & \text{ if } & \left |x\right |\leq \tau ,
\\
\varphi (\tau ; 0, 1)\,\frac{\tau }{|x|}\left (
\frac{\log \tau }{\log |x|}\right )^{\lambda +1} & \text{ if } & \left |x
\right |>\tau ,
\\
\end{array}
\right .
\end{equation}
where $x\in \mathbb{R} $ and $\Phi (\,\cdot \,)$ is the cumulative distribution
function (CDF) of a standard normal. The terms $\tau >1$ and
$\lambda >0$ are functions of $\rho $ and satisfy
\begin{align*}%
& \tau = \{\tau : \mathbb{P}(-\tau \leq Z \leq \tau )= \rho \,
\text{ for }\, Z \sim \varphi (\, \cdot \,; 0, 1)\},
\\
& \lambda = 2(1-\rho )^{-1}\varphi (\tau ; 0, 1) \, \tau \log (\tau ).
\end{align*}
Setting $\rho $ to $0.95$ has proved to be suitable for practical purposes
(see \cite{gagnon2018regression}). Accordingly, this is the value that
is used in our numerical example.

Using the same prior as the normal regression, the joint posterior density
is:
\begin{align*}
\pi (k, \boldsymbol{\beta }_{k}, \sigma _{k} \mid \boldsymbol{\gamma }_{n})
\propto
\frac{|\mathbf{C}_{k}^{T} \mathbf{C}_{k}|^{1/2}}{n^{d_{k} / 2}}
\frac{1}{\sigma _{k}} \prod _{i = 1}^{n} \frac{1}{\sigma _{k}} f
\left (
\frac{\gamma _{i} - \mathbf{c}_{i, k}^{T} \boldsymbol{\beta }_{k}}{\sigma _{k}}
\right ).
\end{align*}

With the change of variable $\eta _{k} = \log \sigma _{k}$, we have
\begin{align*}
\pi (k, \boldsymbol{\beta }_{k}, \eta _{k} \mid \boldsymbol{\gamma }_{n})
\propto
\frac{|\mathbf{C}_{k}^{T} \mathbf{C}_{k}|^{1/2}}{n^{d_{k} / 2}}
\frac{1}{\mathrm{e}^{\eta _{k} n}} \prod _{i = 1}^{n} f\left (
\frac{\gamma _{i} - \mathbf{c}_{i, k}^{T} \boldsymbol{\beta }_{k}}{\mathrm{e}^{\eta _{k}}}
\right ).
\end{align*}

Recall that for the robust models we set
$\hat{\mathcal{I}}_{k}^{-1}$ as in \eqref{eqn_inverse_fisher} but use the
maximizer $\hat{\eta }_{k}$ of the function above. We need log conditionals
and their gradients for optimizers and HMC:
\begin{equation*}
\log \pi (\boldsymbol{\beta }_{k}, \eta _{k} \mid k, \boldsymbol{\gamma }_{n})
\propto - n \eta _{k} + \sum _{i = 1}^{n} \log f\left (
\frac{\gamma _{i} - \mathbf{c}_{i, k}^{T} \boldsymbol{\beta }_{k}}{\mathrm{e}^{\eta _{k}}}
\right ),
\end{equation*}
where the proportional sign has to be understood in the original scale,
and
\begin{align*}
&\frac{\partial }{\partial \boldsymbol{\beta }_{k}} \log \pi (
\boldsymbol{\beta }_{k}, \eta _{k} \mid k, \boldsymbol{\gamma }_{n}) =
\sum _{i = 1}^{n} \mathrm{e}^{-2 \eta _{k}} (\gamma _{i} -
\mathbf{c}_{i, k}^{T} \boldsymbol{\beta }_{k}) \mathbf{c}_{i, k}
\mathds{1}\left ((\gamma _{i} - \mathbf{c}_{i, k}^{T}
\boldsymbol{\beta }_{k}) / \mathrm{e}^{\eta _{k}} \leq \tau \right )
\cr
& + \left [
\frac{\text{sgn}(\gamma _{i} - \mathbf{c}_{i, k}^{T} \boldsymbol{\beta }_{k}) \mathbf{c}_{i, k}}{|\gamma _{i} - \mathbf{c}_{i, k}^{T} \boldsymbol{\beta }_{k}|}
+ (\lambda + 1)
\frac{\text{sgn}(\gamma _{i} - \mathbf{c}_{i, k}^{T} \boldsymbol{\beta }_{k}) \mathbf{c}_{i, k}}{|\gamma _{i} - \mathbf{c}_{i, k}^{T} \boldsymbol{\beta }_{k}| \log \left ((\gamma _{i} - \mathbf{c}_{i, k}^{T} \boldsymbol{\beta }_{k}) / \mathrm{e}^{\eta _{k}}\right )}
\right ] \mathds{1}\left ((\gamma _{i} - \mathbf{c}_{i, k}^{T}
\boldsymbol{\beta }_{k}) / \mathrm{e}^{\eta _{k}} > \tau \right ),
\end{align*}
and
\begin{align*}
\frac{\partial }{\partial \eta _{k}} \log \pi (\boldsymbol{\beta }_{k},
\eta _{k} \mid k, \boldsymbol{\gamma }_{n})
 &= -n + \sum _{i = 1}^{n} \mathrm{e}^{-2 \eta _{k}} (\gamma _{i}
- \mathbf{c}_{i, k}^{T} \boldsymbol{\beta }_{k})^{2} \mathds{1}\left ((
\gamma _{i} - \mathbf{c}_{i, k}^{T} \boldsymbol{\beta }_{k}) /
\mathrm{e}^{\eta _{k}} \leq \tau \right )
\cr
& \qquad + \left [1 + (\lambda + 1)
\frac{1}{ \log \left ((\gamma _{i} - \mathbf{c}_{i, k}^{T} \boldsymbol{\beta }_{k}) / \mathrm{e}^{\eta _{k}}\right )}
\right ] \mathds{1}\left ((\gamma _{i} - \mathbf{c}_{i, k}^{T}
\boldsymbol{\beta }_{k}) / \mathrm{e}^{\eta _{k}} > \tau \right ),
\end{align*}
$\text{sgn}(\, \cdot \,)$ being the sign function.

To use the annealing distributions in Algorithms \ref{algo_RJ_andrieu_2013}, \ref{algo_RJ_andrieu_2018} and \ref{algo_RJ_imp_model}, we work with the log densities; therefore we simply
multiply $\log \pi (\, \cdot \mid k, \boldsymbol{\gamma }_{n})$ by
$1 - t / T$ and
$\log \pi (\, \cdot \mid k', \boldsymbol{\gamma }_{n})$ by $t / T$ to obtain
$\log \rho _{k \mapsto k'}^{(t)}$. To use MALA proposal distributions,
we however need to compute the gradient of
$\log \rho _{k \mapsto k'}^{(t)}$. We now compute the gradient with respect
to $\mathbf{x}_{k}^{(t)}$; the gradient with respect to
$\mathbf{y}_{k'}^{(t)}$ is computed similarly. The proportional sign ``$
\propto $'' is with respect to everything that are not the parameters:
\begin{align*}
\pi (\mathbf{x}_{k}^{(t)} \mid k, \mathbf{D}_{n})^{1 - t / T} q_{k'
\mapsto k}(\mathbf{x}_{k}^{(t)})^{t / T}
&= \left [ \pi (\boldsymbol{\beta }_{k}, \eta _{k} \mid k,
\boldsymbol{\gamma }_{n})\right ]^{1 - t/T}
\cr
&\qquad \times \left [
\frac{|\mathbf{C}_{k}^{T} \mathbf{C}_{k}|^{1/2}}{(2\pi )^{d_{k}/2} \mathrm{e}^{d_{k} \widehat{\eta }_{k}}}
\exp \left (-\frac{1}{2\mathrm{e}^{2\widehat{\eta }_{k}}}(
\boldsymbol{\beta }_{k} - \widehat{\boldsymbol{\beta }}_{k})^{T} (
\mathbf{C}_{k}^{T} \mathbf{C}_{k})(\boldsymbol{\beta }_{k} -
\widehat{\boldsymbol{\beta }}_{k}) \right )\right ]^{t/T}
\cr
& \qquad \times \left [\frac{1}{\sqrt{2\pi (1/(2n))}} \exp \left (-
\frac{1}{2(1/(2n))}(\eta _{k} - \widehat{\eta }_{k})^{2}\right )
\right ]^{t/T}
\cr
&\propto \left [ \pi (\boldsymbol{\beta }_{k}, \eta _{k} \mid k,
\boldsymbol{\gamma }_{n})\right ]^{1 - t/T}
\cr
&\qquad \times \frac{1}{\mathrm{e}^{d_{k} (t/T) \widehat{\eta }_{k}}}
\exp \left (-\frac{(t / T)}{2 \mathrm{e}^{2\widehat{\eta }_{k}}}(
\boldsymbol{\beta }_{k} - \widehat{\boldsymbol{\beta }}_{k})^{T} (
\mathbf{C}_{k}^{T} \mathbf{C}_{k})(\boldsymbol{\beta }_{k} -
\widehat{\boldsymbol{\beta }}_{k}) \right )
\cr
&\qquad \times \exp \left (- n(t / T)(\eta _{k} - \widehat{\eta }_{k})^{2}
\right ),
\end{align*}
where we omitted the superscript ``(t)'' for the variables to simplify.
Therefore,
\begin{align*}
&\frac{\partial }{\partial \boldsymbol{\beta }_{k}} \log \pi (\mathbf{x}_{k}^{(t)}
\mid k, \mathbf{D}_{n})^{1 - t / T} q_{k' \mapsto k}(\mathbf{x}_{k}^{(t)})^{t
/ T}
\cr
&\qquad = (1 - t/T) \frac{\partial }{\partial \boldsymbol{\beta }_{k}}
\log \pi (\boldsymbol{\beta }_{k}, \eta _{k} \mid k, \boldsymbol{\gamma }_{n})
- (t/T) \mathrm{e}^{-2\widehat{\eta }_{k}}(\mathbf{C}_{k}^{T}
\mathbf{C}_{k})(\boldsymbol{\beta }_{k} - \widehat{\boldsymbol{\beta }}_{k}),
\end{align*}
and
\begin{align*}
&\frac{\partial }{\partial \eta _{k}} \log \pi (\mathbf{x}_{k}^{(t)}
\mid k, \mathbf{D}_{n})^{1 - t / T} q_{k' \mapsto k}(\mathbf{x}_{k}^{(t)})^{t
/ T}
\cr
& \qquad = (1 - t/T) \frac{\partial }{\partial \eta _{k}} \log \pi (
\boldsymbol{\beta }_{k}, \eta _{k} \mid k, \boldsymbol{\gamma }_{n}) - 2 n
(t/T) (\eta _{k} - \widehat{\eta }_{k}).
\end{align*}

\section{Proofs}
\label{sec_proofs}

We present the proofs of Theorem~\ref{thm_convergence}, Proposition~\ref{prop_peskun}, Proposition~\ref{prop_inv_algo_model_imp} and Proposition~\ref{prop_posterior_normal} in that order.

\begin{proof}[Proof of Theorem~\ref{thm_convergence}]
We here only prove that
$\{(K, \mathbf{Z}_{K})_{n}(m): m \in \mathbb{N} \} \Longrightarrow \{(K,
\mathbf{Z}_{K})_{\text{limit}}(m): m \in \mathbb{N} \}$ in probability as the
proof of
$\{(K, \mathbf{Z}_{K})_{\text{ideal}}(m): m \in \mathbb{N} \}
\Longrightarrow \{(K, \mathbf{Z}_{K})_{\text{limit}}(m): m \in
\mathbb{N} \}$ is similar. To prove this result, we use Theorem 2 of
\cite{schmon2018large}. We thus have to verify the following three conditions.

\begin{enumerate}
\item
$(K,\mathbf{Z}_{K})_{n}(0)\Longrightarrow (K,\mathbf{Z}_{K})_{
\text{limit}}(0)$ in probability as $n\longrightarrow \infty $.
\item Use $P_{n}$ and $P_{\text{limit}}$ to denote the transition kernels
of $\{(K,\mathbf{Z}_{K})_{n}(m): m\in \mathbb{N} \}$ and
$\{(K,\mathbf{Z}_{K})_{\text{limit}}(m): m\in \mathbb{N} \}$, respectively.
These are such that
\begin{align}
\label{eqn_step2}
\sum _{k} \int \left |P_{n} \phi (k, \mathbf{z}_{k}) - P_{
\text{limit}} \phi (k, \mathbf{z}_{k})\right | \, \pi _{K, \mathbf{Z}_{K}}(k,
\mathbf{z}_{k} \mid \mathbf{D}_{n}) \, \mathrm{d}\mathbf{z}_{k}
\longrightarrow 0,
\end{align}
in probability as $n\longrightarrow \infty $ for all
$\phi \in \text{BL}$, where
$\pi _{K, \mathbf{Z}_{K}}(\, \cdot \,, \cdot \mid \mathbf{D}_{n})$ denotes
the stationary distribution/density of
$\{(K, \mathbf{Z}_{K})_{n}(m): m \in \mathbb{N} \}$ and $\text{BL}$ denotes
the set of bounded Lipschitz functions.
\item The transition kernel $P_{\text{limit}}$ is such that
$P_{\text{limit}} \phi $ is continuous for any
$\phi \in \mathcal{C}_{b}$ (the set of continuous bounded functions).
\end{enumerate}

We start with Step 1. Given that $\mathcal{K}$ is finite (Assumption~\ref{ass1}),
it suffices to verify that
\begin{equation*}
\left |\pi (k\mid \mathbf{D}_{n})\,\mathbb{P}(\mathbf{Z}_{k,n} \in A)
- \bar{\pi }(k)\,\mathbb{P}(\mathbf{Z}_{k,\text{limit}} \in A) \right |
\longrightarrow 0 \quad \text{in probability},
\end{equation*}
for any $k$ and measurable set $A$, where
$\mathbb{P}(\mathbf{Z}_{k,n} \in A)$ and
$\mathbb{P}(\mathbf{Z}_{k,\text{limit}} \in A)$ are computed using the
conditional distributions given that $K=k$. Using the triangle inequality,
we have that
\begin{align*}
& \left |\pi (k\mid \mathbf{D}_{n})\,\mathbb{P}(\mathbf{Z}_{k,n}
\in A) - \bar{\pi }(k)\,\mathbb{P}(\mathbf{Z}_{k,\text{limit}} \in A)
\right |
\cr
&\quad \leq \left |\pi (k\mid \mathbf{D}_{n})\,\mathbb{P}(\mathbf{Z}_{k,n}
\in A) - \bar{\pi }(k)\,\mathbb{P}(\mathbf{Z}_{k,n} \in A) \right |
\cr
&\qquad + \left | \bar{\pi }(k)\,\mathbb{P}(\mathbf{Z}_{k,n} \in A) -
\bar{\pi }(k)\,\mathbb{P}(\mathbf{Z}_{k,\text{limit}} \in A) \right |.
\end{align*}
We now show that both absolute values converge towards 0 in probability
which will allow to conclude by Slutsky's theorem and monotonicity of probabilities.
We first have that
\begin{equation*}
\left |\pi (k\mid \mathbf{D}_{n})\,\mathbb{P}(\mathbf{Z}_{k,n} \in A)
- \bar{\pi }(k)\,\mathbb{P}(\mathbf{Z}_{k,n} \in A) \right |\leq
\left |\pi (k\mid \mathbf{D}_{n}) - \bar{\pi }(k)\right |
\longrightarrow 0,
\end{equation*}
in probability by Assumption~\ref{ass2} and the fact that
$\mathbb{P}(\mathbf{Z}_{k,n} \in A)\leq 1$. Using now that
$\bar{\pi }(k)\leq 1$, we have that
\begin{align*}
&\left | \bar{\pi }(k)\,\mathbb{P}(\mathbf{Z}_{k,n} \in A) -
\bar{\pi }(k)\,\mathbb{P}(\mathbf{Z}_{k,\text{limit}} \in A)\right |
\cr
&\qquad \leq \left | \mathbb{P}(\mathbf{Z}_{k,n} \in A) -
\mathbb{P}(\mathbf{Z}_{k,\text{limit}} \in A)\right |
\cr
&\qquad =\left | \mathbb{P}(\mathbf{X}_{k,n} \in A_{n}) -
\mathbb{P}(\mathbf{X}_{k,\text{limit}} \in A_{n})\right |
\cr
&\qquad \leq \int \left |\pi (\mathbf{x}_{k}\mid k, \mathbf{D}_{n}) -
\varphi (\mathbf{x}_{k}; \hat{\mathbf{x}}_{k}, \boldsymbol{\Sigma }_{k}/n)
\right | \, \mathrm{d}\mathbf{x}_{k} \longrightarrow 0,
\end{align*}
in probability, by Jensen's inequality and Assumption~\ref{ass3}, where
$A_{n}$ is the set $A$ after applying the inverse transformation to retrieve
the original random variables. Note that in the last inequality, we also
used that $A_{n}\subset \mathbb{R} ^{d_{k}}$.

We continue with Step 2 and prove \eqref{eqn_step2}. We have that
\begin{align*}
P_{\text{limit}}((k, \mathbf{z}_{k}), (k', \mathrm{d}\mathbf{y}_{k'})) &=
g_{\text{limit}}(k, k') \, \alpha _{\text{limit}}(k, k') \,\varphi (
\mathrm{d}\mathbf{y}_{k'}; \mathbf{0}, \boldsymbol{\Sigma }_{k'})
\cr
&\quad + \delta _{(k, \mathbf{z}_{k})}(k', \mathrm{d}\mathbf{y}_{k'})
\sum _{l \in \mathbf{N}(k)} g_{\text{limit}}(k, l) \, (1 - \alpha _{
\text{limit}}(k, l)),
\end{align*}
where
\begin{equation*}
g_{\text{limit}}(k, k') :=
\frac{h\left (\frac{\bar{\pi }(k')}{\bar{\pi }(k)}\right )}{\sum _{l \in \mathbf{N}(k)} h\left (\frac{\bar{\pi }(l)}{\bar{\pi }(k)}\right )},
\end{equation*}
and
\begin{equation*}
\alpha _{\text{limit}}(k, k') := 1 \wedge
\frac{\bar{\pi }(k') \, g_{\text{limit}}(k', k)}{\bar{\pi }(k) \, g_{\text{limit}}(k, k')}.
\end{equation*}

By definition,
\begin{align*}
P_{\text{limit}} \phi (k, \mathbf{z}_{k})&= \sum _{k' \in \mathbf{N}(k)}
\int \phi (k', \mathbf{y}_{k'}) \, P_{\text{limit}}((k, \mathbf{z}_{k}),
(k', \mathrm{d}\mathbf{y}_{k'}))
\cr
&=\sum _{k' \in \mathbf{N}(k)}\int \phi (k', \mathbf{y}_{k'}) \, g_{
\text{limit}}(k, k') \, \alpha _{\text{limit}}(k, k') \, \varphi (
\mathbf{y}_{k'}; \mathbf{0}, \boldsymbol{\Sigma }_{k'}) \, \mathrm{d}
\mathbf{y}_{k'}
\cr
&\qquad + \phi (k, \mathbf{z}_{k}) \sum _{l \in \mathbf{N}(k)} g_{
\text{limit}}(k, l) \, (1 - \alpha _{\text{limit}}(k, l)).
\end{align*}
$P_{\text{limit}} \phi $ is thus continuous for any
$\phi \in \mathcal{C}_{b}$ so Condition 3 above is satisfied.

We also have that
\begin{align*}
&P_{n}((k, \mathbf{z}_{k}), (k', \mathrm{d}\mathbf{y}_{k'})) = g(k, k')
\, \varphi (\mathrm{d}\mathbf{y}_{k'}; \mathbf{0},
\hat{\boldsymbol{\Sigma }}_{k'}) \, \alpha ((k, \mathbf{z}_{k}), (k',
\mathbf{y}_{k'}))
\cr
& + \delta _{(k, \mathbf{z}_{k})}(k', \mathrm{d}\mathbf{y}_{k'})
  \sum _{l \in \mathbf{N}(k)}\int \left (1 - \alpha ((k,
\mathbf{z}_{k}), (l, \mathbf{u}_{k\mapsto k'}))\right ) g(k, l) \,
\varphi (\mathbf{u}_{k\mapsto k'}; \mathbf{0},
\hat{\boldsymbol{\Sigma }}_{k'}) \, \mathrm{d}\mathbf{u}_{k\mapsto k'},
\end{align*}
where in this case
\begin{equation*}
\alpha ((k, \mathbf{z}_{k}), (k', \mathbf{y}_{k'})) = 1 \wedge
\frac{\pi _{K, \mathbf{Z}_{K}}(k',\mathbf{y}_{k'} \mid \mathbf{D}_{n}) \, g(k', k) \, \varphi (\mathbf{z}_{k}; \mathbf{0}, \hat{\boldsymbol{\Sigma }}_{k})}{\pi _{K, \mathbf{Z}_{K}}(k,\mathbf{z}_{k} \mid \mathbf{D}_{n}) \, g(k, k') \, \varphi (\mathbf{y}_{k'}; \mathbf{0}, \hat{\boldsymbol{\Sigma }}_{k'})}.
\end{equation*}
Therefore,
\begin{align*}
P_{n} \phi (k, \mathbf{z}_{k})&= \sum _{k'}\int \phi (k', \mathbf{y}_{k'})
\, P_{n}((k, \mathbf{z}_{k}), (k', \mathrm{d}\mathbf{y}_{k'}))
\cr
&= \sum _{k' \in \mathbf{N}(k)}\int \phi (k', \mathbf{y}_{k'}) \, g(k,
k') \, \varphi (\mathbf{y}_{k'}; \mathbf{0}, \hat{\boldsymbol{\Sigma }}_{k'})
\, \alpha ((k, \mathbf{z}_{k}), (k', \mathbf{y}_{k'})) \, \mathrm{d}
\mathbf{y}_{k'}
\cr
&%
\hspace{-15mm}
+ \phi (k, \mathbf{z}_{k}) \sum _{l \in \mathbf{N}(k)}\int \left (1 -
\alpha ((k, \mathbf{z}_{k}), (l, \mathbf{u}_{k\mapsto k'}))\right )
\, g(k, l) \, \varphi (\mathbf{u}_{k\mapsto k'}; \mathbf{0},
\hat{\boldsymbol{\Sigma }}_{k'}) \, \mathrm{d}\mathbf{u}_{k\mapsto k'}.
\end{align*}
Consequently,
\begin{align}
\label{eqn_proof1_1}
&\sum _{k} \int \left |P_{n} \phi (k, \mathbf{z}_{k}) - P_{
\text{limit}} \phi (k, \mathbf{z}_{k})\right | \, \pi _{K, \mathbf{Z}_{K}}(k,
\mathbf{z}_{k} \mid \mathbf{D}_{n}) \, \mathrm{d}\mathbf{z}_{k}
\cr
&\quad \leq \sum _{k} \int \left |\sum _{k' \in \mathbf{N}(k)}\int
\left (\phi (k', \mathbf{y}_{k'}) \, g(k, k') \, \varphi (\mathbf{y}_{k'};
\mathbf{0}, \hat{\boldsymbol{\Sigma }}_{k'}) \, \alpha ((k, \mathbf{z}_{k}),
(k', \mathbf{y}_{k'})) \right .\right .
\cr
&%
\hspace{0mm}
\left .\left .-\phi (k', \mathbf{y}_{k'}) \, g_{\text{limit}}(k, k')
\, \alpha _{\text{limit}}(k, k') \, \varphi (\mathbf{y}_{k'};
\mathbf{0}, \boldsymbol{\Sigma }_{k'})\right ) \mathrm{d}\mathbf{y}_{k'}
\right | \pi _{K, \mathbf{Z}_{K}}(k, \mathbf{z}_{k} \mid \mathbf{D}_{n})
\, \mathrm{d}\mathbf{z}_{k}
\cr
& + \sum _{k} \int \pi _{K, \mathbf{Z}_{K}}(k, \mathbf{z}_{k} \mid
\mathbf{D}_{n})
\cr
&\times \left |\phi (k, \mathbf{z}_{k}) \sum _{l \in \mathbf{N}(k)}
\int \left (1 - \alpha ((k, \mathbf{z}_{k}), (l, \mathbf{u}_{k
\mapsto k'}))\right ) \, g(k, l) \, \varphi (\mathbf{u}_{k\mapsto k'};
\mathbf{0}, \hat{\boldsymbol{\Sigma }}_{k'}) \, \mathrm{d}\mathbf{u}_{k
\mapsto k'}\right .
\cr
&\left .
\hspace{10mm}
- \phi (k, \mathbf{z}_{k}) \sum _{l \in \mathbf{N}(k)} g_{
\text{limit}}(k, l) \, (1 - \alpha _{\text{limit}}(k, l)) \right | \,
\mathrm{d}\mathbf{z}_{k},
\end{align}
using the triangle inequality.

We now show that first term on the RHS in \eqref{eqn_proof1_1} converges
towards 0 in probability. The second term converges towards 0 in probability
following similar arguments. We will thus be able to conclude by Slutsky's
theorem and monotonicity of probabilities.

Firstly,
\begin{align*}
&\sum _{k} \int \left |\sum _{k' \in \mathbf{N}(k)}\int \left (\phi (k',
\mathbf{y}_{k'}) \, g(k, k') \, \varphi (\mathbf{y}_{k'}; \mathbf{0},
\hat{\boldsymbol{\Sigma }}_{k',n}) \, \alpha ((k, \mathbf{z}_{k}), (k',
\mathbf{y}_{k'})) \right .\right .
\cr
&%
\hspace{0mm}
\left .\left .-\phi (k', \mathbf{y}_{k'}) \, g_{\text{limit}}(k, k')
\, \alpha _{\text{limit}}(k, k') \, \varphi (\mathbf{y}_{k'};
\mathbf{0}, \boldsymbol{\Sigma }_{k'})\right ) \mathrm{d}\mathbf{y}_{k'}
\right | \pi _{K, \mathbf{Z}_{K}}(k, \mathbf{z}_{k} \mid \mathbf{D}_{n})
\, \mathrm{d}\mathbf{z}_{k}
\cr
&\leq M \sum _{k} \sum _{k' \in \mathbf{N}(k)} \int \left |g(k, k')
\, \varphi (\mathbf{y}_{k'}; \mathbf{0}, \hat{\boldsymbol{\Sigma }}_{k',n})
\, \alpha ((k, \mathbf{z}_{k}), (k', \mathbf{y}_{k'})) \right .
\cr
&
\hspace{15mm}%
\left .- g_{\text{limit}}(k, k') \, \alpha _{\text{limit}}(k, k') \,
\varphi (\mathbf{y}_{k'}; \mathbf{0}, \boldsymbol{\Sigma }_{k'})\right |
\pi _{K, \mathbf{Z}_{K}}(k, \mathbf{z}_{k} \mid \mathbf{D}_{n}) \,
\mathrm{d}\mathbf{y}_{k'} \, \mathrm{d}\mathbf{z}_{k},
\end{align*}
because we can assume that $|\phi | \leq M$ with $M$ a positive constant
and using Jensen's inequality. Because the sums are on a finite number
of terms, it suffices to prove that the integral converges to 0, for any
$k, k'$. Note that we have not properly defined
$g_{\text{limit}}(k, k')$ when $\bar{\pi }(k)=0$. For the proof, we can consider
that these $g_{\text{limit}}(k, k')$ with $\bar{\pi }(k)=0$ have any definition,
because we can use that
\begin{align*}
&\int \left |g(k, k') \, \varphi (\mathbf{y}_{k'}; \mathbf{0},
\hat{\boldsymbol{\Sigma }}_{k'}) \, \alpha ((k, \mathbf{z}_{k}), (k',
\mathbf{y}_{k'})) \right .
\cr
&
\hspace{15mm}%
\left .- g_{\text{limit}}(k, k') \, \alpha _{\text{limit}}(k, k') \,
\varphi (\mathbf{y}_{k'}; \mathbf{0}, \boldsymbol{\Sigma }_{k'})\right |
\pi _{K, \mathbf{Z}_{K}}(k, \mathbf{z}_{k} \mid \mathbf{D}_{n}) \,
\mathrm{d}\mathbf{y}_{k'} \, \mathrm{d}\mathbf{z}_{k}
\cr
&\leq \int \varphi (\mathbf{y}_{k'}; \mathbf{0},
\hat{\boldsymbol{\Sigma }}_{k'}) \pi _{K, \mathbf{Z}_{K}}(k, \mathbf{z}_{k}
\mid \mathbf{D}_{n}) \, \mathrm{d}\mathbf{y}_{k'} \, \mathrm{d}
\mathbf{z}_{k}
\cr
&
\hspace{15mm}
+ \int \varphi (\mathbf{y}_{k'}; \mathbf{0}, \boldsymbol{\Sigma }_{k'})
\pi _{K, \mathbf{Z}_{K}}(k, \mathbf{z}_{k} \mid \mathbf{D}_{n}) \,
\mathrm{d}\mathbf{y}_{k'} \, \mathrm{d}\mathbf{z}_{k} = 2 \, \pi (k
\mid \mathbf{D}_{n}) \longrightarrow 0,
\end{align*}
using the triangle inequality, that
$0 \leq g, \alpha , g_{\text{limit}}, \alpha _{\text{limit}} \leq 1$, and
that $\pi (k \mid \mathbf{D}_{n}) \longrightarrow \bar{\pi }(k) = 0$.

We can thus restrict our attention to the case where
$\bar{\pi }(k)>0$. Note that for similar reasons, we can restrict our attention
to the case where $g_{\text{limit}}(k, k')>0$ (implying that
$\bar{\pi }(k')>0$). Indeed, consider that $\bar{\pi }(k)>0$, then all terms
involved in $g(k, k')$, i.e.\ $h(\hat{\pi }(l\mid \mathbf{D}_{n}) /
\hat{\pi }(k\mid \mathbf{D}_{n}))$ with $l \in \mathbf{N}(k)$, have a well-defined
limit in probability given by $h(\bar{\pi }(l) / \bar{\pi }(k))$ by Slutsky's
and continuous mapping theorems and Assumption~\ref{ass2}. Given that
$g(k, k')$ is a continuous function of a finite number of random variables
all converging towards a constant in probability,
$g(k, k') \longrightarrow g_{\text{limit}}(k, k') = 0$ by Slutsky's theorem, and $g_{\text{limit}}(k, k') = 0$ when
$\bar{\pi }(k')=0$.

Using twice the triangle inequality,
\begin{align}
\label{eq:weak_conv}
&\int \left |g(k, k') \, \varphi (\mathbf{y}_{k'}; \mathbf{0},
\hat{\boldsymbol{\Sigma }}_{k'}) \, \alpha ((k, \mathbf{z}_{k}), (k',
\mathbf{y}_{k'})) \right .
\cr
&
\hspace{15mm}%
\left .- g_{\text{limit}}(k, k') \, \alpha _{\text{limit}}(k, k') \,
\varphi (\mathbf{y}_{k'}; \mathbf{0}, \boldsymbol{\Sigma }_{k'})\right |
\pi _{K, \mathbf{Z}_{K}}(k, \mathbf{z}_{k} \mid \mathbf{D}_{n}) \,
\mathrm{d}\mathbf{y}_{k'} \, \mathrm{d}\mathbf{z}_{k}
\cr
&\leq \int g(k, k')\, \varphi (\mathbf{y}_{k'}; \mathbf{0},
\hat{\boldsymbol{\Sigma }}_{k'}) \, \pi _{K, \mathbf{Z}_{K}}(k,
\mathbf{z}_{k} \mid \mathbf{D}_{n})
\, |\alpha ((k, \mathbf{z}_{k}), (k', \mathbf{y}_{k'}) -
\alpha _{\text{limit}}(k, k')| \, \mathrm{d}\mathbf{y}_{k'} \,
\mathrm{d}\mathbf{z}_{k}
\cr
&\quad +|g(k, k') - g_{\text{limit}}(k, k')| \, \alpha _{\text{limit}}(k,
k')
 \int \varphi (\mathbf{y}_{k'}; \mathbf{0},
\hat{\boldsymbol{\Sigma }}_{k'}) \pi _{K, \mathbf{Z}_{K}}(k, \mathbf{z}_{k}
\mid \mathbf{D}_{n}) \, \mathrm{d}\mathbf{y}_{k'} \, \mathrm{d}
\mathbf{z}_{k}
\cr
&\quad +g_{\text{limit}}(k, k') \, \alpha _{\text{limit}}(k, k')
 \int |\varphi (\mathbf{y}_{k'}; \mathbf{0},
\hat{\boldsymbol{\Sigma }}_{k'}) - \varphi (\mathbf{y}_{k'}; \mathbf{0},
\boldsymbol{\Sigma }_{k'})|\pi _{K, \mathbf{Z}_{K}}(k, \mathbf{z}_{k}
\mid \mathbf{D}_{n}) \, \mathrm{d}\mathbf{y}_{k'} \, \mathrm{d}
\mathbf{z}_{k}.
\end{align}
It is easily seen that the second term on the RHS converges to 0 given
that $g(k, k') \longrightarrow g_{\text{limit}}(k, k')$. The third term
can be bounded above as follows:
\begin{align*}
&g_{\text{limit}}(k, k') \, \alpha _{\text{limit}}(k, k')
 \int |\varphi (\mathbf{y}_{k'}; \mathbf{0},
\hat{\boldsymbol{\Sigma }}_{k'}) - \varphi (\mathbf{y}_{k'}; \mathbf{0},
\boldsymbol{\Sigma }_{k'})|\pi _{K, \mathbf{Z}_{K}}(k, \mathbf{z}_{k}
\mid \mathbf{D}_{n}) \, \mathrm{d}\mathbf{y}_{k'} \, \mathrm{d}
\mathbf{z}_{k}
\cr
&\leq \int |\varphi (\mathbf{y}_{k'}; \mathbf{0},
\hat{\boldsymbol{\Sigma }}_{k'}) - \varphi (\mathbf{y}_{k'}; \mathbf{0},
\boldsymbol{\Sigma }_{k'})| \, \mathrm{d}\mathbf{y}_{k'}
\cr
&\leq \left [\text{tr}(\boldsymbol{\Sigma }_{k'}^{-1}
\hat{\boldsymbol{\Sigma }}_{k'} - \mathbb{I}_{k'}) - \log \text{det}(
\hat{\boldsymbol{\Sigma }}_{k'}\boldsymbol{\Sigma }_{k'}^{-1})\right ]^{1/2},
\end{align*}
using that
$0 \leq g_{\text{limit}}, \alpha _{\text{limit}}, \pi (k \mid
\mathbf{D}_{n}) \leq 1$ in the first inequality and then Proposition 2.1
of \cite{devroye2018total} in the second one, where
$\mathbb{I}_{k'}$ is the identity matrix of size $d_{k'}$, and
$\text{tr}(\, \cdot \,)$ and $\text{det}(\, \cdot \,)$ are the trace and
determinant operators, respectively. The upper bound converges to 0 under Assumption~\ref{ass3}. For the first term in \eqref{eq:weak_conv}, it is less
direct. Using that $1\wedge x$ is a Lipschitz function with constant
$1$ for all $x \geq 0$ and the triangle inequality,
\begin{align*}
&\int g(k, k')\, \varphi (\mathbf{y}_{k'}; \mathbf{0},
\hat{\boldsymbol{\Sigma }}_{k'}) \, \pi _{K, \mathbf{Z}_{K}}(k,
\mathbf{z}_{k} \mid \mathbf{D}_{n})
\, |\alpha ((k, \mathbf{z}_{k}), (k', \mathbf{y}_{k'}) -
\alpha _{\text{limit}}(k, k')| \, \mathrm{d}\mathbf{y}_{k'} \,
\mathrm{d}\mathbf{z}_{k}
\cr
&\leq \int \left |\pi _{K, \mathbf{Z}_{K}}(k',\mathbf{y}_{k'} \mid
\mathbf{D}_{n}) \, g(k', k) \, \varphi (\mathbf{z}_{k}; \mathbf{0},
\hat{\boldsymbol{\Sigma }}_{k}) \right .
\cr
&%
\hspace{5mm}%
\left .- g(k, k')\, \varphi (\mathbf{y}_{k'}; \mathbf{0},
\hat{\boldsymbol{\Sigma }}_{k'}) \, \pi _{K, \mathbf{Z}_{K}}(k,
\mathbf{z}_{k} \mid \mathbf{D}_{n}) \,
\frac{\bar{\pi }(k') \, g_{\text{limit}}(k', k)}{\bar{\pi }(k) \, g_{\text{limit}}(k, k')}
\right | \, \mathrm{d}\mathbf{y}_{k'} \, \mathrm{d}\mathbf{z}_{k}
\cr
&=\int \left |\pi _{\mathbf{Z}_{K}}(\mathbf{y}_{k'} \mid k',
\mathbf{D}_{n}) \, \pi (k' \mid \mathbf{D}_{n}) \, g(k', k) \,
\varphi (\mathbf{z}_{k}; \mathbf{0}, \hat{\boldsymbol{\Sigma }}_{k})
\right .
\cr
&%
\hspace{0mm}%
\left .- g(k, k')\, \varphi (\mathbf{y}_{k'}; \mathbf{0},
\hat{\boldsymbol{\Sigma }}_{k'}) \, \pi _{\mathbf{Z}_{K}}(\mathbf{z}_{k}
\mid k, \mathbf{D}_{n}) \, \pi (k \mid \mathbf{D}_{n}) \,
\frac{\bar{\pi }(k') \, g_{\text{limit}}(k', k)}{\bar{\pi }(k) \, g_{\text{limit}}(k, k')}
\right | \, \mathrm{d}\mathbf{y}_{k'} \, \mathrm{d}\mathbf{z}_{k}
\cr
&\leq \int \left |\pi _{\mathbf{Z}_{K}}(\mathbf{y}_{k'} \mid k',
\mathbf{D}_{n}) \, \pi (k' \mid \mathbf{D}_{n}) \, g(k', k) \,
\varphi (\mathbf{z}_{k}; \mathbf{0}, \hat{\boldsymbol{\Sigma }}_{k})
\right .
\cr
&%
\hspace{10mm}%
\left .- \varphi (\mathbf{y}_{k'}; \mathbf{0},
\hat{\boldsymbol{\Sigma }}_{k'}) \, \pi _{\mathbf{Z}_{K}}(\mathbf{z}_{k}
\mid k, \mathbf{D}_{n}) \, \bar{\pi }(k') \, g_{\text{limit}}(k', k)
\right | \, \mathrm{d}\mathbf{y}_{k'} \, \mathrm{d}\mathbf{z}_{k}
\cr
&+\left |1 -
\frac{\pi (k \mid \mathbf{D}_{n}) \, g(k, k')}{\bar{\pi }(k) \, g_{\text{limit}}(k, k')}
\right |\bar{\pi }(k') \, g_{\text{limit}}(k', k)
 \int \varphi (\mathbf{y}_{k'}; \mathbf{0},
\hat{\boldsymbol{\Sigma }}_{k'}) \, \pi _{\mathbf{Z}_{K}}(\mathbf{z}_{k}
\mid k, \mathbf{D}_{n}) \, \mathrm{d}\mathbf{y}_{k'} \, \mathrm{d}
\mathbf{z}_{k}.
\end{align*}
The last term is easily seen to converge towards 0 by Slutsky's theorem.

We now analyse the other term. Using the triangle inequality,
\begin{align*}
&\int \left |\pi _{\mathbf{Z}_{K}}(\mathbf{y}_{k'} \mid k',\mathbf{D}_{n})
\, \pi (k' \mid \mathbf{D}_{n}) \, g(k', k) \, \varphi (\mathbf{z}_{k};
\mathbf{0}, \hat{\boldsymbol{\Sigma }}_{k}) \right .
\cr
&%
\hspace{10mm}%
\left .- \varphi (\mathbf{y}_{k'}; \mathbf{0},
\hat{\boldsymbol{\Sigma }}_{k'}) \, \pi _{\mathbf{Z}_{K}}(\mathbf{z}_{k}
\mid k, \mathbf{D}_{n}) \, \bar{\pi }(k') \, g_{\text{limit}}(k', k)
\right | \, \mathrm{d}\mathbf{y}_{k'} \, \mathrm{d}\mathbf{z}_{k}
\cr
&\leq \left |\pi (k' \mid \mathbf{D}_{n}) \, g(k', k) - \bar{\pi }(k')
\, g_{\text{limit}}(k', k)\right |
\cr
&\qquad \times \int \pi _{\mathbf{Z}_{K}}(\mathbf{y}_{k'} \mid k',
\mathbf{D}_{n}) \, \varphi (\mathbf{z}_{k}; \mathbf{0},
\hat{\boldsymbol{\Sigma }}_{k}) \, \mathrm{d}\mathbf{y}_{k'} \,
\mathrm{d}\mathbf{z}_{k}
\cr
& \quad +\bar{\pi }(k') \, g_{\text{limit}}(k', k)
\cr
&\times \int \left |\pi _{\mathbf{Z}_{K}}(\mathbf{y}_{k'} \mid k',
\mathbf{D}_{n}) \, \varphi (\mathbf{z}_{k}; \mathbf{0},
\hat{\boldsymbol{\Sigma }}_{k}) - \varphi (\mathbf{y}_{k'}; \mathbf{0},
\hat{\boldsymbol{\Sigma }}_{k'}) \, \pi _{\mathbf{Z}_{K}}(\mathbf{z}_{k}
\mid k, \mathbf{D}_{n})\right |\, \mathrm{d}\mathbf{y}_{k'} \,
\mathrm{d}\mathbf{z}_{k}.
\end{align*}
The first term is seen to converge towards 0 by Slutsky's theorem.

Using again the triangle inequality, we obtain the following bound on the
other term:
\begin{align*}
&\int \left |\pi _{\mathbf{Z}_{K}}(\mathbf{y}_{k'} \mid k',\mathbf{D}_{n})
\, \varphi (\mathbf{z}_{k}; \mathbf{0}, \hat{\boldsymbol{\Sigma }}_{k}) -
\varphi (\mathbf{y}_{k'}; \mathbf{0}, \hat{\boldsymbol{\Sigma }}_{k'})
\, \pi _{\mathbf{Z}_{K}}(\mathbf{z}_{k} \mid k, \mathbf{D}_{n})
\right |\, \mathrm{d}\mathbf{y}_{k'} \, \mathrm{d}\mathbf{z}_{k}
\cr
&\leq \int \left |\pi _{\mathbf{Z}_{K}}(\mathbf{y}_{k'} \mid k',
\mathbf{D}_{n}) \, \varphi (\mathbf{z}_{k}; \mathbf{0},
\hat{\boldsymbol{\Sigma }}_{k}) - \varphi (\mathbf{y}_{k'}; \mathbf{0},
\hat{\boldsymbol{\Sigma }}_{k'}) \, \varphi (\mathbf{z}_{k}; \mathbf{0},
\hat{\boldsymbol{\Sigma }}_{k})\right |\, \mathrm{d}\mathbf{y}_{k'} \,
\mathrm{d}\mathbf{z}_{k}
\cr
&+\int \left |\varphi (\mathbf{y}_{k'}; \mathbf{0},
\hat{\boldsymbol{\Sigma }}_{k'}) \, \varphi (\mathbf{z}_{k}; \mathbf{0},
\hat{\boldsymbol{\Sigma }}_{k}) - \varphi (\mathbf{y}_{k'}; \mathbf{0},
\hat{\boldsymbol{\Sigma }}_{k'}) \, \pi _{\mathbf{Z}_{K}}(\mathbf{z}_{k}
\mid k, \mathbf{D}_{n})\right |\, \mathrm{d}\mathbf{y}_{k'} \,
\mathrm{d}\mathbf{z}_{k}.
\end{align*}
We prove the convergence to 0 of the first term. The other convergence
follows from a similar argument.

Using the triangle inequality,
\begin{align*}
&\int \left |\pi _{\mathbf{Z}_{K}}(\mathbf{y}_{k'} \mid k',\mathbf{D}_{n})
\, \varphi (\mathbf{z}_{k}; \mathbf{0}, \hat{\boldsymbol{\Sigma }}_{k}) -
\varphi (\mathbf{y}_{k'}; \mathbf{0}, \hat{\boldsymbol{\Sigma }}_{k'})
\, \varphi (\mathbf{z}_{k}; \mathbf{0}, \hat{\boldsymbol{\Sigma }}_{k})
\right |\, \mathrm{d}\mathbf{y}_{k'} \, \mathrm{d}\mathbf{z}_{k}
\cr
&\quad =\int \left |\pi _{\mathbf{Z}_{K}}(\mathbf{y}_{k'} \mid k',
\mathbf{D}_{n}) - \varphi (\mathbf{y}_{k'}; \mathbf{0},
\hat{\boldsymbol{\Sigma }}_{k'}) \right |\, \mathrm{d}\mathbf{y}_{k'}
\cr
&\quad \leq \int \left |\pi _{\mathbf{Z}_{K}}(\mathbf{y}_{k'} \mid k',
\mathbf{D}_{n}) - \varphi (\mathbf{y}_{k'}; \mathbf{0},
\boldsymbol{\Sigma }_{k'}) \right |\, \mathrm{d}\mathbf{y}_{k'}
\cr
&\qquad +\int \left |\varphi (\mathbf{y}_{k'}; \mathbf{0},
\boldsymbol{\Sigma }_{k'}) - \varphi (\mathbf{y}_{k'}; \mathbf{0},
\hat{\boldsymbol{\Sigma }}_{k'}) \right |\, \mathrm{d}\mathbf{y}_{k'}.
\end{align*}
The first term converges to 0 by Assumption~\ref{ass3} after a change of variables,
and the second one converges to 0 as well by Proposition 2.1 of
\cite{devroye2018total} and Assumption~\ref{ass3} as seen above. This concludes
the proof.
\end{proof}

\begin{proof}[Proof of Proposition~\ref{prop_peskun}]
Because $\bar{\pi }(k^{*}) = 1$, we only have to establish the inequality
for any set $A = \{(k, \mathbf{z}_{k}) \in \{k^{*}\} \times B\}$. So
$A \setminus \{(k^{*}, \mathbf{z}_{k^{*}})\}$ implies that a parameter
update is proposed and accepted with a parameter proposal, denoted here
by $\mathbf{y}_{k^{*}}$, in $B$. When the chain is in stationarity,
$P_{\text{limit},1}$ always proposes to update the parameters. Therefore,
\begin{equation*}
P_{\text{limit},1}((k^{*}, \mathbf{z}_{k^{*}}), A \setminus \{(k^{*},
\mathbf{z}_{k^{*}})\}) = \mathbb{P}(\mathbf{y}_{k^{*}} \in B
\text{ is accepted}).
\end{equation*}
In contrast,
\begin{align*}
P_{\text{limit},2}((k^{*}, \mathbf{z}_{k^{*}}), A \setminus \{(k^{*},
\mathbf{z}_{k^{*}})\}) &= g(k^{*}, k^{*}) \, \mathbb{P}(\mathbf{y}_{k^{*}}
\in B \text{ is accepted})
\cr
&\leq \mathbb{P}(\mathbf{y}_{k^{*}} \in B \text{ is accepted})
\cr
&= P_{\text{limit},1}((k^{*}, \mathbf{z}_{k^{*}}), A \setminus \{(k^{*},
\mathbf{z}_{k^{*}})\}).
\end{align*}
The fact that $\bar{\pi }(k^{*}) = 1$ implies the order on the asymptotic
variances \citep{tierney1998note}.
\end{proof}

\begin{proof}[Proof of Proposition~\ref{prop_inv_algo_model_imp}]
We prove the result for the case $T = 1$ (without annealing intermediate
distributions), to simplify; the general case is proved similarly. When
$T = 1$, $r_{\text{RJ2}} = r_{\text{RJ}}$ which is the ratio in \eqref{eqn_acc_prob_RJ}. We prove that the probability to reach the state
$\{k'\}\times \{\mathbf{y}_{k'}^{(j^{*})}\in A_{k'}\}$, from
$\{k\}\times \{\mathbf{x}_{k}\in A_{k}\}$, is equal to the probability
of the reverse move. We denote by $P$ the Markov kernel. We thus prove
that
\begin{align*}
&\int _{\{\mathbf{x}_{k}\in A_{k}\}} \pi (k, \mathbf{x}_{k} \mid
\mathbf{D}_{n})\int _{\{\mathbf{y}_{k'}^{(j^{*})}\in A_{k'}\}} P((k,
\mathbf{x}_{k}), (k', \mathrm{d}\mathbf{y}_{k'}^{(j^{*})})) \,
\mathrm{d}\mathbf{x}_{k}
\cr
&\qquad = \int _{\{\mathbf{y}_{k'}^{(j^{*})}\in A_{k'}\}} \pi (k',
\mathbf{y}_{k'}^{(j^{*})} \mid \mathbf{D}_{n})\int _{\{\mathbf{x}_{k}
\in A_{k}\}} P((k', \mathbf{y}_{k'}^{(j^{*})}), (k, \mathrm{d}
\mathbf{x}_{k})) \, \mathrm{d}\mathbf{y}_{k'}^{(j^{*})}.
\end{align*}
Note that we abused notation by denoting the integration variables
$\mathbf{y}_{k'}^{(j^{*})}$ and $\mathbf{x}_{k}$ because a group of vectors
like $\mathbf{y}_{\bullet }^{(0, \bullet )}$ are used in the transitions
and they are not of the same dimension as $\mathbf{y}_{k'}$ and
$\mathbf{x}_{k}$. To simplify, we will use notation like
$\mathbf{y}_{s}^{(j)} := \mathbf{y}_{s}^{(0, j)}$ to denote the $j$-th
proposal, $j\in \{1,\ldots ,N\}$.

We now introduce notation to improve readability. We define three joint
densities that are used to enhance the approximations when Step 2.(i) is
applied to sample the proposal:
\begin{align*}
\bar{q}_{k \mapsto \mathbf{N}(k)\setminus \{k'\}} (\mathbf{y}_{
\mathbf{N}(k)\setminus \{k'\}}^{(\bullet )}) &:= \prod _{l \in
\mathbf{N}(k) \setminus \{k'\}} \prod _{j=1}^{N} q_{k\mapsto l}(
\mathbf{y}_{l}^{(j)}),
\cr
\bar{q}_{k \mapsto k'}(\mathbf{y}_{k'}^{(\bullet )}) &:= \prod _{j=1}^{N}
q_{k\mapsto k'}(\mathbf{y}_{k'}^{(j)}),
\cr
\bar{\bar{q}}_{k' \mapsto \mathbf{N}(k')\setminus \{k\}}(\mathbf{z}^{(
\bullet )}_{\mathbf{N}(k')\setminus \{k\}}) &:= \prod _{l\in
\mathbf{N}(k') \setminus \{k\}} q_{k'\mapsto l}(\mathbf{z}_{l}^{(j^{*})})
\prod _{j=1 (j\neq j^{*})}^{N} q_{l\mapsto k'}(\mathbf{z}_{k'}^{(j)}).
\end{align*}
The densities $\bar{q}_{k \mapsto \mathbf{N}(k)\setminus \{k'\}}$ and
$\bar{q}_{k \mapsto k'}$ together represent the joint density of the random
variables sampled in the first part of Step 2.(i).
$\bar{\bar{q}}_{k' \mapsto \mathbf{N}(k')\setminus \{k\}}$ represents the
joint density of the random variables sampled in the second part of Step
2.(i).

We now define three joint densities that are used to enhance the approximations
when Step 2.(ii) is applied to sample the proposal:
\begin{align*}
\tilde{q}_{k \mapsto \mathbf{N}(k) \setminus \{k'\}} (\mathbf{y}^{(
\bullet )}_{\mathbf{N}(k) \setminus \{k'\}}) &:= \prod _{l \in
\mathbf{N}(k) \setminus \{k'\}} q_{k\mapsto l}(\mathbf{y}_{l}^{(j^{*})})
\prod _{j=1 (j\neq j^{*})}^{N} q_{l\mapsto k}(\mathbf{y}_{k}^{(j)}),
\cr
\tilde{q}_{k \mapsto k'} (\mathbf{y}_{k'}^{(\bullet )}) &:= q_{k
\mapsto k'}(\mathbf{y}_{k'}^{(j^{*})}) \prod _{j=1 (j\neq j^{*})}^{N} q_{k'
\mapsto k}(\mathbf{y}_{k}^{(j)}),
\cr
\tilde{\tilde{q}}_{k' \mapsto \mathbf{N}(k') \setminus \{k\}}(
\mathbf{z}^{(\bullet )}_{\mathbf{N}(k')\setminus \{k\}})&:= \prod _{l
\in N(k') \setminus \{k\}} \prod _{j=1}^{N} q_{k'\mapsto l}(
\mathbf{z}_{l}^{(j)}).
\end{align*}
The densities $\tilde{q}_{k \mapsto \mathbf{N}(k) \setminus \{k'\}}$ and
$ \tilde{q}_{k \mapsto k'} $ together represent the joint density of the
random variables sampled in the first part of Step 2.(ii).
$\tilde{\tilde{q}}_{k' \mapsto \mathbf{N}(k') \setminus \{k\}}$ represents
the joint density of the random variables sampled in the second part of
Step 2.(ii).

We have that
\begin{align*}
&P((k, \mathbf{x}_{k}), (k', \mathrm{d}\mathbf{y}_{k'}^{(j^{*})}))=
\frac{1}{2} \, \bar{q}_{k \mapsto \mathbf{N}(k)\setminus \{k'\}} (
\mathbf{y}_{\mathbf{N}(k)\setminus \{k'\}}^{(\bullet )}) \, \bar{q}_{k
\mapsto k'}(\mathbf{y}_{k'}^{(\bullet )})
\cr
&\times g_{\text{imp.}}(k, k', \mathbf{x}_{k}, \mathbf{y}_{\mathbf{N}(k)
\setminus \{k'\}}^{(\bullet )}, \mathbf{y}_{k'}^{(\bullet )}) \,
\frac{r_{\text{RJ}}((k,\mathbf{x}_{k}),(k',\mathbf{y}_{k'}^{(j^{*})}))}{N \bar{r}(k, k', \mathbf{x}_{k}, \mathbf{y}_{k'}^{(\bullet )})}
\cr
&\times \bar{\bar{q}}_{k' \mapsto \mathbf{N}(k')\setminus \{k\}}(
\mathbf{z}^{(\bullet )}_{\mathbf{N}(k')\setminus \{k\}})
\cr
&\times \left (1 \wedge
\frac{g_{\text{imp.}}(k', k, \mathbf{x}_{k}, \mathbf{y}_{k'}^{(\bullet )}, \mathbf{z}^{(\bullet )}_{\mathbf{N}(k')\setminus \{k\}})}{g_{\text{imp.}}(k, k', \mathbf{x}_{k}, \mathbf{y}_{\mathbf{N}(k)\setminus \{k'\}}^{(\bullet )}, \mathbf{y}_{k'}^{(\bullet )})}
\, \bar{r}(k, k', \mathbf{x}_{k}, \mathbf{y}_{k'}^{(\bullet )})
\right )
\cr
&\times \mathrm{d}\mathbf{z}^{(\bullet )}_{\mathbf{N}(k')\setminus \{k
\}} \, \mathrm{d}\mathbf{y}_{\mathbf{N}(k)\setminus \{k'\}}^{(
\bullet )} \, \mathrm{d}\mathbf{y}_{k'}^{(\bullet )}
\cr
& + \frac{1}{2} \, \tilde{q}_{k \mapsto \mathbf{N}(k) \setminus \{k'
\}} (\mathbf{y}^{(\bullet )}_{\mathbf{N}(k) \setminus \{k'\}}) \,
\tilde{q}_{k \mapsto k'} (\mathbf{y}_{k'}^{(\bullet )}) \, g_{
\text{imp.}}(k, k', \mathbf{x}_{k}, \mathbf{y}^{(\bullet )}_{
\mathbf{N}(k) \setminus \{k'\}}, \mathbf{y}_{k'}^{(\bullet )}) \,
\frac{1}{N}
\cr
&\times \tilde{\tilde{q}}_{k' \mapsto \mathbf{N}(k') \setminus \{k\}}(
\mathbf{z}^{(\bullet )}_{\mathbf{N}(k')\setminus \{k\}})
\cr
&\times \left (1 \wedge
\frac{g_{\text{imp.}}(k', k, \mathbf{x}_{k}, \mathbf{y}_{k'}^{(\bullet )}, \mathbf{z}^{(\bullet )}_{\mathbf{N}(k')\setminus \{k\}})}{g_{\text{imp.}}(k, k', \mathbf{x}_{k}, \mathbf{y}^{(\bullet )}_{\mathbf{N}(k) \setminus \{k'\}}, \mathbf{y}_{k'}^{(\bullet )})}
\, \bar{r}(k', k, \mathbf{x}_{k}, \mathbf{y}_{k'}^{(\bullet )})^{-1}
\right )
\cr
&\times \mathrm{d}\mathbf{z}^{(\bullet )}_{\mathbf{N}(k')\setminus \{k
\}} \, \mathrm{d}\mathbf{y}^{(\bullet )}_{\mathbf{N}(k) \setminus \{k'
\}} \, \mathrm{d}\mathbf{y}_{k'}^{(\bullet )}
\cr
& \quad + \delta _{(k, \mathbf{x}_{k})}(k', \mathrm{d}\mathbf{y}_{k'}^{(j^{*})})
\, \mathbb{P}(\text{rejection}\mid (k', \mathbf{y}_{k'}^{(j^{*})})),
\end{align*}
where
$\mathbb{P}(\text{rejection}\mid (k', \mathbf{y}_{k'}^{(j^{*})}))$ is the
rejection probability given that the current state is
$(k', \mathbf{y}_{k'}^{(j^{*})})$. Note that we considered that in Step
2.(ii) we set uniformly at random the index of the proposal. This is however
in practice not important (which is why in Algorithm~\ref{algo_RJ_imp_model} we set it to be 1) because of the form of the
acceptance ratio. Note also that we use the notation
$\bar{r}(k, k', \mathbf{x}_{k}, \mathbf{y}_{k'}^{(\bullet )})$ to be clear
about which variables is involved.

The probability of reaching the state
$\{k'\}\times \{\mathbf{y}_{k'}^{(j^{*})}\in A_{k'}\}$, from
$\{k\}\times \{\mathbf{x}_{k}\in A_{k}\}$, is thus given by
\begin{align}
\label{eqn1_proof_prop}
&\int _{\{\mathbf{x}_{k}\in A_{k}\}} \pi (k, \mathbf{x}_{k} \mid
\mathbf{D}_{n})\int _{\{\mathbf{y}_{k'}^{(j^{*})}\in A_{k'}\}}
\frac{1}{2} \, \bar{q}_{k \mapsto \mathbf{N}(k)\setminus \{k'\}} (
\mathbf{y}_{\mathbf{N}(k)\setminus \{k'\}}^{(\bullet )}) \, \bar{q}_{k
\mapsto k'}(\mathbf{y}_{k'}^{(\bullet )})
\cr
& \times g_{\text{imp.}}(k, k', \mathbf{x}_{k}, \mathbf{y}_{\mathbf{N}(k)
\setminus \{k'\}}^{(\bullet )}, \mathbf{y}_{k'}^{(\bullet )}) \,
\frac{r_{\text{RJ}}((k,\mathbf{x}_{k}),(k',\mathbf{y}_{k'}^{(j^{*})}))}{N \bar{r}(k, k', \mathbf{x}_{k}, \mathbf{y}_{k'}^{(\bullet )})}
\cr
&\times \bar{\bar{q}}_{k' \mapsto \mathbf{N}(k')\setminus \{k\}}(
\mathbf{z}^{(\bullet )}_{\mathbf{N}(k')\setminus \{k\}})
\cr
&\times \left (1 \wedge
\frac{g_{\text{imp.}}(k', k, \mathbf{x}_{k}, \mathbf{y}_{k'}^{(\bullet )}, \mathbf{z}^{(\bullet )}_{\mathbf{N}(k')\setminus \{k\}})}{g_{\text{imp.}}(k, k', \mathbf{x}_{k}, \mathbf{y}_{\mathbf{N}(k)\setminus \{k'\}}^{(\bullet )}, \mathbf{y}_{k'}^{(\bullet )})}
\, \bar{r}(k, k', \mathbf{x}_{k}, \mathbf{y}_{k'}^{(\bullet )})
\right )
\cr
&\times \mathrm{d}\mathbf{z}^{(\bullet )}_{\mathbf{N}(k')\setminus \{k
\}} \, \mathrm{d}\mathbf{y}_{\mathbf{N}(k)\setminus \{k'\}}^{(
\bullet )} \, \mathrm{d}\mathbf{y}_{k'}^{(\bullet )} \, \mathrm{d}
\mathbf{x}_{k}
\cr
&+\int _{\{\mathbf{x}_{k}\in A_{k}\}} \pi (k, \mathbf{x}_{k} \mid
\mathbf{D}_{n})\int _{\{\mathbf{y}_{k'}^{(j^{*})}\in A_{k'}\}}
\frac{1}{2} \, \tilde{q}_{k \mapsto \mathbf{N}(k) \setminus \{k'\}} (
\mathbf{y}^{(\bullet )}_{\mathbf{N}(k) \setminus \{k'\}}) \,
\tilde{q}_{k \mapsto k'} (\mathbf{y}_{k'}^{(\bullet )})
\cr
& \times g_{\text{imp.}}(k, k', \mathbf{x}_{k}, \mathbf{y}^{(\bullet )}_{
\mathbf{N}(k) \setminus \{k'\}}, \mathbf{y}_{k'}^{(\bullet )}) \,
\frac{1}{N} \, \tilde{\tilde{q}}_{k' \mapsto \mathbf{N}(k')
\setminus \{k\}}(\mathbf{z}^{(\bullet )}_{\mathbf{N}(k')\setminus \{k
\}})
\cr
&\times \left (1 \wedge
\frac{g_{\text{imp.}}(k', k, \mathbf{x}_{k}, \mathbf{y}_{k'}^{(\bullet )}, \mathbf{z}^{(\bullet )}_{\mathbf{N}(k')\setminus \{k\}})}{g_{\text{imp.}}(k, k', \mathbf{x}_{k}, \mathbf{y}^{(\bullet )}_{\mathbf{N}(k) \setminus \{k'\}}, \mathbf{y}_{k'}^{(\bullet )})}
\, \bar{r}(k', k, \mathbf{x}_{k}, \mathbf{y}_{k'}^{(\bullet )})^{-1}
\right )
\cr
&\times \mathrm{d}\mathbf{z}^{(\bullet )}_{\mathbf{N}(k')\setminus \{k
\}} \, \mathrm{d}\mathbf{y}^{(\bullet )}_{\mathbf{N}(k) \setminus \{k'
\}} \, \mathrm{d}\mathbf{y}_{k'}^{(\bullet )} \, \mathrm{d}\mathbf{x}_{k}
\cr
&+\int _{\{\mathbf{x}_{k}\in A_{k}\}} \pi (k, \mathbf{x}_{k} \mid
\mathbf{D}_{n})
\cr
&\times \int _{\{\mathbf{y}_{k'}^{(j^{*})}\in A_{k'}\}} \delta _{(k,
\mathbf{x}_{k})}(k', \mathrm{d}\mathbf{y}_{k'}^{(j^{*})}) \,
\mathbb{P}(\text{rejection}\mid (k', \mathbf{y}_{k'}^{(j^{*})})) \,
\mathrm{d}\mathbf{x}_{k}.
\end{align}
We now prove that the first integral can be rewritten as that corresponding
to Step 2.(ii) for the reverse move; the second integral corresponds instead
to Step 2.(i) for the reverse move, and the last term to the probability
of rejecting from $(k',\mathbf{y}_{k'}^{(j^{*})})$.

By changing the integration variables
$\mathbf{y}_{k'} \leftarrow \mathbf{y}_{k'}^{(j^{*})}$,
$\mathbf{x}_{k}^{(j^{*})} \leftarrow \mathbf{x}_{k}$ and
$\mathbf{x}_{k}^{(\bullet )} \leftarrow \mathbf{y}_{k'}^{(\bullet )}$ but
with $\mathbf{y}_{k'}$ replaced by $\mathbf{x}_{k}^{(j^{*})}$, the first
integral can be rewritten as
\begin{align*}
&\int _{\{\mathbf{y}_{k'}\in A_{k'}\} \times \{\mathbf{x}_{k}^{(j^{*})}
\in A_{k}\}} \pi (k', \mathbf{y}_{k'} \mid \mathbf{D}_{n}) \,
\frac{1}{2} \, \tilde{q}_{k' \mapsto \mathbf{N}(k')\setminus \{k\}}(
\mathbf{z}^{(\bullet )}_{\mathbf{N}(k')\setminus \{k\}}) \, \tilde{q}_{k'
\mapsto k} (\mathbf{x}_{k}^{(\bullet )})
\cr
&\times g_{\text{imp.}}(k', k, \mathbf{y}_{k'}, \mathbf{x}_{k}^{(
\bullet )}, \mathbf{z}^{(\bullet )}_{\mathbf{N}(k')\setminus \{k\}})
\frac{1}{N} \, \tilde{\tilde{q}}_{k \mapsto \mathbf{N}(k) \setminus
\{k'\}}(\mathbf{y}_{\mathbf{N}(k)\setminus \{k'\}}^{(\bullet )})
\cr
& \times \left (1 \wedge
\frac{g_{\text{imp.}}(k, k', \mathbf{y}_{k'}, \mathbf{y}_{\mathbf{N}(k)\setminus \{k'\}}^{(\bullet )}, \mathbf{x}_{k}^{(\bullet )})}{g_{\text{imp.}}(k', k, \mathbf{y}_{k'}, \mathbf{x}_{k}^{(\bullet )}, \mathbf{z}^{(\bullet )}_{\mathbf{N}(k')\setminus \{k\}})}
\, \bar{r}(k, k', \mathbf{y}_{k'}, \mathbf{x}_{k}^{(\bullet )})^{-1}
\right )
\cr
&\times \mathrm{d}\mathbf{z}^{(\bullet )}_{\mathbf{N}(k')\setminus \{k
\}} \, \mathrm{d}\mathbf{x}_{k}^{(\bullet )} \, \mathrm{d}\mathbf{y}_{
\mathbf{N}(k)\setminus \{k'\}}^{(\bullet )} \, \mathrm{d}\mathbf{y}_{k'},
\end{align*}
given that
\begin{align*}
r_{\text{RJ}}((k,\mathbf{x}_{k}^{(j^{*})}),(k',\mathbf{y}_{k'})) &=
\frac{\pi (k', \mathbf{y}_{k'} \mid \mathbf{D}_{n}) \, q_{k'\mapsto k}(\mathbf{x}_{k}^{(j^{*})})}{\pi (k, \mathbf{x}_{k}^{(j^{*})} \mid \mathbf{D}_{n}) \, q_{k\mapsto k'}(\mathbf{y}_{k'})},
\cr
\bar{q}_{k \mapsto \mathbf{N}(k)\setminus \{k'\}} (\mathbf{y}_{
\mathbf{N}(k)\setminus \{k'\}}^{(\bullet )}) &= \tilde{\tilde{q}}_{k
\mapsto \mathbf{N}(k) \setminus \{k'\}}(\mathbf{y}_{\mathbf{N}(k)
\setminus \{k'\}}^{(\bullet )})
\cr
\bar{q}_{k \mapsto k'}(\mathbf{y}_{k'}^{(\bullet )}) \,
\frac{q_{k'\mapsto k}(\mathbf{x}_{k}^{(j^{*})})}{q_{k\mapsto k'}(\mathbf{y}_{k'})}
&= \tilde{q}_{k' \mapsto k} (\mathbf{x}_{k}^{(\bullet )})
\cr
\bar{\bar{q}}_{k' \mapsto \mathbf{N}(k')\setminus \{k\}}(\mathbf{z}^{(
\bullet )}_{\mathbf{N}(k')\setminus \{k\}}) &= \tilde{q}_{k' \mapsto
\mathbf{N}(k')\setminus \{k\}}(\mathbf{z}^{(\bullet )}_{\mathbf{N}(k')
\setminus \{k\}}).
\end{align*}

The analysis of the second integral in \eqref{eqn1_proof_prop} uses the same
arguments. Finally, the third integral in \eqref{eqn1_proof_prop} can be rewritten
as
\begin{equation*}
\int _{\{\mathbf{y}_{k'}\in A_{k'}\}} \pi (k', \mathbf{y}_{k'} \mid
\mathbf{D}_{n}) \int _{\{\mathbf{x}_{k}\in A_{k}\}} \delta _{(k',
\mathbf{y}_{k'})}{(k, \mathrm{d}\mathbf{x}_{k})} \, \mathbb{P}(
\text{rejection}\mid (k, \mathbf{x}_{k})) \, \mathrm{d}\mathbf{y}_{k'},
\end{equation*}
which concludes the proof.
\end{proof}

\begin{proof}[Proof of Proposition~\ref{prop_posterior_normal}]
The proof essentially relies on straightforward calculations. We have
\begin{align*}
\pi (k, \boldsymbol{\beta }_{k}, \sigma _{k} \mid \boldsymbol{\gamma }_{n})
&\propto \pi (k) \, \frac{1}{\sigma _{k}} \, \prod _{i=1}^{n}
\frac{1}{\sqrt{2 \pi } \, \sigma _{k}} \exp \left (-
\frac{1}{2 \sigma _{k}^{2}} \, (\gamma _{i} - \mathbf{c}_{i, k}^{T}
\, \boldsymbol{\beta }_{k})^{2}\right )
\cr
&\propto \pi (k) \, \frac{1}{\sigma _{k}^{n + 1}} \, \exp \left (-
\frac{1}{2 \sigma _{k}^{2}} \sum _{i=1}^{n} (\gamma _{i} - \mathbf{c}_{i,
k}^{T} \, \boldsymbol{\beta }_{k})^{2}\right ).
\end{align*}
In \cite{gagnon2018supp}, it is proved that
\begin{align*}
\sum _{i=1}^{n} (\gamma _{i} - \mathbf{c}_{i, k}^{T} \,
\boldsymbol{\beta }_{k})^{2} = (\boldsymbol{\beta }_{k} -
\widehat{\boldsymbol{\beta }}_{k})^{T} \mathbf{C}_{K}^{T} \mathbf{C}_{K} (
\boldsymbol{\beta }_{k} - \widehat{\boldsymbol{\beta }}_{k}) + \|
\boldsymbol{\gamma }_{n} - \widehat{\boldsymbol{\gamma }}_{k}\|_{2}^{2},
\end{align*}
where
$\widehat{\boldsymbol{\beta }}_{k} := (\mathbf{C}_{k}^{T} \mathbf{C}_{k})^{-1}
\mathbf{C}_{k}^{T} \boldsymbol{\gamma }_{n}$. Multiplying and dividing by
the appropriate terms yield
\begin{align*}
\pi (k, \boldsymbol{\beta }_{k}, \sigma _{k} \mid \boldsymbol{\gamma }_{n})
&\propto \pi (k) \,
\frac{\Gamma ((n - d_{k}) / 2) \, \pi ^{d_{k} / 2}}{\|\boldsymbol{\gamma }_{n} - \widehat{\boldsymbol{\gamma }}_{k}\|_{2}^{n - d_{k}} \, |\mathbf{C}_{k}^{T} \mathbf{C}_{k}|^{1/2}}
\cr
&\quad \times
\frac{2^{1-\frac{n-d_{k}}{2}} \|\boldsymbol{\gamma }_{n} - \widehat{\boldsymbol{\gamma }}_{k}\|_{2}^{n - d_{k}}}{\Gamma ((n-d_{k})/2) \, \sigma _{k}^{n-d_{k}+1}}
\, \exp \left \{  -\frac{1}{2\sigma _{k}^{2}}\, \|\boldsymbol{\gamma }_{n}
- \widehat{\boldsymbol{\gamma }}_{k}\|_{2}^{2}\right \}
\cr
&\quad \times
\frac{ |\mathbf{C}_{k}^{T} \mathbf{C}_{k}|^{1/2}}{(2 \pi )^{d_{k} / 2} \sigma _{k}^{d_{k}}}
\exp \left (-\frac{1}{2\sigma _{k}^{2}} (\boldsymbol{\beta }_{k} -
\widehat{\boldsymbol{\beta }}_{k})^{T} \mathbf{C}_{K}^{T} \mathbf{C}_{K} (
\boldsymbol{\beta }_{k} - \widehat{\boldsymbol{\beta }}_{k})\right ),
\end{align*}
which concludes the proof.
\end{proof}

\end{document}